\documentclass[rmp,aps,amssymb,preprint]{revtex4-1}
\usepackage[utf8]{inputenc}

 \usepackage{amsthm} 
 \usepackage{amsmath} 
 \usepackage{mathtools} 
\usepackage{array}
\usepackage{enumerate}
\usepackage{bbm}
\usepackage{color}
\usepackage{xcolor}
\usepackage{graphicx,lipsum,afterpage,subcaption}
\usepackage[draft]{todonotes} 
\usepackage{tikz}
\usetikzlibrary{matrix}
\usepackage{hyperref}   

  \usetikzlibrary{positioning,calc,fit,shapes.geometric,patterns}
  \pgfdeclarelayer{background}
  \pgfdeclarelayer{foreground}
  \pgfsetlayers{background,main,foreground}
  \tikzstyle{vec}=[circle,inner sep=1pt,outer sep=-1pt,fill]
  \tikzstyle{border}=[thick]
  \tikzstyle{favborder}=[border,dotted]
  \tikzstyle{exclborder}=[border,dashed]
  
\newtheorem{Theorem}{Theorem}[section]
\newtheorem{Proposition}[Theorem]{Proposition}

\newtheorem{Corollary}[Theorem]{Corollary}

\newtheorem{definition}[Theorem]{Definition}
\newtheorem{remark}[Theorem]{Remark}

\newcommand{\CH}{\mathbb{C}_h^{n\times n}} 
\newcommand{\CM}{\mathbb{C}^{n\times n}} 
\newcommand{\CD}{\mathcal{D}_h^{n\times n}} 
\newcommand{\Rho}{R}

\newcommand{\marg}{{\mathsf{marg}}}
\newcommand{\ext}{{\mathsf{ext}}}

\newcommand{\reals}{\mathbb{R}}
\newcommand{\realspace}{\mathbb{R}^n}
\newcommand{\complexs}{\mathbb{C}}

\newcommand{\gambles}{\mathcal{H}}

\newcommand{\domain}{\mathcal{K}}

\newcommand{\mdesirs}{\mathcal{M}}

\renewcommand{\epsilon}{\varepsilon}

 \definecolor{pastelgreen}{rgb}{0.47, 0.87, 0.47}
 \definecolor{lightgreen}{RGB}{222,251,152}
 \definecolor{palegreen}{rgb}{0.6, 0.98, 0.6}
 \definecolor{mossgreen}{rgb}{0.68, 0.87, 0.68}
  	\definecolor{limegreen}{rgb}{0.2, 0.8, 0.2}
  	\definecolor{inchworm}{rgb}{0.7, 0.93, 0.36}
  	\definecolor{grannysmithapple}{rgb}{0.66, 0.89, 0.63}
 
  \usepackage{framed}
\definecolor{shadecolor}{rgb}{0.70, 0.89, 0.73} 

  \newenvironment{SnugshadeF}[1][254,239,224]{
   \definecolor{shadecolor}{RGB}{#1}%
  \begin{snugshade}%
}{%
    \end{snugshade}%
}

  \newenvironment{SnugshadeB}[1][254,239,224]{
  \begin{snugshade}%
}{%
    \end{snugshade}%
}

  \newenvironment{SnugshadeE}[1][238,233,233]{
    \definecolor{shadecolor}{RGB}{#1}
  \begin{snugshade}%
}{%
    \end{snugshade}%
}

\newcounter{example}
\newenvironment{example}[1][]{\refstepcounter{example}\par\noindent\textit{~\\Example~\theexample. #1}\small }
{\par}

\newcounter{MethodF}
\newenvironment{MethodF}[1][]{\refstepcounter{MethodF}\par\noindent\textbf{#1}\begin{upshape}}
{\end{upshape} \par}

  \newcounter{MethodB}
\newenvironment{MethodB}[1][]{\refstepcounter{MethodB}\par\noindent\textbf{#1}\begin{upshape}}
{\end{upshape} \par}

\begin{document}

\title{Quantum mechanics: The Bayesian theory generalised to the space of
    Hermitian matrices}

\author{Alessio Benavoli}
\email[]{alessio@idsia.ch}
\author{Alessandro Facchini}
\email[]{alessandro.facchini@idsia.ch}
\author{Marco Zaffalon}
\email[]{zaffalon@idsia.ch}
\affiliation{Istituto Dalle Molle di Studi sull'Intelligenza Artificiale (IDSIA). \\ Lugano (Switzerland)}

\date{\today}

\begin{abstract}
We consider the problem of gambling on a quantum experiment and enforce rational behaviour by a few rules. These rules yield, in the classical case, the Bayesian theory of probability via duality theorems. In our quantum setting, they yield the Bayesian theory generalised to the space of Hermitian matrices. This very theory is quantum mechanics: in fact, we derive all its four postulates from the generalised Bayesian theory. This implies that quantum mechanics is self-consistent. It also leads us to reinterpret the main operations in quantum mechanics as probability rules: Bayes' rule (measurement), marginalisation (partial tracing), independence (tensor product).
To say it with a slogan, we obtain that quantum mechanics is the Bayesian theory in the complex numbers.

\end{abstract}


\maketitle


\section{Introduction}
Quantum mechanics (QM) is based on four main axioms -- or postulates. These axioms were derived after a long process of trial and error, which involved a considerable amount of guessing by the originators of the theory. The motivation for the axioms is not always clear and even to experts the basic axioms of 
QM often appear counter-intuitive \cite{nielsen2010quantum}.

The aim of this paper is twofold:
\begin{itemize}
\item to derive quantum mechanics  from a single principle of self-consistency;
\item to show that QM is just the Bayesian theory generalised to the complex Hilbert space.
\end{itemize}

To this end, we will base our analysis on a gambling system, where a subject, whom we call Alice, has only to specify which gambles on a certain quantum experiment she is willing to accept. We require a few simple rules of self-consistency on Alice's assessments. If the rules are satisfied, then we say that the assessment are coherent and that Alice is rational. These coherence rules have a long history in the foundational approaches to traditional probability as well as decision theory; imposing them in the classical case allows one, for instance, to derive a very general and powerful version of the Bayesian theory \cite{smith1961,williams1975, walley1991}.

And so essentially do we in this paper: by following the ideas of the classical case while extending them to the quantum one, we derive a theory of probability apt to work in the complex Hilbert space. We show that this very theory is QM. It has probabilities (density matrices), Bayes' rule (measurement), marginalisation (partial tracing), independence (tensor product). Such a formulation of QM is an actual generalisation of the Bayesian theory, which we recover as a special case. But it obviously exhibits also peculiar behaviours that are impossible to obtain in the classical case, such as a strong form of stochastic dependence: entanglement -- and in fact the QM violation of Bell's inequalities is tantamount, in our formulation, to the violation of the probabilistic notion of Fr\'echet bounds. 

The idea of justifying QM from rationality principles on a gambling system is not new. It has been described in pioneering work about the Bayesian interpretation of QM (QBism) \cite{Caves02,Fuchs01,Schack01,Fuchs02,Fuchs03,Schack04,Fuchs04,Caves07,Appleby05a,Appleby05b,Timpson08,longPaper, 
Fuchs&SchackII,mermin2014physics} and  Pitowsky's quantum gambles \cite{pitowsky2003betting,Pitowsky2006}. QBism and Pitowsky aim  at  proving the coherence of an $n$-dimensional quantum system by considering
the probability assignments derived from QM (via Born's rule or, more in general, Gleason's theorem). 

Our approach is different. We do not aim at showing that the probabilities \emph{derived} from QM are Bayesian; we aim at showing that \emph{QM itself is the Bayesian theory}, once it is generalised to work in the complex  space of hermitian matrices, $\mathbb{C}_h^{n\times n}$ -- this also shows in a definite sense that QM is strongly self-consistent. We obtain this result by bringing to light a duality relation between sets of gambles defined on $\mathbb{C}_h^{n\times n}$ and density matrices, in the same way as there is a duality relation between sets of real-valued gambles and probabilities in the classical case.

\subsection{Outline of the paper}

Section~\ref{sec:classicDes} reviews the foundations of de Finetti's Bayesian theory from the point of view of the theory of coherent sets of desirable gambles. It discusses the inherent generalisation of this theory that allows it to deal smoothly also with sets of probabilities (Bayesian robustness/imprecise probability).

Section~\ref{sec:intro-clp} develops the foundations of the desirability theory in the quantum case.

Section~\ref{sec:1staxiom} works out the duality relation between gambles and ``probabilities'' in the quantum case. On this basis it derives from desirability the first axiom of quantum mechanics. It also derives Born's rule as a direct outcome of the notion of expectation in the theory of desirability and discusses the relation to Gleason's theorem.

Section~\ref{sec:2ndaxiom} derives the second axiom of QM from the desirability definition of conditioning, and its implications when it is used as a rule to compute future beliefs.

Section~\ref{sec:3ndaxiom} derives the third axiom of QM from considerations of temporal coherence: that is, the fact that the desirability theory should be self-consistent also through time.

Section~\ref{sec:4ndaxiom} derives the fourth axiom of quantum mechanics as a consequence of the definition of ``stochastic'' independence in desirability. It also shows that entanglement is a strong form of ``stochastic'' independence that is specific to quantum desirability, in that it cannot be achieved in the classical case; it discusses how this leads to the violation of Fr\'echet bounds in the quantum case, that is, Bell's inequalities.

Section \ref{sec:qbysm} presents a discussion of QBism and a comparison with the approach presented in this paper.

Section~\ref{sec:conc} concludes the paper. The proofs of all the main results are reported in the Appendix.


\section{Probability and desirability}\label{sec:classicDes}

\subsection{Axiomatisation  of probability}
Modern probability theory was axiomatised  by  \citet{kolmogorov1950foundations}. 
By the Kolmogorov axiomatisation, probability is defined on a 
probability space, which is a triplet $(\Omega,\mathcal{F},P)$, where $\Omega$ is an arbitrary set (elements  $\omega$  
of  $\Omega$ are said to be elementary events), $\mathcal{F}$ is an arbitrary $\sigma$-algebra of subsets of $\Omega$ 
(elements of $\mathcal{F}$ are said to be events)  and $P$ is a measure that yields 
values in the segment $[0,1]$  of the real line (\textit{first axiom}),  normalised by the condition $P(\Omega)=1$ (\textit{second axiom}) {and satisfying the $\sigma$-additivity condition $P(\cup_{i=1}^{\infty} A_i)=\sum_{i=1}^{\infty }P(A_i)$ for each sequence ($A_i: ~i=1,2,3,\dots$)  of pairwise disjoint members  of  $\mathcal{F}$ (\textit{third 
axiom})}. 

Kolmogorov theory also contains the additional axiomatic definition of conditional
probabilities. By definition, the conditional probability of $B$ given the event $A$ is defined as follows  
(\textit{fourth axiom}):
$$
P(B|A)=\frac{P(A\cap B)}{P(A)} \text{ with } P(A)>0.
$$ 

From these axioms, it is also possible to derive  marginalisation, law of total probability, and so on.

To derive this formulation, Kolgomorov exploited  Borel's measure-theoretic approach to probability 
(hence $\sigma$-additivity). The axioms were instead motivated by Kolmogorov from ideas by von-Mises  about the 
frequency definition of  probability. He used frequency reasons to take the segment $[0,1]$, define the normalisation 
$P(\Omega)=1$ and to 
motivate  additivity  of probability and the  definition of 
conditioning \cite{kolmogorov1950foundations}.

\subsection{Subjective probability and desirability}
\Citet{finetti1937} established a different foundation of probability theory on the notion of ``coherence'' (self-consistency). This means, roughly speaking, that one should assign and manipulate 
probabilities so as not to possibly be made a \textit{sure loser} in a gambling system based on them. 
De Finetti used this approach to give a subjectivistic foundation to the theory of probability,\footnote{{De Finetti actually considered only finitely additive probabilities.
Since in the paper we always assume $\Omega$ to be finite, $\sigma$-additivity and finite additivity coincide.}} by showing that all  probability 
axioms could be derived  by imposing the principle of self-consistency alone on a subject's odds about an uncertain experiment. 
The subject is considered rational if she chooses her odds so that there is no bet that leads her to a sure loss (no Dutch books are possible).
In this way, since mathematically odds are the inverse of probabilities,  de Finetti provided a justification of Kolmogorov's axioms as a rationality criterion
on a gambling system.

\citet{williams1975} and then \citet{walley1991} have shown that it is possible to justify probability in a way that is even simpler, more general and elegant. 
To understand this betting framework, we consider an experiment whose outcome  $\omega$ belongs to a certain  space of possibilities $\Omega$ (e.g., the experiment may be tossing a coin or throwing a dice). We can model Alice's beliefs about $\omega$ by asking her whether she accepts to engage
in certain risky transactions, called \textit{gambles}, whose outcome depends on the actual
outcome of the experiment. Mathematically, a gamble is a bounded real-valued function on $\Omega$, $g:\Omega 
\rightarrow \mathbb{R}$, which is interpreted as an uncertain reward in a linear utility scale. If Alice accepts a gamble $g$, this means that she commits herself to 
receive $g(\omega)$ \emph{utiles}\footnote{Abstract units of utility, indicating the satisfaction derived from an economic transaction; we can approximately identify it with money provided we deal with small amounts of it \cite[Sec.~3.2.5.]{finetti1974}.} if the experiment is performed and the outcome of the experiment eventually happens 
to be the event $\omega \in \Omega$. Since $g(\omega)$ can be negative, Alice can also lose utiles. Therefore Alice's acceptability of a gamble depends on her knowledge about the experiment.

The  set $\domain$ of gambles that Alice accepts is called her set of \emph{desirable} (or \emph{acceptable}) \emph{gambles}. 
One such set is said to be \emph{coherent} when it satisfies a few simple rationality criteria.
These criteria are very easy to understand; we introduce them with an example.

\begin{SnugshadeE}
\begin{example}
\label{ex:coin0}
 Let us assume that the uncertain experiment  is about the outcome of the toss of a fair coin:
 $\Omega=\{Head,Tail\}$.
 Before the experiment, the Bookie asks Alice what are the gambles that she is willing to accept.
 A gamble $g$ in this case has two components $g(Head)=g_1$ and $g(Tail)=g_2$.
 If Alice accepts $g$ then  she commits herself to receive/pay $g_1$ if  the outcome is Head and   $g_2$ if Tail.
Since a gamble is in this case an element of $\mathbb{R}^2$,  $g=(g_1,g_2)$, we can plot the gambles  Alice accepts  in a two-dimensional Cartesian coordinate system with coordinate axes $g_1$ and $g_2$. 
 
 First, the Bookie asks Alice if she is willing to accept $g=(1,1)$  -- this gamble means that she receives $1$ utile, no matter the result of the experiment. Alice accepts this gamble,
given that it increases her wealth without ever decreasing it. Similarly, she accepts $h=(1,0)$, $f=(0,1)$, any positive scaling of these gambles
$\nu f,\gamma h$ (with $\nu,\gamma>0$) and their sum  $ f+ h$.
This means that she is also accepting  any positive combination $ \nu f+ \gamma h$, since
the resulting vector is always non-negative.
Summing up, Alice always accepts the first quadrant, Figure~\ref{fig:coin0}(a).
We have excluded the gamble~$(0,0)$, because we will  only consider gambles that are strictly desirable for Alice, that is, gambles for which  she always gains something (even an epsilon):
$$
\domain_1=\{g \in \mathbb{R}^2 \mid g\neq 0, ~g_i \geq0 \}.
$$

Then, the Bookie asks Alice if she is willing to accept $g=(-1,-1)$. Since she loses $1$ utile no matter the outcome of the experiment,  Alice does not accept this gamble, nor any gamble $\nu f+\gamma h$ with  $h=(-1,0)$, $f=(0,-1)$ and $\nu,\gamma>0$. In other words, Alice always rejects the third quadrant, Figure~\ref{fig:coin0}(b).

Then  the Bookie asks Alice about $g=(-0.1,1)$ -- she loses $0.1$ if Head and wins $1$  if Tail.
Since Alice knows that the coin is fair, she accepts this gamble as well as
all the gambles of the form  $\nu g$ with  $\nu>0$, because this is just a change of scale.
Similarly, she accepts all the gambles $\nu g + \gamma h$ for any $\nu,\gamma>0$ 
and $h \in \domain_1$, since these gambles are even more favourable for her.
Now, the Bookie asks Alice about $g=(1,-0.1)$ and the argument is symmetric to the above case.
We therefore obtain the following set of desirable gambles (see Figure~\ref{fig:coin0}(c)):
$$
\domain_2=\{g \in \mathbb{R}^2 \mid 10g_1+g_2\geq 0 \text{ and } g_1+10g_2\geq 0\}.
$$

Finally, the Bookie asks Alice about $g=(-1,1)$ -- she looses $1$ if Head and wins $1$  if Tail.
Alice could accept this gamble because she knows that the coin is fair.
However, since we are only considering gambles that are strictly desirable for Alice, she only accepts
gambles of the form  $g=(-1,1)+\epsilon$ with $\epsilon>0$ and, by rationality,
all gambles $\nu g$ with $\nu>0$ and $\nu g + \gamma h$ for any $\nu,\gamma>0$ and $h \in \domain_2$.
A similar conclusion can be derived for the symmetric gamble $g=(1,-1)$.
Figure~\ref{fig:coin0}(d) is her final set of desirable gambles about the experiment concerned with the toss of a fair coin, which in a formula becomes
$$
\domain_3=\{g \in \mathbb{R}^2 \mid g_1+g_2> 0\}.
$$

Alice does not accept any other gamble. In fact, if Alice  would also accept for instance $h=(-2,0.5)$ then, 
since she has also accepted $g=(1+\epsilon,-1)$, i.e., $g\in \domain_3$, she must also accept $g+h$ (because this gamble should also be favourable to her).
However, $g+h=(-1+\epsilon,-0.5)$ is always negative, Alice always loses utiles in this case. In other words, by accepting $h=(-2,0.5)$ Alice incurs a sure loss -- she is irrational.

\centering
\begin{minipage}{0.23\textwidth}
\begin{tikzpicture}[scale=1.5]
    \draw[->] (-1,0) coordinate (xl) -- (1,0) coordinate (xu) node[right] {$g_1$};
    \draw[->] (0,-1) coordinate (yl) -- (0,1) coordinate (yu) node[above] {$g_2$};
        \draw (0,0) circle (1.pt);
    \begin{pgfonlayer}{background}
      \draw[border] (0,0) -- (0,1) coordinate (a1away);
      \draw[border] (0,0) -- (1,0) coordinate (a2away);
      \fill[blue!50,  opacity=0.5] (0,0) -- (1,0) -| (1,1) --(0,1);
    \end{pgfonlayer}
  \end{tikzpicture}
   \captionof{subfigure}{}
  \end{minipage}
  \begin{minipage}{0.23\textwidth}
   \begin{tikzpicture}[scale=1.5]
    \draw[->] (-1,0) coordinate (xl) -- (1,0) coordinate (xu) node[right] {$g_1$};
    \draw[->] (0,-1) coordinate (yl) -- (0,1) coordinate (yu) node[above] {$g_2$};
    \draw (0,0) circle (1.pt);
        \begin{pgfonlayer}{background}
      \draw[border] (0,0) -- (0,1) coordinate (a1away);
      \draw[border] (0,0) -- (1,0) coordinate (a2away);
      \fill[blue!50,  opacity=0.5] (0,0) -- (1,0) -| (1,1) --(0,1);
    \end{pgfonlayer}
    \begin{pgfonlayer}{background}
      \fill[pattern=north west lines, pattern color=black] (0,0) -- (-1,0) -| (-1,-1) --(0,-1);
    \end{pgfonlayer}
  \end{tikzpicture}
     \captionof{subfigure}{}
  \end{minipage}
  \begin{minipage}{0.23\textwidth}
        \begin{tikzpicture}[scale=1.5]
    \draw[->] (-1,0) coordinate (xl) -- (1,0) coordinate (xu) node[above] {$g_1$};
    \draw[->] (0,-1) coordinate (yl) -- (0,1) coordinate (yu) node[above] {$g_2$};
    \draw (0,0) circle (1.pt);
    \begin{pgfonlayer}{background}
      \draw[border] (0,0) -- (0,-1) coordinate (a1away);
      \draw[border] (0,0) -- (-1,0) coordinate (a2away);
      \fill[pattern=north west lines, pattern color=black] (0,0) -- (-1,0) -| (-1,-1) --(0,-1);
    \end{pgfonlayer}
    \begin{pgfonlayer}{background}
      \draw[-] (0,0) -- (-0.1,1) coordinate (a1away);
      \draw[-] (0,0) -- (1,-0.1) coordinate (a2away);
      \fill[blue!50,  opacity=0.5] (0,0) -- (a1away) -| (a2away) --(0,0);
    \end{pgfonlayer}
  \end{tikzpicture}
     \captionof{subfigure}{}
  \end{minipage}  
  \begin{minipage}{0.23\textwidth}
        \begin{tikzpicture}[scale=1.5]
    \draw[->] (-1,0) coordinate (xl) -- (1,0) coordinate (xu) node[above] {$g_1$};
    \draw[->] (0,-1) coordinate (yl) -- (0,1) coordinate (yu) node[above] {$g_2$};
    \draw (0,0) circle (1.pt);
    \begin{pgfonlayer}{background}
      \draw[dashed] (0,0) -- (-1,1) coordinate (a1away);
      \draw[dashed] (0,0) -- (1,-1) coordinate (a2away);
      \fill[blue!50,  opacity=0.5] (0,0) -- (a1away) -| (a2away) --(0,0);
    \end{pgfonlayer}
  \end{tikzpicture}
     \captionof{subfigure}{}
  \end{minipage}
\captionof{figure}{\centering Alices' sets of coherent strictly desirable gambles for the experiment of tossing a fair coin.  \label{fig:coin0}}
  
\end{example}
\end{SnugshadeE}

Summing up, the rationality criteria are:
\begin{enumerate}
 \item Any gamble $g\neq0 $ such that $g(\omega)\geq0$ for each $\omega \in \Omega$ must be desirable for Alice, given 
that it may increase Alice's capital without ever decreasing it
 (\textbf{accepting partial gain)}. 
 \item Any gamble $g$ such that $g(\omega)\leq0$ for each $\omega \in \Omega$ must not be desirable for Alice, given 
that it may only decrease Alice's capital without ever increasing it  (\textbf{avoiding partial loss}). 
 \item If Alice finds $g$ to be desirable, that is,
$g \in \domain$, then also $\lambda g$ must be desirable for her for any $0<\lambda \in \mathbb{R}$ (\textbf{positive homogeneity}).
We can think of this rule simply as an invariance to a change of currency.
\item If Alice finds $g_1$ and $g_2$ desirable, that is,
$g_1,g_2 \in \domain$, then she also must accept $g_1+g_2$, i.e., $g_1+g_2 \in \domain$ (\textbf{additivity}). 
\end{enumerate}
If the set of desirable gamble $\domain$ satisfies these property we say that it is \emph{coherent}, or, equivalently, that Alice is rational. 

Note how by these four axioms we express some truly minimal requirements: the first means that Alice likes to increase her wealth; the second that she does not like to decease it; the third and fourth together simply rephrase the assumption that Alice's utility scale is linear. In spite of the simple character of these requirements, these four axioms alone define a very general theory of probability, of which, for instance, the Bayesian theory of probability is a special case \cite{williams1975,walley1991}. In the next sections we will in fact show how we can derive probability from desirability. This points also to the fact that the kind of coherence embodied by the four axioms above is more primitive that the usual Dutch-book coherence; it does subsume it but it is more fundamental. 

Apart from allowing us to derive the traditional axioms and theory of probability, 
the above desirability setting allows us to easily define operations such as marginalisation, 
conditioning, independence, integration, etc., all at the level of gambles
 \cite{miranda2008a,Couso20111034,Vicig2000235,zaffalon2010e,zaffalon2011a}.
 We point the reader to \citet{augustin2014} for a recent survey. Such operations become the 
common operations we are used to when we focus on the probabilistic statements we derive from desirability.

The theory of desirable gambles cannot be used directly for QM, because of the presence of complex numbers. In the next section, we verify  that we can extend such a theory to the space 
of complex $n \times n$ matrices ($n$-dimensional quantum system). In doing so, we will also extend and employ the operations mentioned above, as taken from the original theory of desirability, into the generalised version needed by QM.


\section{Quantum Desirability}
\label{sec:intro-clp}
Our aim in the next sections is to lay the foundations of the theory of desirability in the case of QM. We call it \emph{quantum desirability}.
We will show that quantum desirability includes the theory of desirable gambles presented in the previous section as a particular case. 
We will exploit this fact to present at the same time examples of  quantum desirability and classical desirability (e.g., quantum vs.\ classical coins). This is done in order to help the reader, who is not familiar with desirability, to understand the connection between classical desirability and probability as well as quantum desirability and quantum mechanics.

\subsection{Foundations of desirability}\label{sec:foundations}
The first step for defining quantum desirability is to specify what is the experiment,
what is a gamble on this experiment and how the payoff for the gamble is computed.
{
To this end, we consider an experiment relative to an $n$-dimensional  quantum system and two subjects: the gambler (Alice) and the bookmaker.
The $n$-dimensional quantum system is prepared by the bookmaker in some quantum state.
We assume that Alice has her personal knowledge about the experiment (possibly no knowledge at all).}

\begin{enumerate}
\item   {
The bookmaker 
 announces that he will measure the quantum system along  its $n$ orthogonal directions and so the outcome of the measurement is an element of $\Omega=\{\omega_1,\dots,\omega_n\}$,  with $\omega_i$ denoting the elementary event ``detection along $i$''. 
Mathematically,  it means that the quantum system is measured along its eigenvectors,\footnote{We mean the eigenvectors of the density matrix of the quantum 
system.} i.e., the projectors\footnote{A projector $\Pi$ is a set of $n$ positive semi-definite matrices in $\CH$ such that   $\Pi_i\Pi_k=0$, $(\Pi_i)^2=\Pi_i=(\Pi_i)^\dagger$,  $\sum_{i=1}^n \Pi_i=I$.} $\Pi^*=\{\Pi^*_{1},\dots,\Pi^*_{n}\}$
and $\omega_i$ is the event ``indicated'' by the $i$-th projector. The bookmaker is fair, meaning that he will correctly perform the experiment and report
 the actual results to Alice.} 
  \item {Before the experiment, Alice declares the set of gambles she is willing to accept.  Mathematically, a gamble $G$ on this experiment 
is a Hermitian matrix in $\CM$, the space of all Hermitian  $n \times n$ matrices being denoted by $\CH$.  We will denote the set of gambles Alice is willing to accept by $\domain \subseteq \CH$.}
\item By accepting  a gamble $G$, Alice commits herself to receive  $\gamma_{i}\in \reals$ utiles  if the outcome of the experiment eventually happens to be 
$\omega_i$. The value $\gamma_{i}$ is defined from $G$ and $\Pi^{*}$ as follows:
 \begin{equation}
  \Pi^{*}_{i}G\Pi^{*}_{i}=\gamma_{i}\Pi^{*}_{i} \text{ for } i=1,\dots,n.
 \end{equation} 
It is a real number since $G$ is Hermitian.
\end{enumerate}
This gambling system has been partially inspired by that of \citet{pitowsky2003betting}.\footnote{The important difference is that we define the gambles and payoffs directly in $\CH$.}

We recall that, by accepting a gamble $G$, Alice commits herself to receive $\gamma_{i}$ whatever event,  indicated by  $\Pi^{*}_i$, occurs. Since $\gamma_{i}$ can be negative, Alice can lose utiles and hence the  desirability of a gamble depends on Alice's beliefs about the quantum experiment.
{Note that in the paper we use the starred notation, $\Pi^*=\{\Pi_i^*\}_{i=1}^n$, to denote the $n$ orthogonal directions of the quantum state  prepared by the bookmaker. Conversely, a generic set of projectors is denoted without star, $\Pi=\{\Pi_i\}_{i=1}^n$.}

\begin{SnugshadeE}
 \begin{example}[Classical coin.]
\label{ex:coinA}
Let us consider again the classical coin tossing experiment. The possible outcomes are $\Omega=\{Head,Tail\}$ and they can be associated 
one-to-one
 to the canonical basis on $\mathbb{R}^2$, $\Omega=\{e_1,e_2\}$. The  measurement $\Pi^*$  here can be represented 
by the projectors
  $$
  \Pi^*_1=e_1e_1^T=\left[\begin{array}{cc}
         1 &0 \\ 0 &0 \\
        \end{array}\right], ~~\Pi^*_2=e_2e_2^T=\left[\begin{array}{cc}
         0 &0 \\ 0 &1 \\
        \end{array}\right],
  $$
  and therefore they are completely known by Alice.
  Since  $\Pi^*_iG\Pi^*_i$ depends only on the diagonal elements of $G$, without loss of generality Alice can 
restrict herself to only consider gambles $G$ that are diagonal Hermitian matrices.
  Since the diagonal elements of a diagonal Hermitian matrix are real numbers, the gambles are actually vectors 
$g=(g_1,g_2)$ in $\mathbb{R}^2$, i.e.,
  $G=diag(g)$.  Therefore, the experiment payoffs are:
  $$
   \Pi^*_1G\Pi^*_1=g_1\Pi^*_1, ~~~\Pi^*_2G\Pi^*_2=g_2\Pi^*_2,
  $$
i.e.,  Alice receives $g_1$ if the result of the toss is \textit{Head} and $g_2$ if it is \textit{Tail}.
  For instance,  the gamble $g=(1,1)$ means that no matter the result of the experiment Alice will receive $1$ 
utile, while the gamble   $h=(1,-2)$ means that she will receive $1$ utile if the outcome is \textit{Head} and she will lose $2$ utiles if \textit{Tail}. 
 If Alice accepts the gamble $h$, this means that she believes that \textit{Head} is more probable than \textit{Tail}.
Classical desirability can be regarded to be only about gambles that are diagonal Hermitian matrices.
 \end{example}
\end{SnugshadeE}

 \begin{SnugshadeE}
 \begin{example}[Quantum coin.]
\label{ex:spinA}
Let us consider now a two-level quantum system (qubit) such as for instance the vertical or horizontal polarisation of a photon 
$\Omega=\{V,H\}$.
The gambles $G$ are in this case two-dimensional Hermitian matrices ($n=2$). If Alice accepts the gamble $G$ and 
the experiment is performed, she receives $\gamma_1$ utiles
if the photon is detected along the direction indicated by $\Pi^*_1$ and  $\gamma_2$ if the photon is 
detected along the direction $\Pi_2^{*}$. 
For instance,  the gamble $G_1=I=[1,0;0,1]$ should be desirable for Alice because no matter
$\Pi^*$  and the outcome of the experiment $\Pi^{*}_i$, she will always receive $1$ utile, since
$\Pi^{*}_iI\Pi^{*}_i=\Pi^{*}_i$. 

Consider now the gamble $G_2=[1,-\iota;\iota,-2]$, with $\iota$ being the imaginary unit.
 Alice knows she will be rewarded by the bookie as follows:
 $$
 \Pi^{*}_1G_2\Pi^{*}_1=\gamma_1\Pi^{*}_1, ~~~\Pi^{*}_2G_2\Pi^{*}_2=\gamma_2\Pi^{*}_2,
 $$
 if $ \Pi^{*}$ is performed and the outcome is  $\Pi^{*}_1$ or, respectively, $\Pi^{*}_2$.
 {Alice can decide to accept $G_2$ based on her beliefs on the quantum system.
For instance, if she knows that  the directions of the states} of the 
quantum system are $\Pi^{*}_1=[1,-\iota;\iota,1]/2$ and $\Pi^{*}_2=[1,\iota;-\iota,1]/2$,
 then by accepting the gamble $G_2$ she   commits herself  to receive 
 $0.25$ utile if  $\Pi^{*}_1$  happens and to lose $0.75$ utiles  if $\Pi^{*}_2$ happens.
 \end{example}
 \end{SnugshadeE}
 
We now recall some well-known results from linear algebra. Each  Hermitian matrix  $G$ in $\CH$ can be decomposed 
 as  $G={U}{\Lambda}{U}^{\dagger}$, where $U$ is an $n \times n$ unitary complex matrix, i.e., $U{U}^{\dagger}={U}^{\dagger}U=I$, and $\Lambda$ is diagonal (eigenvalues matrix) with $n$ real-valued diagonal elements.  A matrix $G \in \CH$ is said to be:
\begin{description}
 \item[PSD]   Positive Semi-Definite, if all its eigenvalues are  positive (denoted as $G\geq 0$);
 \item[PSDNZ]  PSD and non-zero,  it is PSD and $G \neq 0$ (denoted as $G\gneq 0$);
  \item[PD]    Positive Definite, if all its eigenvalues are  strictly positive (denoted as $G>0$). 
 \end{description}
Similar definitions hold for negative semi-definite, negative semi-definite and non-zero, and negative definite ones.
A known result  for PSD and PSDNZ matrices is the following. Let $G \in \CH$, it holds that 
 $G\geq 0$  if and only if   $CGC^\dagger \geq0$  for any matrix  $C \in \CM$. 
 Moreover, if $G\gneq 0$ and $\Pi=\{\Pi_1,\dots,\Pi_n\}$ is any projection measurement, then there 
is $j\in\{1,\dots,n\}$ such that $ \Pi_jG\Pi_j \gneq0$.
For the convenience of the reader, the proof of this result is given  in Proposition \ref{pro:CGC>0} in the Appendix.
 
 Let us now go back to our experiment. Denote by $\gambles^+=\{G\in\CH:G\gneq0\}$ the
subset of all PSDNZ matrices  in $\CH$: the set of \emph{positive gambles}.
The set of negative gambles is similarly given by $\gambles^-=\{G\in\CH:G\lneq0\}$.  Alice examines the gambles in $\CH$ and comes up with the subset $\domain$ of
the gambles that she finds desirable. Alice's rationality is characterised as follows.

  \begin{SnugshadeB}
\begin{MethodB}[Alice is rational] if  the following conditions hold:
\begin{enumerate}
 \item Any gamble $G \in \CH$ such that $G \gneq0$ must be desirable for Alice, given that it may increase Alice's 
utiles without ever decreasing them
 (\textbf{accepting partial gain}). This means that $ \gambles^+ \subseteq \domain$.
\item Any gamble $G \in \CH$ such that $G \lneq0$ must not be desirable for Alice, given that it may only decrease 
Alice's utiles without ever increasing them  (\textbf{avoiding partial loss}). This means that $ \gambles^- \cap 
\domain=\emptyset$.
\item If Alice finds $G$ desirable, that is
$G \in \domain$, then also $\nu G$ must be desirable for her for any $0<\nu \in \reals$ (\textbf{positive homogeneity}).
\item If Alice finds $G_1$ and $G_2$ desirable, that is
$G_1,G_2 \in \domain$, then she also must accept $G_1+G_2$, i.e., $G_1+G_2 \in \domain$ (\textbf{additivity}). 
\end{enumerate}
\end{MethodB}
\end{SnugshadeB}
To understand these rationality criteria, we must remember that mathematically the payoff for any gamble $G$
is computed as $\Pi_i^{*} G \Pi_i^{*}$ if the outcome of the experiment is the event indicated by $\Pi_i^{*}$.
Then the first two rationality criteria above hold no matter the experiment $\Pi^{*}$ that 
is eventually performed. 
In fact,  from the properties of PSDNZ matrices discussed before (see also Proposition~\ref{pro:CGC>0}), if 
 $G \gneq0$ then  $\Pi_i^{*} G \Pi_i^{*}=\gamma_{i} \Pi_i^{*}$ with $\gamma^{*}_{i}\geq0$ for any $i$ and 
$\gamma_{j}>0$ for some $j$. Therefore, by accepting $G \gneq0$, Alice can only increase her utiles.
Symmetrically,  if $G \lneq0$ then $\Pi_i^{*} G \Pi_i^{*} = \gamma_{i} 
\Pi_i^{*}$ with $\gamma_{i}\leq 0$ for any $i$. 
Therefore, Alice must avoid the gambles $G \lneq0$ because they can only decrease her utiles.
This justifies  the first two rationality criteria. 
 For the last two, observe that 
 $$
 \Pi_i^{*} (G_1+G_2) \Pi_i^{*}=\Pi_i^{*} G_1 \Pi_i^{*}+\Pi_i^{*} G_2 \Pi_i^{*}=(\gamma_i+\vartheta_i) \Pi_i^{*},
 $$ 
 where we have  exploited the fact that $\Pi_i^{*} G_1 \Pi_i^{*}=\gamma_i  \Pi_i^{*}$ and $\Pi_i^{*} G_2 
\Pi_i^{*}=\vartheta_i \Pi_i^{*}$. Hence, if Alice accepts $G_1,G_2$, she must also accept $G_1+G_2$ because this 
will lead to a reward of $\gamma_i+\vartheta_i$.
 Similarly, if $G$ is desirable for Alice, then also $\Pi_i^{*} (\nu G) \Pi_i^{*}=
 \nu\Pi_i^{*}  G \Pi_i^{*}$ should be desirable for any $\nu>0$. 
 
In other words, as in the case of classical desirability, the four conditions above state only minimal requirements: that Alice would like to increase her wealth and not decrease it (conditions $1$ and $2$); and that Alice's utility scale is linear (conditions $3$ and $4$). The first two conditions should be plainly uncontroversial. The linearity of the utility scale is routinely assumed in the theories of personal, and in particular Bayesian, probability as a way to isolate  considerations of uncertainty from those of value (wealth).

We can characterise $\domain$ also from a geometric point of view. In fact, from the above properties, it follows that a coherent set of desirable gambles $\domain$ is a convex cone\footnote{A subset $\domain$ of a vector space is called 
{\bf convex} if it is closed under convex combinations:
$G,F\in \domain$ implies $wG+(1-w)F\in \domain$, for any $w\in[0,1]$. It is called
 a {\bf cone} if it is closed under non-zero, positive scalar multiplication: $
G\in \domain$ implies $ \nu G\in \domain$, for any $\nu>0$.
It can be verified that a cone $\domain$ is convex if and only if  it is closed under non-zero, positive linear combinations,
 i.e., iff $\alpha G_1 + \beta G_2$ belongs to $\domain$ for any positive scalars $\alpha,\beta>0$ and $G_1,G_2 \in \domain$.
}
 (positive scaling and additivity) in $\CH$ that includes $ \gambles^+$ (accepting partial gains) and excludes $ \gambles^-$ (avoiding partial losses). 
Without loss of generality we can also assume that $\domain$ is not pointed, i.e., $0\notin \domain$: Alice does not accept the null gamble.\footnote{We will show that this allows us to derive the avoiding partial loss criterion from the accepting partial gain one.}
A coherent set of desirable gambles is therefore a \textit{non-pointed convex cone}. 
 
From now onwards, we assume that Alice's set of desirable gambles mathematically satisfies the following properties.
 
\begin{definition}
  \label{def:sdg}
Let $\domain $ be a subset of $\CH$. We say that  $\domain$ is   a {\bf coherent set of strictly desirable gambles (SDG)} if
\begin{description}
 \item[(S1)] $\domain$ is a non-pointed convex-cone (positive homogeneity and additivity);
 \item[(S2)] if $G\gneq0$ then $G \in \domain$ (accepting partial gain);
   \item[(S3)] if $G \in \domain$ then either $G \gneq0$ or $G -\Delta \in  \domain$ for some $0<\Delta \in \CH$ (openness).
\end{description}
\end{definition}

\noindent 
The additional openness property (S3) in Definition \ref{def:sdg} is not necessary for rationality, but it is technically convenient in this paper as it precisely isolates the kind of models we will use and hence it makes the derivations simpler. Geometrically, it amounts to excluding from the set all the gambles that are on the border of the cone, apart from the gambles $G \gneq0$ that are always desirable (thus, by an abuse of terminology, strictly desirable gambles are often referred to as open convex cones). Property~(S3) can be regarded as an \emph{Archimedean} property for desirable gambles. Indeed, when instead of $\CH$ we consider the space $\realspace$, loosely speaking, it implies that a set of strictly desirable gambles (seen as a non-pointed open convex cone of $\realspace$)  has a corresponding (dual) probabilistic representation that is equivalent to the set. In Section \ref{sec:1staxiom} we will show that, \emph{mutatis mutandis}, the same happens in $\CH$: a set of strictly desirable gambles has a corresponding (dual)  representation in term of density matrices that is equivalent to the set.

 {As explained in Example~\ref{ex:coin0}, property (S3) can also be given a  gambling interpretation: it means that we will only consider gambles that are \emph{strictly} desirable for Alice; these are gambles that are either  positive (PSDNZ) or for which Alice is willing to pay a positive amount to have them (an epsilon).
 We will discuss further on the meaning of buying a gamble in Section \ref{sec:prevision}.}

\begin{remark}
The assumption that $0\notin \domain$ implies that SDG also satisfies:
\begin{itemize}
 \item  if $G \lneq 0$, then $G \notin \domain$ (avoiding partial loss).
\end{itemize}
Indeed, assume $G \lneq 0$ is in  $\domain$. Notice that $ -G \gneq 0$, and because of the acceptance of partial gain, $- G \in \domain$. 
Thus, by additivity we have that $G-G=0   \in \domain$, a contradiction. This means that a non-pointed cone that includes $\gambles^+$ automatically avoids partial loss.
\end{remark}
It is clear from the definition that the minimal SDG is the convex cone $\domain$  that includes all  $G\gneq0$ and nothing else, i.e., $\gambles^+$. It characterises a state of full ignorance.
Conversely, an SDG is called \emph{maximal} if there is no other SDG including it. 
In terms of rationality, a maximal SDG is a set of gambles that Alice cannot extend to another SDG by accepting other gambles while keeping at the same time rationality.
 It also represents a situation in which Alice is sure about the state of the system, as we will show in the next examples.  Notice that, while the minimal SDG is unique ($\gambles^+$), there are many distinct maximal SDGs. This means that, in general, a set of desirable gambles can be extended to different maximal SDGs.\footnote{For the extension to be possible, the original set must at least avoid partial loss, though.}

To clarify these and the previous definitions, we now list some examples of SDG using the classical and quantum coin examples.

\begin{SnugshadeE}
 \begin{example}[Classical coin.]
\label{ex:coinB}
 Let us go back to our classical coin example. We consider four possible situations.
 \begin{enumerate}
  \item Alice is in a complete state of ignorance about the coin (the coin may be heavily biased to one side or to the other) and so Alice decides to accept only gambles that are non-negative, i.e., all the gambles $G=diag(g)\gneq 0$ with $g=(g_1,g_2)$ such that 
$g_1,g_2\geq 0$. In this way, she can never lose utiles. We can plot Alice's SDG in a two-dimensional Cartesian coordinate system.  The result is shown in Figure~\ref{fig:coin}(a).
 Hence, a complete state of ignorance corresponds to an SDG that coincides with the first quadrant (excluded the zero), i.e., $\gambles^+$.
 \item This time Alice has some information about the coin and she accepts taking some risk. Based on this 
information, Alice finds the gambles  $f=(0.8,-0.2)$ and $h=(-0.4,0.6)$ desirable but not strictly
(besides the gambles $G=diag(g)\gneq 0$ that are always desirable for her).
This means that Alice's SDG is this time the cone in Figure~\ref{fig:coin}(b). It is an open cone that includes
all the non-zero non-negative gambles and has as border the rays $\lambda f$ for  $\lambda>0$ 
and $\gamma h$ for  $\gamma>0$.
\item We consider a third situation in which Alice has even more information on the coin
and finds the gambles  $f=(1,-1)$ and $h=(-1,1)$ desirable but not strictly, Figure~\ref{fig:coin}(c).
The resulting SDG is in this case a degenerate cone (a \textbf{maximal SDG}). It is clear (also from the figure)
that a degenerate convex cone cannot be further enlarged by including other gambles, as it would not be a cone anymore.
In this sense it is maximal: there are not other SDGs including it. 
We have already seen in Example \ref{ex:coin0} that this is the SDG corresponding to a fair coin.
{We will show in Section \ref{sec:1staxiom} that a maximal SDG always corresponds to a situation in which Alice assigns a single probability to Head and Tail. 
In the previous two cases Alice's SDG is more generically represented by a set of probabilities.} 
\item Finally assume that Alice finds $f=(2,-1)$ and $h=(-2,1)$ desirable, i.e., $f,g \in \domain$. Then we can easily show
that $\domain$ is incoherent. In fact, we have shown that if Alice accepts $f$ and $g$ she must also accept $\alpha f+ \beta g$ with  $\alpha,\beta>0$.
However, for $\alpha=\beta=0.5$ we have that  $(f + h)/2=(-0.5,-0.5) \in \gambles^-$  and, thus, Alice
does not avoid sure (and hence partial) loss; in this sense, she is irrational. This is a form of Dutch-book incoherence.
 \end{enumerate}
In the next section we will give the probabilistic models relative to these three situations.

\centering
\begin{minipage}{0.23\textwidth}
    \begin{tikzpicture}[scale=1.5]
    \draw[->] (-1,0) coordinate (xl) -- (1,0) coordinate (xu) node[right] {$g_1$};
    \draw[->] (0,-1) coordinate (yl) -- (0,1) coordinate (yu) node[above] {$g_2$};
        \draw (0,0) circle (1.pt);
    \begin{pgfonlayer}{background}
      \draw[border] (0,0) -- (0,1) coordinate (a1away);
      \draw[border] (0,0) -- (1,0) coordinate (a2away);
      \fill[blue!50,  opacity=0.5] (0,0) -- (1,0) -| (1,1) --(0,1);
    \end{pgfonlayer}
  \end{tikzpicture}
     \captionof{subfigure}{}
  \end{minipage}\quad
\begin{minipage}{0.23\textwidth}
      \begin{tikzpicture}[scale=1.5]
    \draw[->] (-1,0) coordinate (xl) -- (1,0) coordinate (xu) node[right] {$g_1$};
    \draw[->] (0,-1) coordinate (yl) -- (0,1) coordinate (yu) node[above] {$g_2$};
    \draw (0,0) circle (1.pt);
    \begin{pgfonlayer}{background}
      \draw[dashed] (0,0) -- (-0.66,1) coordinate (a1away);
      \draw[dashed] (0,0) -- (1,-0.25) coordinate (a2away);
      \fill[blue!50,  opacity=0.5] (0,0) -- (a1away) -| (a2away) --(0,0);
    \end{pgfonlayer}
  \end{tikzpicture}
     \captionof{subfigure}{}
  \end{minipage}\quad
\begin{minipage}{0.23\textwidth}
      \begin{tikzpicture}[scale=1.5]
    \draw[->] (-1,0) coordinate (xl) -- (1,0) coordinate (xu) node[above] {$g_1$};
    \draw[->] (0,-1) coordinate (yl) -- (0,1) coordinate (yu) node[above] {$g_2$};
    \draw (0,0) circle (1.pt);
    \begin{pgfonlayer}{background}
      \draw[dashed] (0,0) -- (-1,1) coordinate (a1away);
      \draw[dashed] (0,0) -- (1,-1) coordinate (a2away);
      \fill[blue!50,  opacity=0.5] (0,0) -- (a1away) -| (a2away) --(0,0);
    \end{pgfonlayer}
  \end{tikzpicture}
     \captionof{subfigure}{}
  \end{minipage}
\captionof{figure}{Alices' sets of coherent strictly desirable gambles corresponding to three different 
degrees of beliefs about the  classical coin.  \label{fig:coin}}
\end{example}
\end{SnugshadeE}

\begin{SnugshadeE}
\begin{example}[Quantum coin.]
\label{ex:spinB}
Let us go back to the qubit, Example \ref{ex:spinA}. We consider four possible situations
that are  equivalent to those discussed for the classical coin.
\begin{enumerate}
 \item Alice is completely ignorant about the quantum system. Hence, she decides 
to accept only gambles that cannot decrease her utiles (no matter the outcome of the quantum experiment). This means that her SDG coincides with  $\domain_1=\{G \in \CH \mid G\gneq0\}$ -- she only accepts non-zero non-negative 
gambles. To plot the full  $\domain_1$, now we need four dimensions. We exploit  the fact that any Hermitian matrix can be decomposed as
 \begin{equation}
 \label{eq:pauli}
 G=\left[\begin{matrix}
          v+z & x-\iota y\\
          x+\iota y & v-z
         \end{matrix}\right]=
  vI+x\sigma_x+y\sigma_y+z\sigma_z, 
  \end{equation} 
 where $(x,y,z,v)\in \mathbb{R}^4$ and $\sigma_i$ are the Pauli matrices (a basis for $\mathbb{C}_h^{2 \times 2}$):
 \begin{equation}
 \label{eq:paulimat}
 \sigma_x=\left[\begin{matrix}
          0 & 1\\
          1 & 0
         \end{matrix}\right],~~~ \sigma_y=\left[\begin{matrix}
          0 & -\iota\\
          \iota & 0
         \end{matrix}\right],~~~\sigma_z=\left[\begin{matrix}
          1 & 0\\
          0 & -1
         \end{matrix}\right].
  \end{equation} 
  The matrix $G$ is PSDNZ provided
 that $v > 0$ and $x^2+y^2+z^2\leq v^2$ and therefore
  $$
  \domain_1=\{(x,y,z,v)\in \mathbb{R}^4 \mid v \neq 0, ~\sqrt{x^2+y^2+z^2}\leq v\},
  $$
  which is the four-dimensional version of the ``ice cream cone''. Figure~\ref{fig:cones}(a)
  shows a three-dimensional slice of this cone.
\item Assume now Alice has more information about the quantum system, she is ready to take some risk. We assume 
that her SDG $\domain_2$ coincides with 
$$
\begin{array}{rcl}
\domain_2&=&\{G \in \CH \mid G \gneq0\} \cup \{G \in \CH \mid Tr(G^\dagger D_1)>0 \text{ and } Tr(G^\dagger 
D_2)>0\},
\end{array}
$$
where $Tr$ denotes the trace and $D_1=diag(0.2,0.8)$, $D_2=diag(0.6,0.4)$.
It can be verified that $\domain_2$ is a valid SDG. From $\domain_2$ it follows for instance that she finds the gambles
$$
G_1=\left[\begin{array}{cc}
     3 & 0\\0 &  -0.5
    \end{array}\right], ~~~G_2=\left[\begin{array}{ll}
     10 & 1-\iota\\ 1+\iota &  -2
    \end{array}\right]
$$
to be desirable.
We can plot  Alice's SDG also in this case exploiting Pauli's matrix decomposition~\eqref{eq:pauli} of $G$.
We obtain that
$$
\begin{array}{l}
Tr(G^\dagger D_1)=0.2(v+z)+0.8(v-z)>0, ~~Tr(G^\dagger D_2)=0.6(v+z)+0.4(v-z)>0,\end{array}
$$
which define two open hyper-spaces in $\mathbb{R}^4$; whence 
$$
\begin{array}{rcl}
\domain_2&=&\{(x,y,z,v) \in \mathbb{R}^4 \mid  v \neq 0, ~\sqrt{x^2+y^2+z^2}\leq v\}\\
&\cup& \{(x,y,z,v) \in \mathbb{R}^4 \mid 0.2(v+z)+0.8(v-z)>0,~~0.6(v+z)+0.4(v-z)>0\}.
\end{array}
$$
The projection of $\domain_2$ on the two-dimensional plane of coordinates $(v+z,v-z)$ is equal to the convex cone in Figure~\ref{fig:coin}(b).
\item Alice's SDG $\domain_3$ this time coincides with 
$$
\domain_3=\{G \in \CH \mid G \gneq0\} \cup\{G \in \CH \mid Tr(G^\dagger D)>0 \},~\textit{ with }D=\frac{1}{2}\left[\begin{array}{cc}
    1 & -\iota\\\iota &  1
    \end{array}\right].
$$
 An example of gamble that is desirable is $G=[1,-\iota;\iota,-2]$, since 
$Tr(G^\dagger D)=0.5>0$. It can be verified that $\domain_3$ is a valid SDG and is also a degenerate cone (\textbf{maximal} SDG).
To show that, let us consider $H=F-\epsilon I \notin \domain_3$ such that  $Tr(F^\dagger D) \leq 0$
and so $Tr(F^\dagger D)-\epsilon Tr( D)=Tr(F^\dagger D)-\epsilon  < 0$  for any $\epsilon>0 $ (we have exploited $Tr( D)=1$).
Then consider $G=-H$. $G$  is in $\domain_3$ because  $Tr(G^\dagger D) =-Tr(H^\dagger D)> 0$ 
 for any  $\epsilon>0$.
Since $H-H=0$, $\domain_3$ cannot include any further gamble while remaining SDG (for any  $\epsilon>0$, openness).

We will show that a {maximal SDG denotes a situation in which Alice uses a single density matrix to represent her beliefs about the state of the quantum system.}
By exploiting Pauli's matrix decomposition~\eqref{eq:pauli} of $G$, we obtain that
$Tr(G^\dagger D)=v+y>0$. The projection of $\domain_3$ on $(v,y)$ is shown in Figure~\ref{fig:cones}(c).
\item Finally, assume that  Alice's   SDG $\domain_4$ includes the gambles
$$
G_1=\left[\begin{array}{cc}
     1 & -\iota\\\iota &  -2
    \end{array}\right], ~~~G_2=\left[\begin{array}{ll}
     -2 & \iota\\ -\iota &  1
    \end{array}\right];
$$
then she is irrational. In fact,  desirability of $G_1,G_2$ implies desirability of
$\alpha G_1+ \beta G_2$. Considered  $\alpha=\beta=0.5$, we have that
$(G_1+G_2)/2<0$; Alice
does not avoid sure loss. Again this is a form of Dutch-book incoherence.
\end{enumerate}

\centering
\begin{minipage}{0.30\textwidth}
\includegraphics[width=4cm]{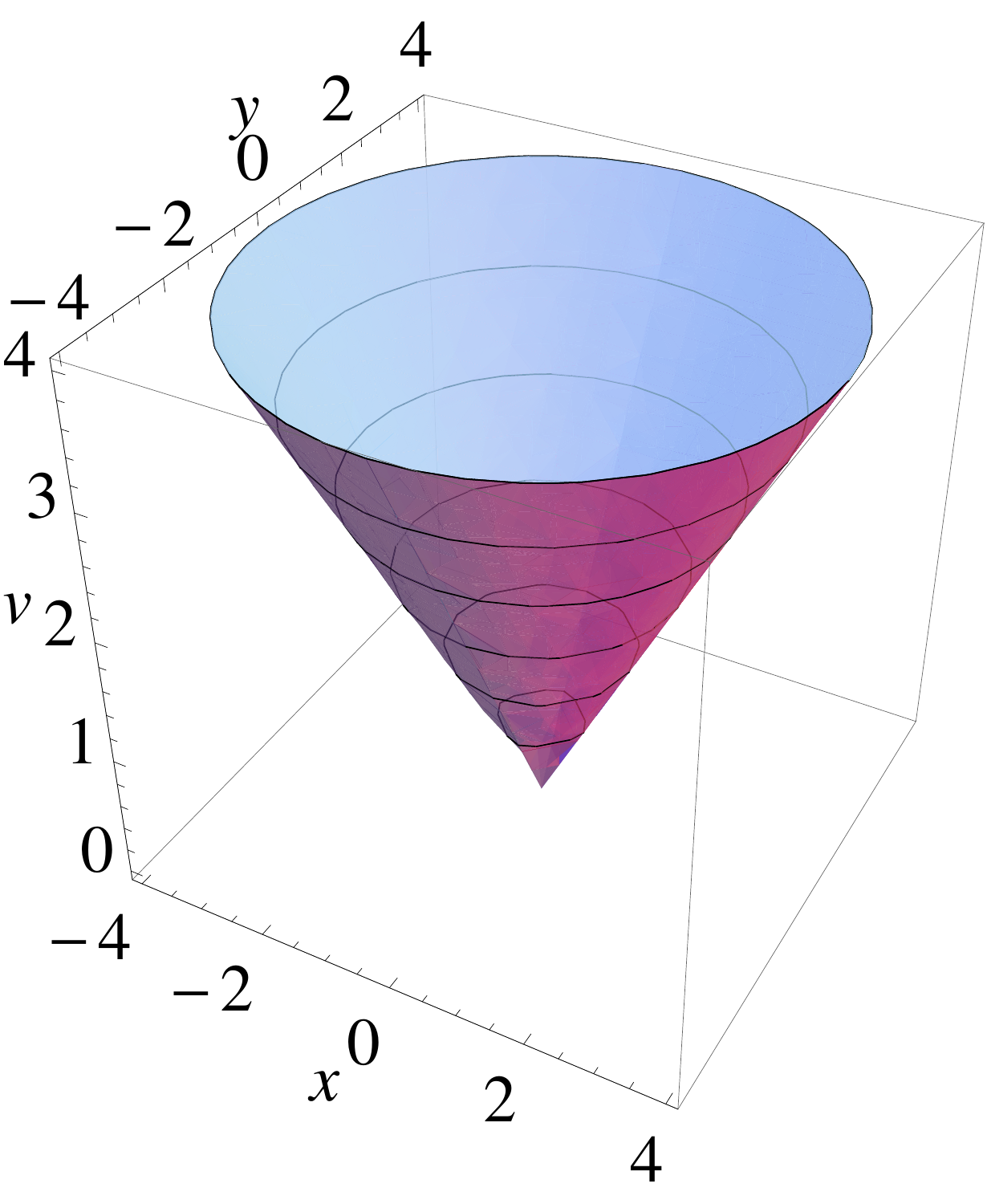}
   \captionof{subfigure}{}
  \end{minipage}
\begin{minipage}{0.30\textwidth}
      \begin{tikzpicture}[scale=1.5]
    \draw[->] (-1,0) coordinate (xl) -- (1,0) coordinate (xu) node[right] {$v+z$};
    \draw[->] (0,-1) coordinate (yl) -- (0,1) coordinate (yu) node[above] {$v-z$};
    \draw (0,0) circle (1.pt);
    \begin{pgfonlayer}{background}
      \draw[dashed] (0,0) -- (-0.66,1) coordinate (a1away);
      \draw[dashed] (0,0) -- (1,-0.25) coordinate (a2away);
      \fill[blue!50,  opacity=0.5] (0,0) -- (a1away) -| (a2away) --(0,0);
    \end{pgfonlayer}
  \end{tikzpicture}
     \captionof{subfigure}{}
  \end{minipage}
    \begin{minipage}{0.30\textwidth}
      \begin{tikzpicture}[scale=1.5]
    \draw[->] (-1,0) coordinate (xl) -- (1,0) coordinate (xu) node[above] {$v$};
    \draw[->] (0,-1) coordinate (yl) -- (0,1) coordinate (yu) node[above] {$y$};
    \draw (0,0) circle (1.pt);
    \begin{pgfonlayer}{background}
      \draw[dashed] (0,0) -- (-1,1) coordinate (a1away);
      \draw[dashed] (0,0) -- (1,-1) coordinate (a2away);
      \fill[blue!50,  opacity=0.5] (0,0) -- (a1away) -| (a2away) --(0,0);
    \end{pgfonlayer}
  \end{tikzpicture}
     \captionof{subfigure}{}
  \end{minipage}
  \captionof{figure}{Projection of Alice's sets of coherent strictly desirable gambles corresponding to three different 
degrees of beliefs about the  quantum system.}
  \label{fig:cones}
\end{example}
\end{SnugshadeE}


\section{First  axiom of QM}
\label{sec:1staxiom}
The aim of this section is to derive the first axiom of QM as a direct consequence of Alice's rationality.
The first axiom states:
   \begin{SnugshadeF}
\begin{MethodF}[Axiom I:]
associated to any isolated physical system is a complex vector space
with inner product (that is, a Hilbert space) known as the state space of the
system. The system is completely described by its density operator, which is a
positive operator $\rho$ with trace one, acting on the state space of the system.
\end{MethodF}
\end{SnugshadeF}
We will proceed in two steps. First, by exploiting the natural duality in  Hilbert spaces defined by the inner product, we will show that the dual of a coherent set of strictly desirable gambles is always included in the set of positive operators (PSD matrices). 
Second, by imposing a further property to the duality transformation, i.e.,  that it has to transform constants to constants -- or, in other words, that it has to preserve the utility scale --, we will prove that the dual of a coherent set of strictly desirable gambles is  a closed and convex set of 
 \textit{trace one positive operators}: i.e., a closed and convex set of density matrices $\rho$. Imposing the last property can be regarded as a ``normalisation'' in the space of positive operators; for this reason there is no loss of generality in doing so.

\subsection{Duality of coherence}

In mathematics, any vector space  $V$ has a corresponding dual vector space  $V^\bullet$ consisting of all linear 
functionals on $V$.
The space $\CH$ is a vector space and, therefore, we can define  its dual by means of the standard inner product 
on $\CH$:
$$
 G \cdot R = Tr(G^\dagger R),
$$
with $R,G\in \CH$. 
By using the inner product, we start by defining a general -- and standard, see, e.g., \citet[Sec.~2.2]{aliprantis2007cones} -- notion of dual for any subset of $\CH$.

\begin{definition}\label{def:dual-set}
Let $C$ be a subset of $\CH$.  Its  {\bf dual set} is defined as
\begin{equation}
\label{eq:polarset}
 C^\bullet=\{ \Rho \in \CH \mid G\cdot \Rho \geq  0~ \forall G \in C\}.
\end{equation} 
\end{definition}

A set of strictly desirable gambles $\domain$ is a subset of $\CH$ and, therefore, we can define its dual
through the transformation~\eqref{eq:polarset}. We can easily show that the dual is
\begin{equation}
\label{eq:polarcone}
\domain^\bullet=\{ \Rho \in \CH \mid \Rho\geq0,~~ G\cdot \Rho \geq  0~ \forall G \in \domain\}.
\end{equation} 
This means that the dual set of an SDG includes only \textbf{positive operators}, i.e., Hermitian matrices
that are PSD.
This result follows by the fact that, as stated by condition (S2) in Definition \ref{def:sdg},  $ \domain$ includes all PSDNZ matrices.  
In fact, take the PSDNZ gamble $ G=u_iu^\dagger_i $  with $0\neq u_i \in \complexs^n$  and so $G \in \domain$ (since $G\gneq0$) then:
 $$
 G\cdot \Rho=Tr(G^\dagger \Rho)=Tr(u_iu^\dagger_i \Rho)=Tr(u^\dagger_i \Rho u_i)=u^\dagger_i \Rho u_i;
 $$
 but $u^\dagger_i \Rho u_i$ must be greater than zero for $\Rho$ to be in  $\domain^\bullet$.
 Since $u^\dagger_i \Rho u_i\geq 0$ must hold for any $u_i$, this  implies that $\Rho\geq 0$
 by definition of PSD matrix.
It follows immediately from the definitions that 
the set $\domain^\bullet$ is a closed convex cone. This is proven in the Appendix, Proposition \ref{pro:rho}. We  therefore know that the dual of an SDG only includes PSD matrices ($\Rho\geq0$), and  it is a closed convex cone.
It can also be verified that when Alice is in a full state of ignorance, i.e., $\domain=\gambles^+$, then
 $\domain^\bullet=\{ \Rho \in \CH \mid \Rho\geq0\}$, i.e., the dual includes all PSD matrices.
 
 We now impose a further property to the duality transformation $(\cdot)^\bullet$.
 \begin{definition}\label{def:constant} 
We say that a matrix $\Rho \in \CH$ {\bf preserves constants} if
\begin{align}
\label{eq:mapconstants}
G\cdot \Rho=
Tr(G^\dagger \Rho)=c,
\end{align}
for each constant matrix\footnote{Here the term constant means matrices that represents constant payoffs. It can be verified that all  constant Hermitian matrices  are of the form $cI$ with $c \in \mathbb{R}$.} $G=cI$  with $c \in 
\mathbb{R}$.
\end{definition}
We then ask that members of the dual of a set of gambles  preserves constants. The meaning of adding 
condition~\eqref{eq:mapconstants} to the definition of the dual set  is simply that  we aim at preserving the scale in which utility is measured.\footnote{Stated in terms of money it means that
one euro in the primal remains one euro in the dual.}
In such case, the following holds.
\begin{Proposition}\label{prop:dualconstant}
Let $C$ be a subset of $\CH$. If we impose that members of
  $C^\bullet$ preserve constants, then we obtain the following set: 
\begin{align}
\label{eq:credalll}
\mdesirs&=\{ R \in \CH \mid R\geq0,~~Tr(R)=1,~~ G\cdot R \geq  0~ \forall G \in C\}.
\end{align} 
\end{Proposition}
The proof of Proposition \ref{prop:dualconstant} is immediate. Clearly $\mdesirs \subseteq C^\bullet$. Let us now consider $G=cI$ with $c>0$ that is in $C$ (since $G \gneq0$).
Then for any $\Rho \in C^\bullet$, we have that  $Tr(G^{\dagger} R)=Tr(c\Rho)=c Tr(\Rho)$.
To satisfy~\eqref{eq:mapconstants}, we therefore must have  $c Tr(\Rho)=c$, which implies that $Tr(\Rho)=1$. 

{
Geometrically, it can be observed that  $\mdesirs$ is a bounded section of the closed convex cone $C^\bullet$ (it is a conic section)
and, therefore, is a closed convex set -- it is not a cone anymore. Since $\mdesirs$ is a conic section of the dual, it completely characterises $C^\bullet$. Hence, by an abuse of terminology, we will also refer to $\mdesirs$ as the dual of $C$.}

{When $\mdesirs$ is (the section of) the dual of a SDG $\domain$ we call it a \textbf{quantum credal set}.\footnote{In the classical theory of desirable gambles, a credal set is a closed convex set of
probabilities \cite{levi1980enterprise,walley1991,augustin2014}.} 
The term \textit{credal} is used because $\mdesirs$ is  the dual of Alice's set of desirable gambles 
$\domain$. In fact, since $\domain$ reflects Alice's beliefs about the quantum system  then, by duality, also
$\mdesirs$ reflects her beliefs about the quantum system. }

A Hermitian matrix that is PSD and with trace one is by definition a \textbf{density matrix} (positive operator with trace one).

  \begin{SnugshadeB}
\begin{MethodB}[Subjective formulation of the first axiom of QM]\\
The dual of Alice's set of coherent strictly desirables gambles is the quantum credal set:
\begin{align}
\label{eq:credaldef}
\mdesirs&=\{ \rho \in \CD \mid G\cdot \rho \geq  0~ \forall G \in \domain\},
\end{align} 
where  $\CD=\{ R \in \CH \mid R\geq0,~~Tr(R)=1\}$ is the set of all density matrices.
\end{MethodB}
\end{SnugshadeB}

Note that the above formulation is more general then the first axiom of QM, in the sense that it gives us more flexibility in modelling uncertainty about a certain experiment. Indeed, it allows us to consider any situation in between the state of full ignorance $\domain=\gambles^+$, and thus $\mdesirs=\CD$,\footnote{In QM, the state of maximum ignorance is usually defined as $\rho=I/n$ ($I$ is the identity matrix), but strictly speaking this is not correct. {The single density matrix $\rho=I/n$ means that for Alice all the directions are equally probable. If we were in the classical coin example this would in fact mean that Alice is assuming that the coin is fair.  This will be clarified in the examples.}}
and a situation  in which Alice is sure about the state of the system, that is the situation where  her coherent sets of strictly desirable gambles is a degenerate cone, i.e., a maximal SDG. Notice that in the latter case, the dual $\mdesirs$ includes a single density matrix, and we thus obtain exactly the \textit{first axiom of QM}.

  \begin{SnugshadeF}
\begin{MethodF}[Desirability and classical probability]\\
In Example \ref{ex:coinA}, we have seen that  a gamble relative to  a classical experiment is a diagonal $G=diag(g)$ with $g \in \reals^n$.
The inner product with $\rho$ gives:
 $$
 G \cdot \rho=Tr(G^\dagger \rho)=\sum_{i=1}^n g_i \rho_{ii},
 $$
 where $\rho_{ii}$ denote the diagonal elements of $\rho$, which are real numbers.
Therefore, without loss of generality, we can also assume that $\rho$ is diagonal.
 From~\eqref{eq:credaldef}, we can derive that the dual of the SDG  $\domain$ is:
$$
\mdesirs=\left\{ \rho \in \mathcal{D}_h^{n \times n} \text{ diagonal } \Big| \sum_{i=1}^n g_i \rho_{ii} \geq 0~ \forall diag(g) \in \domain\right\}.
$$
Since $\rho$ is positive, with trace one and diagonal, this can also be rewritten as
$$
\mdesirs=\left\{ p \in \reals^n \Big|   \sum_{i=1}^n g_ip_i \geq 0~ \forall g 
\in \domain\right\},
$$
where we have used the substitution $p_i=\rho_{ii}$, so that $p_i\geq0$ and $\sum_{i=1}^np_i=1$.  So the dual $\mdesirs$ of $\domain$ in the classical case is a set 
of probability mass functions.  This is exactly standard duality in probability and proves the correspondence
between SDG and sets of probabilities in the classical case.
This is just a particular case of the theory we have developed here for quantum mechanics.
\end{MethodF}
\end{SnugshadeF}

\begin{SnugshadeE}
 \begin{example}[Classical coin.]
 \label{ex:coinC}
The three cones discussed in Example~\ref{ex:coinB} are shown in Figure~\ref{fig:coindual} (in blue) together with their dual cones (in red) $\domain^\bullet$.  The 
(quantum) credal sets $\mdesirs$  are equal to the intersection of $\domain^\bullet$ with the line $\{p\geq 0 \mid \sum_{i=1}^2 p_i=1\}$ (normalisation, black line). Note that the dual of the set of non-zero non-negative gambles is the first quadrant  Figure~
\ref{fig:coindual}(a), while the dual of a maximal cone is a line  -- see Figure~\ref{fig:coindual}(c) --
whose intersection with $\{p\geq 0 \mid \sum_{i=1}^2 p_i=1\}$ gives the point $p=(0.5,0.5)$, i.e.,
the probability mass function of a fair coin. 

Let us discuss the subjective interpretation of the three cases. The (quantum) credal set of the state of full ignorance is the set of all probability mass functions: for Alice everything is possible.
In the second case, the (quantum) credal set is equal to the closed convex set of probability mass functions:
$$
\mdesirs_2=\{(0.2,0.8)w+(1-w)(0.6,0.4) \text{ with } w\in [0,1]\}.
$$
Alice believes now that the probability of head is in the interval  $[0.2,0.6]$ and, thus,
the probability of tail in  $[0.4,0.8]$.
In the third case Alice is  sure that the coin is fair. A maximal SDG corresponds to a single probability mass function:
$\mdesirs_3=\{(0.5,0.5)\}$.

\centering
\begin{minipage}{0.30\textwidth}
    \begin{tikzpicture}[scale=1.5]
    \draw[dashed,->] (-1,0) coordinate (xl) -- (1.5,0) coordinate (xu)  node[right] {$g_1$};
    \draw[dashed,->] (0,-1) coordinate (yl) -- (0,1.5) coordinate (yu) node[above] {$g_2$};
    \draw[-] (1,0) -- (0,1);
        \draw (0,0) circle (1.pt);
        \draw[fill] (1,0) circle (0.7pt) node[below] {(1,0)};
         \draw[fill] (0,1) circle (0.7pt) node[left] {(0,1)};
    \begin{pgfonlayer}{background}
      \draw[dashed] (0,0) -- (0,1.2) coordinate (a1away);
      \draw[dashed] (0,0) -- (1.2,0) coordinate (a2away);
      \fill[blue!50,opacity=0.5] (0,0) -- (1.2,0) -| (1.2,0) --(0,1.2);
    \end{pgfonlayer}
        \begin{pgfonlayer}{background}
      \draw[border] (0,0) -- (0,1.5) coordinate (a1away);
      \draw[border] (0,0) -- (1.5,0) coordinate (a2away);
      \fill[red!50,opacity=0.5] (0,0) -- (1.5,0) -| (1.5,0) --(0,1.5);
          \end{pgfonlayer}
  \end{tikzpicture}   \captionof{subfigure}{}
  \end{minipage}  
  \begin{minipage}{0.30\textwidth}
          \begin{tikzpicture}[scale=1.5]
    \draw[->] (-1,0) coordinate (xl) -- (1.5,0) coordinate (xu)  node[right] {$g_1$};
    \draw[->] (0,-1) coordinate (yl) -- (0,1.5) coordinate (yu) node[above] {$g_2$};
          \draw[fill] (1,0) circle (0.7pt) node[below] {(1,0)};
         \draw[fill] (0,1) circle (0.7pt) node[left] {(0,1)};
    \draw[-] (1,0) -- (0,1);
        \draw (0,0) circle (1.pt);
        \begin{pgfonlayer}{background}
      \draw[dashed] (0,0) -- (-1.122,1.7) coordinate (a1away);
      \draw[dashed] (0,0) -- (1.7,-0.425) coordinate (a2away);
      \fill[blue!50,  opacity=0.5] (0,0) -- (-1.122,1.7) -| (-1.122,1.7)  --(1.7,-0.425);
    \end{pgfonlayer}
        \begin{pgfonlayer}{background}
      \draw[border] (0,0) -- (1.2,0.8) coordinate (a1away);
      \draw[border] (0,0) -- (0.4,1.6) coordinate (a2away);
      \fill[red!50,  opacity=0.7] (0,0) -- (1.2,0.8) -| (1.2,0.8) --(0.4,1.6);
    \end{pgfonlayer}
  \end{tikzpicture}
   \captionof{subfigure}{}
  \end{minipage}
  \begin{minipage}{0.30\textwidth}
   \begin{tikzpicture}[scale=1.5]
    \draw[->] (-1,0) coordinate (xl) -- (1.1,0) coordinate (xu)  node[right] {$g_1$};
    \draw[->] (0,-1) coordinate (yl) -- (0,1.1) coordinate (yu) node[above] {$g_2$};
           \draw[fill,red] (0.5,0.5) circle (0.7pt);
        \draw[fill] (1,0) circle (0.7pt) node[below] {(1,0)};
         \draw[fill] (0,1) circle (0.7pt) node[left] {(0,1)};
    \draw[-] (1,0) -- (0,1);
        \draw (0,0) circle (1.pt);
        \begin{pgfonlayer}{background}
      \draw[dashed] (0,0) -- (-1.1,1.1) coordinate (a1away);
      \draw[dashed] (0,0) -- (1.1,-1.1) coordinate (a2away);
      \fill[blue!50,  opacity=0.5] (-1.1,1.1) -- (1.1,1.1) -| (1.1,1.1)  --(1.1,-1.1);
    \end{pgfonlayer}
        \begin{pgfonlayer}{background}
      \draw[border,red!50] (0,0) -- (1.1,1.1) coordinate (a1away);
    \end{pgfonlayer}
  \end{tikzpicture}
     \captionof{subfigure}{}
  \end{minipage}
    \captionof{figure}{Dual sets for the classical coin.}
    \label{fig:coindual}
\end{example}
\end{SnugshadeE}

\begin{SnugshadeE}
\begin{example}[Quantum Coin.]
\label{ex:spinC}
Let us go back to the qubit of Example~\ref{ex:spinB}.
 
From~\eqref{eq:credal1}, the dual of the first SDG $\domain_1=\{G \in \CH \mid G\gneq0\}$ is
$\mdesirs=\CD$; it is the set of all valid density matrices, which is clearly a state of complete ignorance.
Also in this case we can plot $\mdesirs$ exploiting Pauli's matrix decomposition~\eqref{eq:pauli}:
$$
 \rho=\left[\begin{matrix}
          0.5+z & x-\iota y\\
          x+\iota y & 0.5-z
         \end{matrix}\right]=
  \frac{1}{2}I+x\sigma_x+y\sigma_y+z\sigma_z. 
 $$
  Note that $Tr( \rho)=1$, but it must also be PSD:  this means that  $x^2+y^2+z^2\leq \frac{1}{4}$, which is the 
equation of a sphere (the Bloch sphere) --  see Figure~\ref{fig:sphera}(a).

For the second cone, 
$$
\begin{array}{rcl}
\domain_2&=&\{G \in \CH \mid G \gneq0\}\\
&\cup& \{G \in \CH \mid Tr(G^\dagger D_1)>0 \text{ and } Tr(G^\dagger 
D_2)>0\},
\end{array}
$$
with  $D_1=diag(0.2,0.8)$ and $D_2=diag(0.6,0.4)$; we have that the quantum credal set coincides with
$$
\mdesirs_2=\{ \rho=w D_1+ (1-w) D_2 \mid w\in[0,1]\}.
$$
This is the set of all density matrices obtained as convex combinations of $D_1,D_2$.
In fact, first note that $Tr(w D_1+ (1-w) D_2 )=1$ and $w D_1+ (1-w) D_2  \geq0$ and so $\mdesirs_2$ is a set of 
valid density matrices. Moreover, it is clear that for any $G$ such that $Tr(G^\dagger D_1)>0$ and $Tr(G^\dagger D_2)>0$, $Tr(G^\dagger (w D_1+ (1-w) D_2))>0$ for any $w  \in[0,1]$. 
The set $\mdesirs_2$ is shown in Figure~\ref{fig:sphera}(b) (red line), inside the Bloch sphere.
This is equivalent to the second case described in Example \ref{ex:coinC}. Alice believes
that the density matrix of the quantum system belongs to the closed convex set
 $w D_1+ (1-w) D_2 $ with  $w\in[0,1]$.
 Note that this convex combination \textbf{does not indicate a mixing state}, it means we are considering
 all density matrices that can be obtained as  $w D_1+ (1-w) D_2 $.
 They are all mixed density matrices for $w\in[0,1]$, since $D_1,D_2$ are mixed states
 and any convex combination of mixed states is mixed too.

In the third case, the dual of Alice's SDG is
$$
\domain_3=\{G \in \CH \mid G \gneq0\} \cup\{G \in \CH \mid Tr(G^\dagger D)>0 \}, \textit{ with } D=\frac{1}{2}\left[\begin{array}{cc}
    1 & -\iota\\\iota &  1
    \end{array}\right],
$$
and $\mdesirs_3=\{ \rho=D\}$ is made of a single density matrix. It is a pure state, because it is in the border of the Bloch sphere.   The set $\mdesirs_3$ is shown in Figure~\ref{fig:sphera}(b) (red point), inside  the Bloch sphere. Alice believes  that the state of the quantum system is $ \rho=D$.
Maximal SDGs  correspond to single density matrices (non necessarily ``pure states'').

\centering
  \begin{minipage}{0.30\textwidth}
\includegraphics[width=4cm]{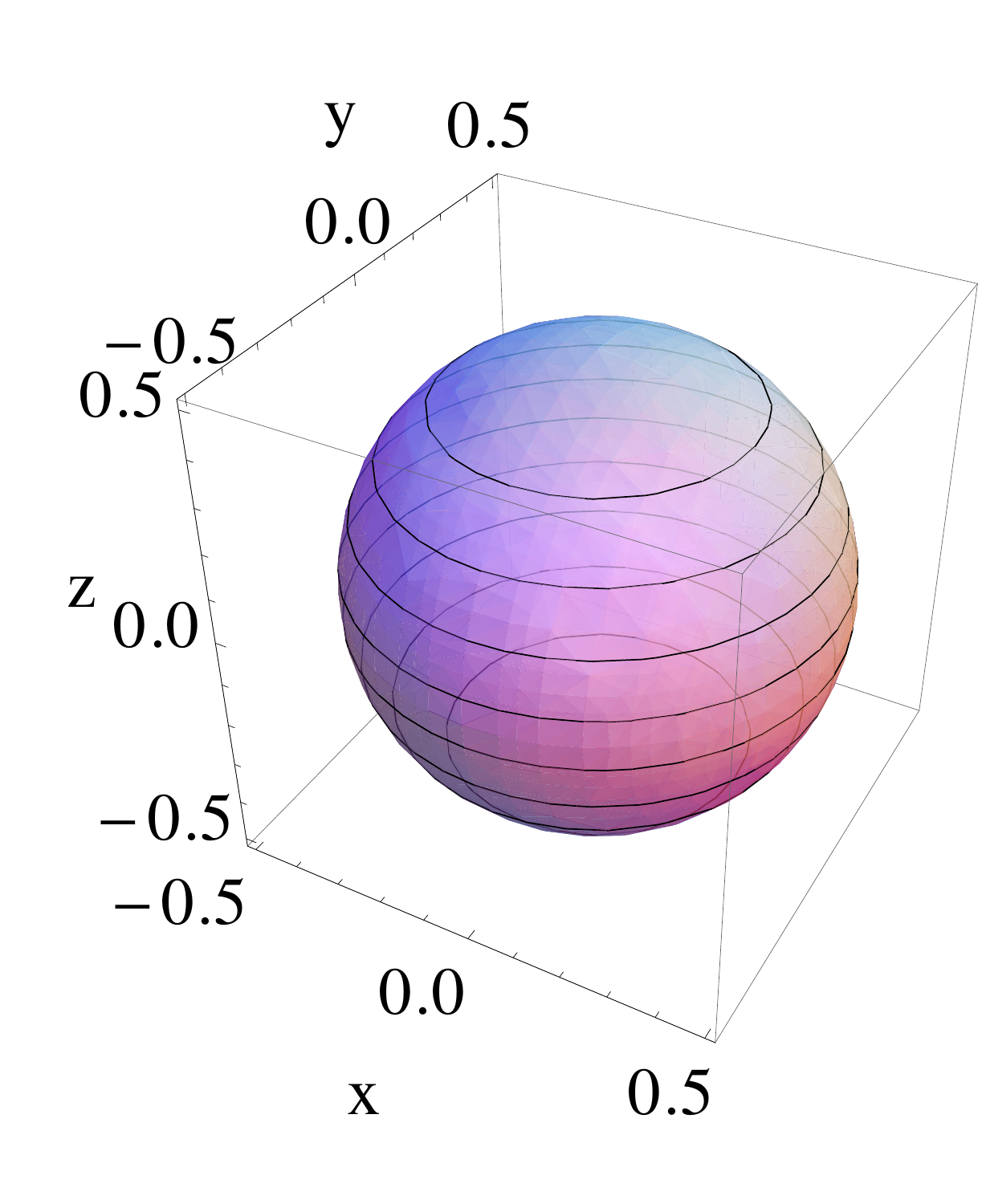}
     \captionof{subfigure}{}
  \end{minipage}
  \begin{minipage}{0.30\textwidth}
\includegraphics[width=4cm]{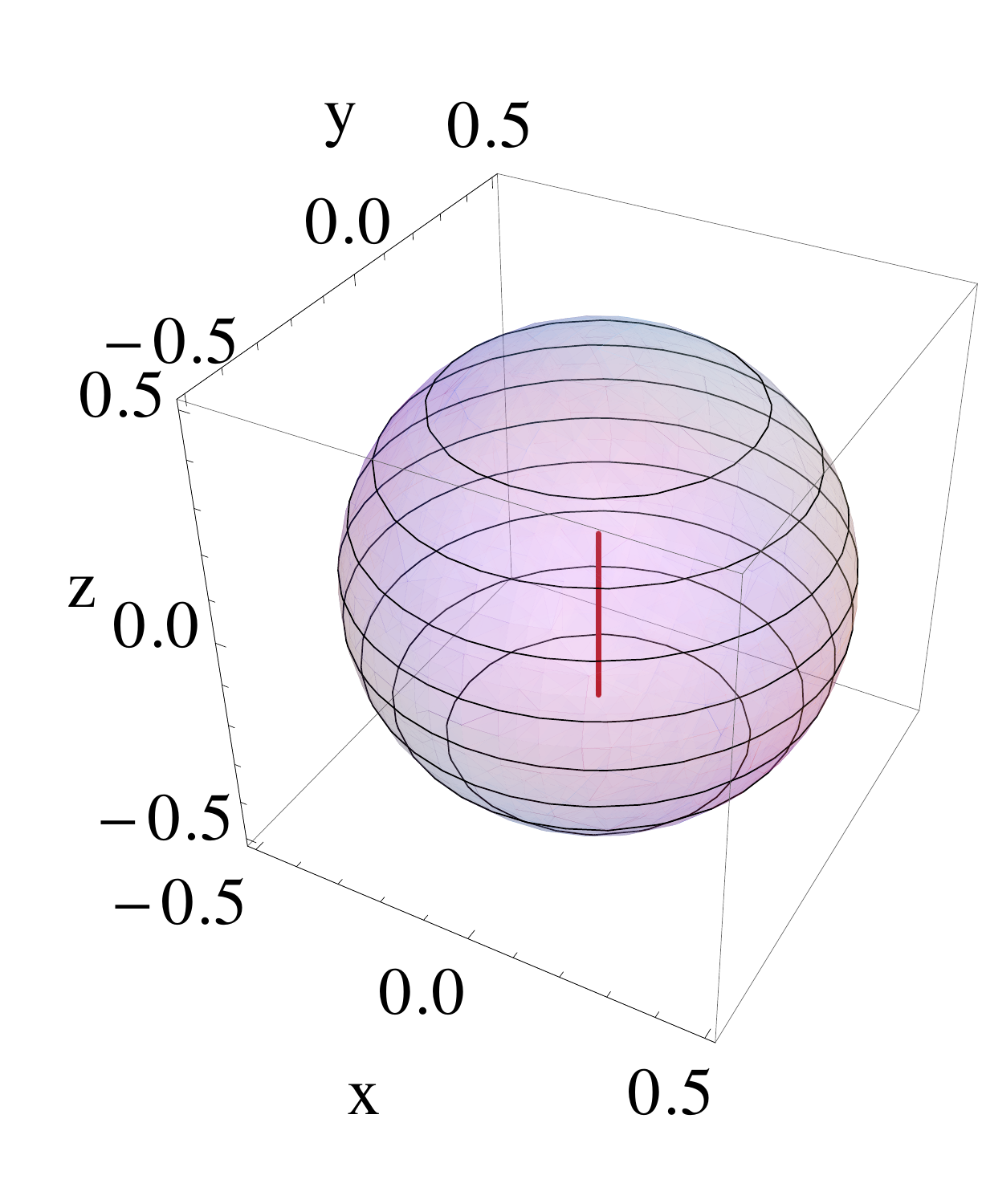}
     \captionof{subfigure}{}
  \end{minipage}
  \begin{minipage}{0.30\textwidth}
  \vspace{7mm}
\includegraphics[width=4cm]{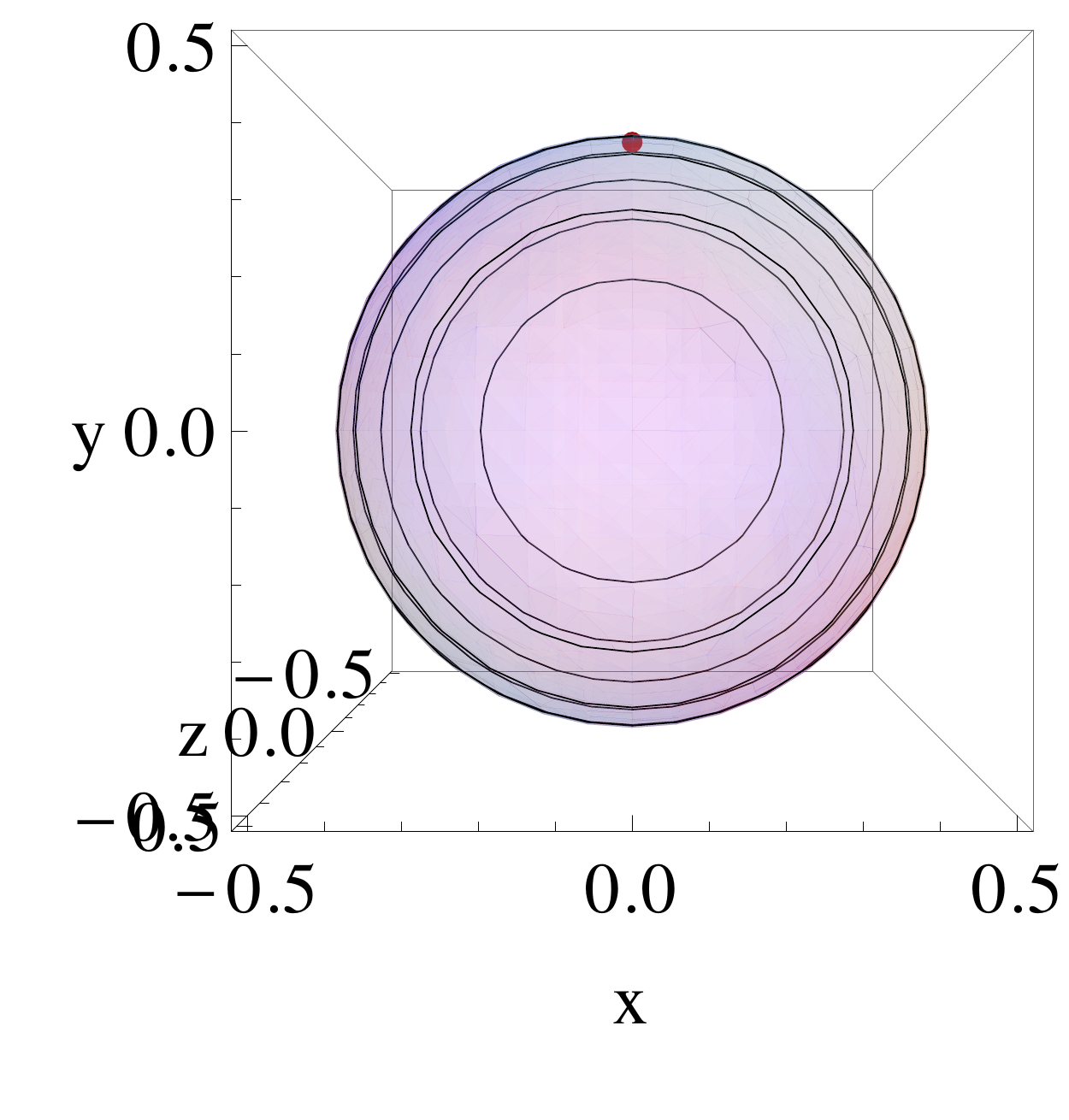}
     \captionof{subfigure}{}
  \end{minipage}
    \captionof{figure}{Credal quantum sets  for the quantum coin in the three cases.\label{fig:sphera}}    
\end{example}
\end{SnugshadeE}

\subsection{The representation theorem}

To fully deduce the first axiom of QM as a consequence of Alice's rationality, we must also prove that the duality map $(\cdot)^\bullet:\domain \mapsto\mdesirs$ establishes a bijection between strictly desirable sets of gambles $\domain$ and quantum credal sets $\mdesirs$.
This can be done by firstly verifying that~\eqref{eq:credaldef} can also be rewritten as (Appendix, Proposition \ref{prop:credal1equiv}):
\begin{equation}
\label{eq:credal1}
\mdesirs=\{ \rho \in \CD \mid   G\gneq0 \text{ or } G\cdot \rho > 0~ \forall G \in \domain\},
\end{equation} 
which means that $\mdesirs$ includes all  $\rho\in \CD$ such that $ G\cdot \rho > 0$ for all $G \in \domain$ that are not  $G\gneq0 $.
Note that $\mdesirs=\CD$ whenever $\domain$  includes only PSDNZ matrices ($G\gneq0$).
The subjective interpretation of~\eqref{eq:credal1} is straightforward:  Alice always accepts the gambles $G\gneq0 $ and, thus, they are not informative on Alice's beliefs
about the quantum system. The latter are indeed determined by the gambles that are not $G\gneq0 $.
From  these considerations, it can easily be derived that  the inverse of $(\cdot)^\bullet$, i.e. $(\cdot)^\circ :\mdesirs \mapsto \domain$,  is 
\begin{equation}
\label{eq:credal2}
\mdesirs^\circ=\{ G \in \CH \mid G  \gneq0\}\cup\{G \in \CH \mid  G\cdot \rho > 0~ \forall \rho \in \mdesirs\}.
\end{equation} 
We observe that the dual of a quantum credal set $\mdesirs$ is the union of the set of all the gambles $G\gneq 0$
with the gambles $G \cdot \rho> 0, ~~\rho\in \mdesirs$. The fact that we always include the gambles $G\gneq 0$
can be read as a consequence of rationality: those gambles are always desirable for Alice.

The proof that~\eqref{eq:credal2} is the inverse of~\eqref{eq:credal1} is a consequence of the following equivalences:\footnote{The proof of Theorem  \ref{thm:dualityopen} can be found in the Appendix.}
\begin{Theorem}[Representation theorem]
\label{thm:dualityopen} 
It holds that
\begin{itemize}
\item $(\domain^\bullet)^\circ=\domain$ for every SDG $\domain$ and 
\item $( \mdesirs^\circ)^\bullet= \mdesirs$ for every quantum credal set $ \mdesirs$. 
\end{itemize}
Any SDG is in correspondence with a closed convex set of density matrices
and vice versa. 
\end{Theorem}

Assume that $\mdesirs=\{\hat{\rho}\}$, then from~\eqref{eq:credal2} it follows that
$$
\mdesirs^\circ=\{ G \in \CH \mid G  \gneq0\}\cup\{G \in \CH \mid  G\cdot \hat{\rho}=Tr(G^\dagger \hat{\rho}) > 0\},
$$
and in Example~\ref{ex:spinB}, case~3, we have shown that  this defines a maximal SDG.
Therefore,  maximal SDGs are in a one-to-one correspondence with singleton quantum credal sets, i.e.,
quantum credal sets including a single density matrix. 

\subsection{Coherent previsions, Born's rule and Gleason theorem}
\label{sec:prevision}

A coherent set of desirable gambles implicitly defines a probabilistic model for the quantum system. The way to see this is to
consider gambles of the form $G-cI$, where $c$ is a real value so that $cI$ is used as a constant gamble here, and $G \in \CH$ 
is any gamble.
Say that Alice is willing to accept the gamble  $G-cI$. From a behavioural point of view, we can reinterpret this by saying that she is willing to buy 
gamble $G$ at price $c$, since she is giving away $c$ utiles while gaining $G$. Focusing on the supremum price for which the transaction is desirable for Alice, leads us to the following:

\begin{definition}\label{def:lower}
 Let $\domain$ be a SDG. For all $G \in \domain$, 
 \begin{equation}
 \label{eq:lower}
  \underline{P}(G)= \sup \left\{c \in \mathbb{R} \mid G-cI \in  \domain\right\}
 \end{equation} 
is called the {\bf coherent lower prevision} of $G$.   \\ A {\bf coherent upper prevision} is instead defined as
 \begin{equation} \label{eq:upper}
  \overline{P}(G)= \inf \left\{c \in \mathbb{R} \mid cI-G \in  \domain\right\}.
 \end{equation}
 \end{definition}
 
Condition~\eqref{eq:upper} makes it clear that the coherent upper prevision for a gamble $G$ is Alice's infimum selling price for it. It should also be observed that lower and upper previsions are conjugate: $\overline{P}(G)=-  \underline{P}(-G)$.

The terminology ``prevision" is borrowed from de Finetti. It should be clear, however, that  previsions are nothing else but \textbf{expectations}.\footnote{In QM  the expected value of an operator $H$ (Hermitian matrix) is defined as $Tr(H^\dagger \rho)$, see for instance  \citet{nielsen2010quantum}.}

 \begin{Proposition}\label{prop:previsions-and-cqs}
 Consider SDG $\domain$ and its dual $\mdesirs$.
Then we have that, for every $G \in \CH$, 
 \begin{equation}
 \label{eq:lowupprev}
  \underline{P}(G)= \inf_{\rho \in\mdesirs } Tr(G^\dagger \rho), ~~~\;\;~~  \overline{P}(G)= \sup_{\rho 
\in\mdesirs } Tr(G^\dagger \rho).
 \end{equation}
 This shows that coherent lower (resp. upper) previsions are just the lower (resp. upper) envelope of the expectations of $G$ computed w.r.t.\ the set of density matrices 
$\mdesirs $.
\end{Proposition}
\noindent 

Indeed, first notice that, for any $\rho \in \CD$, $(G-cI)\cdot \rho=Tr((G-cI)^\dagger \rho)=Tr(G^\dagger \rho)-cTr(\rho)$. But since $Tr(\rho)=1$, we get that
\begin{equation}
\label{eq:linearityc}
(G-cI)\cdot \rho=Tr(G^\dagger \rho)-c.
\end{equation} 
Since we are looking for the supremum $c$ such that $G-cI$ is still in $\domain$, by characterisation~\eqref{eq:credal2} and~\eqref{eq:linearityc}  we obtain that   $G-cI  \in  \domain$, provided that $Tr(G^\dagger \rho)-c>0$ for all $\rho \in \mdesirs$ ($\mdesirs$ being the dual of $\domain$).
Since $Tr(G^\dagger \rho)-c>0$ for all $\rho \in \mdesirs$, it follows that 
\begin{equation}
c\leq \inf_{\rho \in\mdesirs } Tr(G^\dagger \rho) 
\end{equation} 
holds for any $c \in \mathbb{R}$ such that $ G-cI \in  \domain$. 
Since $c < \inf_{\rho \in\mdesirs } Tr(G^\dagger \rho)$ implies that $ G-cI \in  \domain$, necessarily  $\underline{P}(G)=\inf_{\rho \in\mdesirs } Tr(G^\dagger \rho)$.

Equation~\eqref{eq:lowupprev} is important and gives us a way to compute coherent lower and upper previsions in terms of a set of density matrices $\rho$. Recall that we may give rise to a set of density matrices as a consequence of Alice's total or partial ignorance about a quantum system. When Alice's beliefs correspond to a maximal SDG instead, we get that  $\mdesirs $ includes only a single density matrix, and from~\eqref{eq:lowupprev} we  get the standard expectation in QM: $\overline{P}(G)=\underline{P}(G)=Tr(G^\dagger \rho)$. The converse, that $\overline{P}(G)=\underline{P}(G)$ for any $G\in \domain$ implies 
$\mdesirs=\{\rho\} $, is proven in the Appendix, Proposition \ref{prop:linearprev}.

In the case that  $\overline{P}(G)=\underline{P}(G)$, their common value is denoted by $P(G)$ and is called the {\bf  linear prevision} of $G$. A  linear prevision is just an expectation w.r.t.\ a density matrix $\rho$; it is coherent because we have shown that so is $\rho$.

   \begin{SnugshadeB}
\begin{MethodB}[Born's rule]\\
Born's rule is a law of QM that gives the probability that a measurement on a quantum system will yield a given result.
We will show how we can derive it. Assume that $  \mdesirs=\{\hat{\rho}\}$ includes a single density matrix  so that, for $G=\Pi_i$, $\underline{P}(\Pi_i)=\overline{P}(\Pi_i)=P(\Pi_i)$.
We then have from~\eqref{eq:lowupprev}:
$$
P(\Pi_i)= \inf\limits_{\rho \in \mdesirs}Tr(\Pi_i^\dagger \rho)=\sup\limits_{\rho \in \mdesirs}Tr(\Pi_i^\dagger \rho)=Tr(\Pi_i^\dagger \hat{\rho})=Tr(\Pi_i \hat{\rho} \Pi_i),
$$
where we have exploited $\Pi_i=\Pi_i^\dagger=(\Pi_i)^2$.
${P}(\Pi_i)$ is Alice's fair buying and selling price for the event $\Pi_i$ and   is therefore 
related to her probability that such an event happens:
$$
p_i:={P}(\Pi_i)=Tr(\Pi_i \hat{\rho} \Pi_i).
$$
We have therefore derived Born's rule. More, in general, when $  \mdesirs$ is not a singleton, 
$\underline{P}(\Pi_i),\overline{P}(\Pi_i)$ respectively determine the lower and upper probability
of the event indicated by $\Pi_i$.
\end{MethodB}
   \end{SnugshadeB}

\begin{SnugshadeE}
 \begin{example}[Classical coin.]
 \label{ex:coinD}
 Let us go back to our coin in Example \ref{ex:coinC} and compute the coherent lower previsions.
 
In the state of full ignorance, for any $diag(g) \in \CH$
  $$
  \underline{P}(g)= \inf_{p } \sum_{i=1}^2 g_ip_i =\min_i g_i.
  $$
  The infimum is taken w.r.t.\ all probability mass functions  and, therefore,
  the worst case is a probability mass function that puts all the mass in 
  the minimum element of the vector $g \in \reals^2$.
  Clearly, we have that  $  \overline{P}(G)=\max_i g_i$.
  In a state of full ignorance, Alice's expected values for  the payoff of any gamble $g$
  is included in $[\min g,\max g]$.  

  In the  second case, for any $diag(g) \in \CH$
  $$
  \underline{P}(G)= \inf_{p=(0.2,0.8)w+(1-w)(0.6,0.4) \text{ with } w\in [0,1] } \sum_{i=1}^2 g_ip_i=\min\left(g_10.2+g_20.8,g_10.6+g_20.4\right),
  $$
  which means that the lower prevision (expectation) of $G$ is the minimum between the expectations
  of $g_10.2+g_20.8$ and $g_10.6+g_20.4$.  $\overline{P}(G)$ can be obtained by computing the maximum.
  
  In the  third case, for any $diag(g) \in \CH$
  $$
  \underline{P}(G)=\overline{P}(G)={P}(G)=  \frac{1}{2}g_1+\frac{1}{2}g_2,
  $$
  the expectation of  $g$ w.r.t.\ a fair coin.
 \end{example}
\end{SnugshadeE}
In the above example, we have assessed that when $\mdesirs$ includes all probability mass functions, Alice is in a state of ignorance about the probability of Head and Tail of the coin; whence we have that $  \underline{P}(diag(g))=\min_i g_i$ and $  \overline{P}(diag(g))=\max_i g_i$.
This means that Alice only knows that she will receive/pay a quantity belonging to  the interval $[\min_i g_i,\max_i g_i]$. That is why
she only accepts positive gambles in a state of complete ignorance. A similar result holds for  QM.
 Indeed, assume that Alice is in a state of complete ignorance about the density matrix associated to the quantum system, i.e., $\mdesirs=\CD$.
 Then it holds that
   $$
  \underline{P}(G)= \inf_{\rho \in \CD } Tr(G^\dagger \rho)=\lambda_{min}(G), ~~  \overline{P}(G)= \sup_{\rho \in \CD } Tr(G^\dagger \rho)=\lambda_{max}(G),
  $$
  where $\lambda_{min}(G)$ (resp. $\lambda_{max}(G)$) is the minimum (resp. maximum) eigenvalue of $G$.
This is proven in the Appendix, Proposition \ref{pro:eigenvalue}. 

\begin{SnugshadeE}
 \begin{example}[Quantum coin.]
\label{ex:spinD}
Let us go back to the qubit of Example \ref{ex:spinC} and compute the coherent lower and upper previsions.

In the state of full ignorance, for any $G \in \CH$
  $$
  \underline{P}(G)= \inf_{\rho \in \CD } Tr(G^\dagger \rho)=\lambda_{min}(G),
  $$
  where $\lambda_{min}(G)$ is the minimum eigenvalue of $G$. The infimum is taken w.r.t.\ all density matrices   and, therefore,
  the worst case is a density matrix that  puts ``all the mass'' in 
  the minimum eigenvalue of $G$ (the density matrix that achieves the minimum is equal to the projector corresponding
  to the eigenvector relative to $\lambda_{min}(G)$).
  Clearly, we have that $ \overline{P}(G)=\lambda_{max}(G)$.
  In a state of full ignorance, Alice's expected values for  the payoff of any gamble $G$
  is included in $[\lambda_{min}(G),\lambda_{max}(G)]$, as proven in Proposition~\ref{pro:eigenvalue}.

  In the  second case, for any $G \in \CH$
  $$
  \underline{P}(G)= \inf_{\rho \in \mdesirs_2 } Tr(G^\dagger \rho)=\min\left(Tr(G^\dagger D_1),Tr(G^\dagger D_2)\right),
  $$
  while $\overline{P}(G)$ can be obtained  computing the maximum.
  
  In the  third case, for any $G \in \CH$,
  $$
  \underline{P}(G)=\overline{P}(G)={P}(G)= Tr(G^\dagger D),
  $$
  which is the expectation of $G$ w.r.t.\ the density matrix $D$.
  In this case, we have a that $\underline{P}(G)=\overline{P}(G)$, because Alice is sure about the density matrix.
 \end{example}

\end{SnugshadeE}

Finally we will mention that it is possible to re-derive the first axiom of QM without explicitly using duality. 
This can be obtained by defining lower and upper coherent previsions first and then using Gleason's theorem \cite{gleason1957measures}
to relate coherent  previsions to density matrices. 

  \begin{SnugshadeB}
\begin{MethodB}[Gleason's theorem:] suppose $\mathcal{H}$ is a separable Hilbert space of complex dimension at least $n\geq 3$. Then for any quantum 
probability measure on the lattice $Q$ of self-adjoint projection operators on $\mathcal{H}$ there exists a unique trace class 
operator $\rho$ such that $p_i=P(\Pi_i) = Tr(\Pi_i \rho)$ for any self-adjoint projection $\Pi_i$ in $Q$. 
\end{MethodB}
\end{SnugshadeB}

Now, let us assume that Alice's SDG $\mathcal{K}$ is maximal, given that Gleason's theorem holds only for a single density matrix. Then we compute the lower (and upper) prevision induced by $\mathcal{K}$ through~\eqref{eq:lower}.
Since $\mathcal{K}$ is maximal, it holds that $  \underline{P}(G)=\overline{P}(G)=P(G)$.
Since $P$ is linear, $P(G)$ is completely defined by its value on the projectors: writing $G=\sum_{i=1}^n \lambda_i \Pi_i$,  then $P(G)=\sum_{i=1}^n \lambda_i P(\Pi_i)$. Then for $n\geq 3$, we can invoke Gleason's theorem and conclude that there exists a unique trace class 
operator $\hat{\rho}$ such that $p_i=P(\Pi_i) = Tr(\Pi_i \hat{\rho})$. Therefore, $\hat{\rho}$ is the dual
of Alice's maximal SDG. Since $\hat{\rho}$ is unique, it  uniquely defines $\mathcal{K}$ ($n\geq 3$).

This approach is similar to that developed by \citet{pitowsky2003betting,Pitowsky2006}.
The main difference is that the probabilities $P(\Pi_i)$ for us are derived quantities, defined from $\domain$.
This means that we start from a coherent set of desirable gambles $\domain$
and we derive the probabilities $P(\Pi_i)$ that it induces on the  projectors $\Pi_i$.
Moreover, we can immediately state that these probabilities are coherent (self-consistent) because they are induced by a coherent set of desirable gambles.

\subsection{Interpretation of the first axiom in terms of desirability and consequences}
In this section we have derived the first axiom of QM as a direct consequence of rationality, via duality.
By duality, we can then conclude that:
 \begin{enumerate}
  \item A closed convex set of density matrices represents Alice's personal knowledge about the quantum experiment, that is, her beliefs  about the possible results of the experiment. A state of full ignorance is represented by the set of all density matrices; a state
 of maximal personal knowledge is represented by a single density matrix and any intermediate degree of beliefs is represented by a closed convex set
 of density matrices. Quantum states (density matrices) are therefore personal judgements of a subject 
 about the outcome of a quantum experiment (as standard probabilities are
 personal judgements of a subject about the outcome of a classical uncertain experiment).
 \item Closed convex sets of density matrices are coherent -- in the same way as standard probabilities are coherent. This means that a subject who states her knowledge about the quantum experiment in terms of a closed convex set of density matrices cannot be made a \textbf{partial loser} in a gambling system based on them: the dual set is a coherent set of desirable gambles and so it avoids partial loss.
 \item We have shown that coherent sets of desirable gambles over $\mathbb{C}_h^{n \times n}$ stand to \textbf{density matrices}
as coherent sets of desirable gambles over $\mathbb{R}^n$ stand to \textbf{probabilities}.
The consequence is that the \textbf{first axiom of QM }on density matrices on $\mathbb{C}_h^{n \times n}$ is structurally and formally equivalent to 
\textbf{Kolmogorov's first and second axioms} about probabilities on $\mathbb{R}^n$. In fact, they can be both derived via duality
from a coherent set of desirable gambles.
 \end{enumerate}
 
 
\section{Second axiom of QM}
\label{sec:2ndaxiom}
This section shows how to derive a coherent rule for updating beliefs. In particular our aim is to answer this question: how should Alice change her gambling assessments in the prospect of obtaining new information in the form of an event? Answering this question will lead us to derive the second axiom of QM. The second axiom states:
   \begin{SnugshadeF}
\begin{MethodF}[Axiom II:]
quantum projection measurements are described by a collection  
$\{\Pi_i\}_{i=1}^n$  of
projection operators that satisfy the completeness equation $\sum_{i=1}^n \Pi_i =I$. These are operators acting on the 
state space of the
system being measured. If the state of the quantum system is $\rho$ immediately
before the measurement then  the state after the measurement is
$$
\hat{\rho}=\dfrac{\Pi_i \rho \Pi_i}{Tr(\Pi_i \rho \Pi_i)},
$$
provided that $Tr(\Pi_i \rho \Pi_i)>0$ and the probability that result $i$ occurs is given by
$p_i=Tr(\Pi_i \rho \Pi_i)$.
\end{MethodF}
\end{SnugshadeF}

\subsection{Coherent updating}
\label{sec:updating}

We initially assume that Alice considers an event ``indicated'' by a certain projector $\Pi_i$
in $\Pi=\{\Pi_i\}_{i=1}^n$. The information it represents is: an experiment $\Pi$ is performed and the event indicated by $\Pi_i$ happens.\footnote{{
We assume  that the quantum measurement device is a ``perfect meter'' (an ideal common  assumption in QM), i.e., there are not observational errors -- Alice can trust the
received information.}} Under this assumption, Alice can focus on gambles that are contingent on the event $\Pi_i$: these are the gambles such that ``outside'' $\Pi_i$ no utile is received or due -- status quo is maintained --; in other words, they represent gambles that are called off if the outcome of the experiment is not $\Pi_i$. 
Mathematically, these  gambles are of the form
$$
G=\left\{\begin{array}{ll}
  H &  \text{if } \Pi_i \text{ occurs},\\
  0 &  \text{if } \Pi_j \text{ occurs, with}  ~j\neq i.\\
  \end{array}\right.
$$
 It is then clear that $H=\alpha\Pi_i$ with $\alpha \in \reals$ since  $\Pi_i\Pi_j=0$ for each $j\neq i$.

In this light, we should define Alice's conditional set of strictly desirable gambles by simply restricting the attention to gambles of the form $\Pi_i G \Pi_i=\alpha\Pi_i$ with $G \in \CH$.

\begin{definition}
 let  $\domain$ be an SDG, the set obtained as
\begin{equation}
\label{eq:condition}
\domain_{\Pi_i}=\left\{G \in \CH \mid  G \gneq0 \textit{ or }\Pi_i G \Pi_i \in \domain \right\}
\end{equation} 
is therefore called the {\bf set of desirable gambles conditional} on  $\Pi_i$.\footnote{This is structurally the way conditioning is defined also in the classical theory of desirable gambles, see \citet[Ch.~1]{augustin2014}.} \end{definition}

The conditional set includes  all
the PSDNZ matrices (they are always desirable for Alice) and all the gambles $G$ such that $\Pi_i G \Pi_i \in \domain$.
 
In the Appendix, Proposition \ref{prop:conditioning}, we show that Alice's set of desirable gambles conditional on  $\Pi_i$,
that is $\domain_{\Pi_i}$, is still a
non-pointed convex cone that accepts partial gains and satisfies the openness property. In other words,
the conditioning operation transforms SDGs into SDGs and, therefore, preserves coherence.

We can also compute the dual of  $\domain_{\Pi_i}$, i.e.,
$\mdesirs_{\Pi_i}$ -- we call it a \textbf{conditional quantum credal set} --, by applying the duality
transformation~\eqref{eq:credal1}. It is useful
to investigate what is the relationship between $\mdesirs_{\Pi_i}$
and the unconditional quantum credal set $\mdesirs$ that is the dual of $\domain$.
The following diagram gives the relationships among  $\domain,\mdesirs, \domain_{\Pi_i},\mdesirs_{\Pi_i}$.
$$
\begin{tikzpicture}
  \matrix (m) [matrix of math nodes,row sep=3em,column sep=6em,minimum width=2em]
  {
     \domain & \,\domain_{\Pi_i} \\
     \mdesirs & \mdesirs_{\Pi_i} \\};
  \path[-stealth]
    (m-1-1) edge [<->] node [left] {dual} (m-2-1)
            edge [double] node [below] {conditioning} (m-1-2)
    (m-2-1.east|-m-2-2) edge [double]  node [below] {conditioning}
          (m-2-2)
    (m-2-2) edge [<->] node [right] {dual} (m-1-2);
\end{tikzpicture}
$$
Let us show how $\mdesirs_{\Pi_i}$ can be derived by $\mdesirs$.

 \begin{SnugshadeB}
\begin{MethodB}[Subjective formulation of the second axiom of QM]\\
Given a quantum credal set $\mdesirs$, the corresponding quantum credal set conditional on $\Pi_i$ is obtained as
\begin{equation}
\label{eq:rhobayes}
 \mdesirs_{\Pi_i}=\left\{\dfrac{\Pi_i \rho \Pi_i}{Tr(\Pi_i \rho \Pi_i)} \Big|  \rho \in 
\mdesirs\right\},
\end{equation} 
 provided that $Tr(\Pi_i \rho \Pi_i)>0$  for every $\rho \in \mdesirs$. Note that the latter  condition  implies that $\Pi_j \notin \mdesirs$
 for any $j\neq i$.
\end{MethodB}
\end{SnugshadeB}
In fact, let us consider the set $\left\{ \frac{\Pi_i \rho \Pi_i}{Tr(\Pi_i \rho \Pi_i)} \mid  \rho \in \mdesirs \right\}$, 
which is a well-defined closed convex set of density matrices.
From~\eqref{eq:credal2}, its dual SDG can be obtained as
$$
\begin{array}{l}
\left\{G \in \CH \mid G\gneq 0\} \cup\{ G \in \CH \mid G \cdot \frac{\Pi_i \rho \Pi_i}{Tr(\Pi_i \rho \Pi_i)}>0 \forall \rho \in \mdesirs \right\}\\
=\left\{G \in \CH \mid G\gneq 0\} \cup\{ G \in \CH \mid G \cdot (\Pi_i \rho \Pi_i)>0, \forall \rho \in \mdesirs \right\},\end{array}
$$
where the denominator $Tr(\Pi_i \rho \Pi_i)$ has been neglected because it is always strictly positive by hypothesis.
Since
$$
 G \cdot (\Pi_i \rho 
\Pi_i)= Tr(G^\dagger \Pi_i \rho \Pi_i)= Tr((\Pi_i G \Pi_i)^\dagger  \rho)= (\Pi_i G \Pi_i) \cdot \rho ,
$$
we can rewrite the dual as
 $$
 \left\{G \in \CH \mid G\gneq 0\right\} \cup \left\{G \in \CH \mid   (\Pi_i G \Pi_i) \cdot \rho>0, \forall \rho \in \mdesirs \right\}.
 $$
By definition of $\mdesirs$, this set coincides with~\eqref{eq:condition}, which is the  definition of $\domain_{\Pi_i}$. Showing the inverse (from $\domain_{\Pi_i}$ to $\mdesirs_{\Pi_i}$) is analogous.

In summary, we have derived the second axiom of QM from the usual definition of conditioning in desirability, once it is extended to the quantum case.

\begin{remark}
 We have called ${\domain}_{\Pi_i}$ a conditional set of desirable gambles because it is the set of desirable gambles conditional on $\Pi_i$. Note also that the last part of the axiom, that is,
$p_i=Tr(\Pi_i \rho \Pi_i)$ -- the probability that $\Pi_i$ occurs -- has been derived in Section \ref{sec:prevision}.
\end{remark}

\begin{SnugshadeE}
 \begin{example}[Classical coin.]
 \label{ex:coinE}
 Let us going back to the classical coin example.  Alice considers the case that the coin has landed Head up and wonders how she should consequently change her beliefs about the experiment. Even without doing any calculation, it is clear that Alice should only accept the gambles that give a positive payoff for Head, no matter the payoff for Tail.
 
Let us show that we can reach the same conclusion by applying the calculus descrived before. The conditional event is \textit{Head}, represented
 by the projector $\Pi_1=e_1e_1^T$. To see the gambles that  Alice is willing to  accept based on this information, we compute the conditional SDG.

For the first SDG, $ \domain^{(1)}=\left\{G \in \CH \text{ diagonal } \mid G \gneq0\right\}$,
    the conditional SDG is
  $$
  \begin{array}{l}
  \domain^{(1)}_{\Pi_1}=\left\{G \in \CH \text{ diag. } \mid G \gneq0 \textit{ or }\Pi_1 G\Pi_1 \in    \domain^{(1)}\right\}=\left\{G \in \CH \text{ diag. } \mid g\gneq0 \textit{ or } g_{1}>0 \right\}.
  \end{array}
  $$
The interpretation of this result is straightforward. Given the event ``the coin has landed head up'', Alice accepts all gambles $G\gneq0 $ as well as all the gambles $G=diag(g_1,g_2)$ such that   $g_1>0$ (the gambles with positive payoff for Head, no matter the value of $g_2$).

For the second and third SDGs, we derive again that 
  $$
  \begin{array}{rcl}
  \domain^{(2)}_{\Pi_1}=\domain^{(3)}_{\Pi_1}=\left\{G \in \CH \text{ diagonal } \mid g\gneq0 \textit{ or } g_{1}>0 \right\},
  \end{array}
  $$
  consistently to what discussed before.
 
Therefore, in all the cases, the conditional SDG is the green region in Figure~\ref{fig:coincondition}.
By applying~\eqref{eq:credal1} to Alice's conditional SDG, it can easily  be derived that the corresponding conditional set of probabilities
$\mdesirs_{\Pi_1}$  includes only the probability mass function $p=(1,0)$.
In fact, let us consider the third case  $\domain^{(3)}_{\Pi_1}$ (SDG relative to the fair coin). In the unconditional case, Alice's set of density matrices includes only $\rho_3=diag(0.5,0.5)$, i.e., she believes that the coin is fair.
From~\eqref{eq:rhobayes}, we derive that  her conditional set of density matrices is
$$
\dfrac{\Pi_1 \rho_3 \Pi_1}{Tr(\Pi_i \rho_3 \Pi_i)}=\dfrac{\left[\begin{matrix}
                                                           0.5 & 0\\
                                                           0 & 0
                                                          \end{matrix}\right]
}{0.5}=\left[\begin{matrix}
                                                           1 & 0\\
                                                           0 & 0
                                                          \end{matrix}\right],
$$
whose diagonal is $p=(1,0)$. This is just Bayes' rule. We are simply applying Bayes' rule to the density matrices
(in this case probability mass functions) in $\mdesirs$.
Under the assumption that ``the coin has landed head up'', Alice's knowledge about the coin experiment ``has collapsed'' to  $p=[1,0]$ -- she knows
that the result of the experiment is Head. 
Moreover, this conclusion is the same in the three cases. If we think about the three cases as associated to three different people, Alice A, Alice B and Alice C, then all Alices' conditional SDGs
coincide. This shows that, in spite of the fact that we are considering subjective information, the \textbf{different subjects may reach the same conclusion conditional on some evidence}.

\centering
\begin{minipage}{0.30\textwidth}
\begin{tikzpicture}[scale=1.5]
    \draw[->] (-1,0) coordinate (xl) -- (1,0) coordinate (xu) node[right] {$g_1$};
    \draw[->] (0,-1) coordinate (yl) -- (0,1) coordinate (yu) node[above] {$g_2$};
        \draw (0,0) circle (1.pt);
                        \draw[->,black!30!green,thick] (0,0) coordinate (xl) -- (0,1) coordinate (xu);
    \begin{pgfonlayer}{background}
      \draw[border] (0,0) -- (0,1) coordinate (a1away);
      \draw[border] (0,0) -- (1,0) coordinate (a2away);
      \fill[blue!50,  opacity=0.5] (0,0) -- (1,0) -| (1,1) --(0,1);
    \end{pgfonlayer}
        \begin{pgfonlayer}{background}
            \fill[green!50,  opacity=0.5] (0,-1) -- (0,1) -| (1,1) --(1,-1);
    \end{pgfonlayer}
  \end{tikzpicture}\captionof{subfigure}{}
  \end{minipage}   
  \begin{minipage}{0.30\textwidth}
      \begin{tikzpicture}[scale=1.5]
    \draw[->] (-1,0) coordinate (xl) -- (1,0) coordinate (xu) node[right] {$g_1$};
    \draw[->] (0,-1) coordinate (yl) -- (0,1) coordinate (yu) node[above] {$g_2$};
  \draw[->,black!30!green,thick] (0,0) coordinate (xl) -- (0,1) coordinate (xu);
    \draw (0,0) circle (1.pt);
    \begin{pgfonlayer}{background}
      \draw[dashed] (0,0) -- (-0.66,1) coordinate (a1away);
      \draw[dashed] (0,0) -- (1,-0.25) coordinate (a2away);
      \fill[blue!50,  opacity=0.5] (0,0) -- (a1away) -| (a2away) --(0,0);
    \end{pgfonlayer}
            \begin{pgfonlayer}{background}
            \fill[green!50,  opacity=0.5] (0,-1) -- (0,1) -| (1,1) --(1,-1);
    \end{pgfonlayer}
  \end{tikzpicture}\captionof{subfigure}{}
  \end{minipage} 
    \begin{minipage}{0.30\textwidth}
      \begin{tikzpicture}[scale=1.5]
    \draw[->] (-1,0) coordinate (xl) -- (1,0) coordinate (xu) node[above] {$g_1$};
    \draw[->] (0,-1) coordinate (yl) -- (0,1) coordinate (yu) node[above] {$g_2$};
  \draw[->,black!30!green,thick] (0,0) coordinate (xl) -- (0,1) coordinate (xu);
    \draw (0,0) circle (1.pt);
    \begin{pgfonlayer}{background}
      \draw[dashed] (0,0) -- (-1,1) coordinate (a1away);
      \draw[dashed] (0,0) -- (1,-1) coordinate (a2away);
      \fill[blue!50,  opacity=0.5] (0,0) -- (a1away) -| (a2away) --(0,0);
    \end{pgfonlayer}
            \begin{pgfonlayer}{background}
            \fill[green!50,  opacity=0.5] (0,-1) -- (0,1) -| (1,1) --(1,-1);
    \end{pgfonlayer}
  \end{tikzpicture}\captionof{subfigure}{}
  \end{minipage}   
\captionof{figure}{Alices' conditional sets of coherent strictly desirable gambles corresponding to three different 
degrees of beliefs about the  classical coin.  \label{fig:coincondition}}
 \end{example}
\end{SnugshadeE}

\begin{SnugshadeE}
\begin{example}[Quantum coin.]
\label{ex:spinE}
Let us go back to the qubit of Example \ref{ex:spinC}. In this case, Alice considers the information ``the measurement $\Pi$ has been performed and the event
indicated by the projector 
$$
\Pi_1=\frac{1}{2}\left[\begin{array}{cc}
    1 & -\iota\\\iota &  1
    \end{array}\right]
$$
has  been observed''.  In this case, we can immediately conclude that Alice knows the outcome of the experiment (detection along the direction $\Pi_1$) -- there is not uncertainty. As a consequence, Alice will only accept the gambles that give a positive payoff
 w.r.t.\ $\Pi_1$, no matter the payoff for $\Pi_2$. Our goal is to compute Alice's conditional SDG in the three cases.

In the state of full ignorance, Alice's SDG is $\domain^{(1)}=\{G \in \CH \mid G\gneq0\}$ and after conditioning it becomes
$$
\domain^{(1)}_{\Pi_1}=\{G \in \CH \mid G \gneq0 \textit{ or } \Pi_1G\Pi_1=\gamma_1\Pi_1 \in \domain^{(1)}\}.
$$
The last condition $\gamma_1\Pi_1 \in \domain^{(1)}$ can also be written as $\gamma_1>0$: apart from the positive gambles, Alice only accepts those gambles that give positive payoffs along $\Pi_1$.
By Pauli's matrix decomposition of $G$, we obtain that $\Pi_1G\Pi_1=\gamma_1\Pi_1 $, with $\gamma_1=v + y$, whence the resulting SDG is
$$
\domain^{(1)}_{\Pi_1}=\{(x,y,z,v) \in \reals^4 \mid v + y>0  \}.
$$
A slice of this set is shown in Figure~\ref{fig:conditionalcones}(a) together with the  ice-cream cone -- Alice's
SDG before conditioning. Alice's SDG is now the hyperspace $v + y>0$ (the oblique plane is $v + y=0$).

For the second SDG, 
$$
\begin{array}{rcl}
\domain^{(2)}=\{G \in \CH \mid G \gneq0\} \cup \{G \in \CH \mid Tr(G^\dagger D_1)>0 \text{ and } Tr(G^\dagger D_2)>0\}
\end{array}
$$
with $D_1=diag(0.2,0.8)$ and $D_2=diag(0.6,0.4)$; the conditional SDG is
$$
\domain^{(2)}_{\Pi_1}=\{G \in \CH \mid G \gneq0 \textit{ or }  \Pi_1G\Pi_1=\gamma_1\Pi_1 \in \domain^{(2)}\}. 
$$
Since $Tr(\Pi_1G^\dagger\Pi_1 D_i)=Tr(G^\dagger\Pi_1 D_i\Pi_1)$, and based on Pauli's matrix decomposition of $G$, these constraints become:
$$
\begin{array}{rcl}
 \domain^{(2)}_{\Pi_1}&=&\{(x,y,z,v) \in \reals^4 \mid v \neq 0 \text{ and } \sqrt{x^2+y^2+z^2}\leq v\}\cup \{(x,y,z,v) \in \reals^4 \mid v + y>0 \}\\
&=&\{(x,y,z,v) \in \reals^4 \mid v + y>0 \},
\end{array}
$$
where we have exploited that $Tr(G^\dagger\Pi_1 D_i\Pi_1)\propto v+y$.

In the third case, Alice's SDG is
$$
\domain^{(3)}=\{G \in \CH \mid G \gneq0\} \cup\{G \in \CH \mid Tr(G^\dagger D)>0 \}, \textit{ with } D=\frac{1}{2}\left[\begin{array}{cc}
    1 & -\iota\\\iota &  1
    \end{array}\right],
$$
and after conditioning $\domain^{(3)}_{\Pi_1}=\{G \in \CH \mid G \gneq0 \textit{ or }  \Pi_1G\Pi_1=\gamma_1\Pi_1 \in \domain^{(3)}\}$,
which is equivalent to
$$
\domain^{(3)}_{\Pi_1}=\{G \in \CH \mid G \gneq 0\} \cup\ \{G \in \CH \mid Tr(\Pi_1G^\dagger\Pi_1 D)>0\}.
$$
Note that $Tr(\Pi_1G^\dagger\Pi_1 D)=Tr(G^\dagger\Pi_1 D\Pi_1)=Tr(G^\dagger D)$, where we have exploited that $D=\Pi_1$, whence $\domain^{(3)}_{\Pi_1}=\{G \in \CH \mid G \gneq0\} \cup \{G \in \CH \mid Tr(G^\dagger D)>0\}$.
Again by Pauli's matrix decomposition of $G$, these constraints become:
$$
\begin{array}{rcl}
\domain^{(3)}_{\Pi_1}&=&\{(x,y,z,v) \in \reals^4 \mid v\neq 0 \text{ and }\sqrt{x^2+y^2+z^2}\leq v\}\cup \{(x,y,z,v) \in \reals^4 \mid v + y>0 \}\\
&=& \{(x,y,z,v) \in \reals^4 \mid v + y>0 \}. \end{array}
$$

In other words, Alice's state of knowledge about the quantum coin experiment ``has collapsed'' to  $\Pi_1$ -- she knows
that the result of the experiment is ``Vertical'' along $\Pi_1$.
Therefore Alice is going to only accept gambles that give positive payoffs contingently to this information. 
Moreover, this conclusion is the same in the three cases, as in the case of the classical coin.

\centering
\includegraphics[width=4cm]{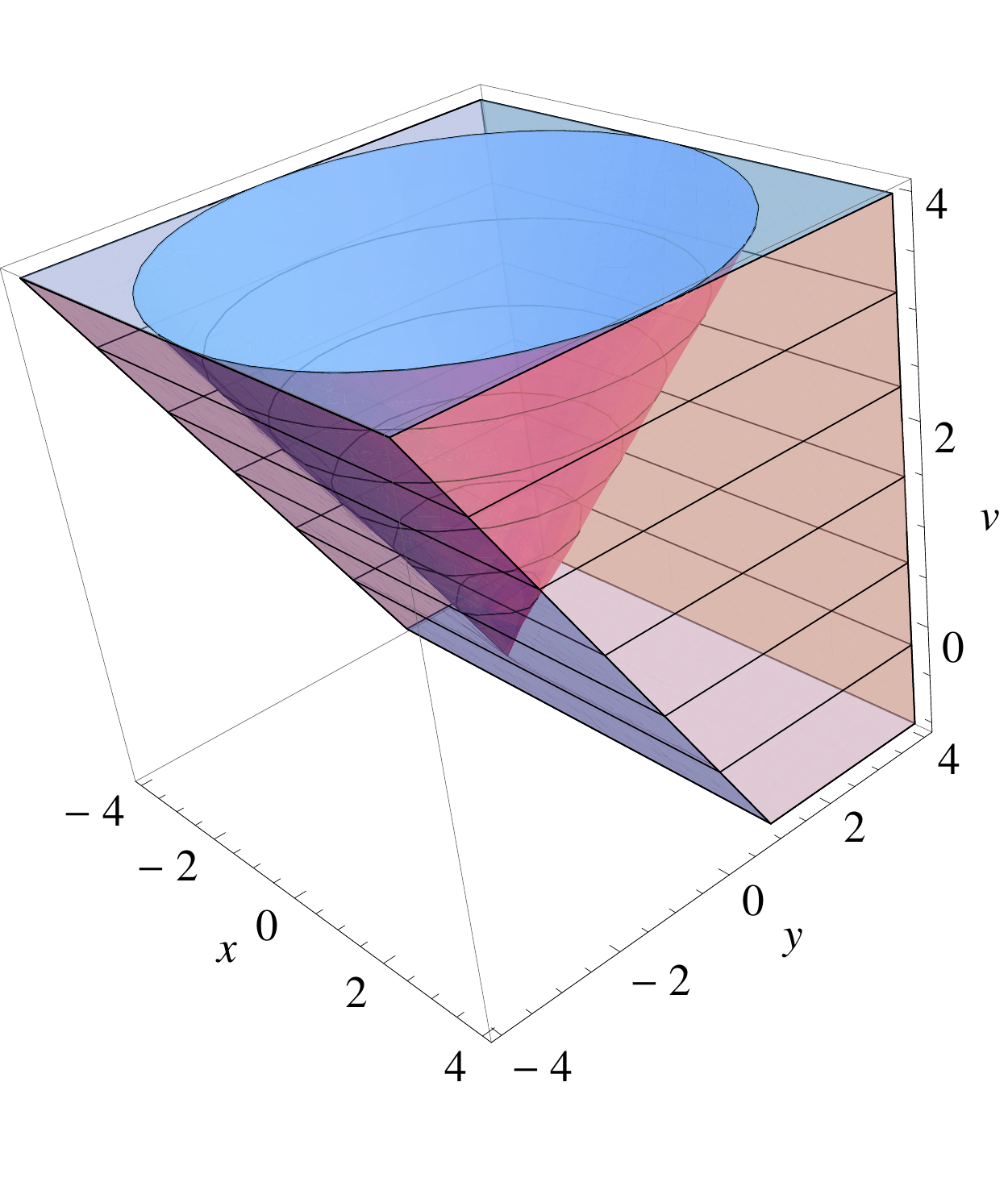}
\captionof{figure}{Alices'  set of coherent conditional SDG  (the hyperspace bounded by the plane that is excluded from Alice's conditional SDG) corresponding to three different 
degrees of beliefs about the  quantum  state.  \label{fig:conditionalcones}}

\end{example}
\end{SnugshadeE}

\subsection{Conditioning for non-elementary events}

So far we have defined conditioning only based on $\Pi_i$, that is, when Alice considers
that the event indicated by $\Pi_i$ has happened. 
In analogy with probability
we call the event indicated by $\Pi_i$ an \textit{elementary event}.
In QM, it is called a \textit{selective measurement}.
We are also interested in non-elementary events, such as ``the event indicated
by $\Pi_i$ or $\Pi_k$ has happened'' or, more in general,
``the event indicated by one of the projectors in the set $\{\Pi_j\}_{j\in J}$
has happened with $J\subseteq \{1,\dots,n\}$''. In the classical case, these events
would for instance be: ``the dice has landed $1$ or $2$'' or ``the outcome of the dice is even''.
Now we set out to define a conditioning rule for Alice's SDG for non-elementary events.

Similarly to before, Alice should focus on gambles that are contingent on one of the events in $\{\Pi_j\}_{j\in J}$.
Conditioning for non-elementary events is then defined as
$$
\domain_{\Pi_J}=\left\{G \in \CH \Bigg| G \gneq 0 \textit{ or }\sum_{j\in J} \Pi_j G  \Pi_j\in \domain 
\right\}
$$
for some subset $J$ of the indexes  $\{1,\dots,n\}$. 
The set
 $\domain_{\Pi_J}$ is again an SDG
(Appendix, Corollary
\ref{cor:nonelementary}).

Note that in general
$$
\sum_{j=1}^n \Pi_j G  \Pi_j \neq G.
$$
This means that in case Alice considers that one of the events $\Pi_i$ has happened without knowing which one in particular (this is called
a \textit{non-selective measurement} in QM), her SDG after taking into account this information is different from $\domain$.

This marks an important difference between the classical and the quantum theory of desirable gambles. 
In the classical case, we always have that $\sum_{j=1}^n \Pi_j G  \Pi_j =G$, taking into account that $G$ is diagonal and $\Pi_j=e_je_j^T$, where $e_j$ is an element of the canonical basis of $\reals^n$ --
as discussed in the classical coin examples. For instance, if in the dice example Alice assumes that the outcome of the dice thrown is a number between $1$ and $6$, then she does not change her set of desirable gambles because
this information does not change her knowledge about the outcome of the experiment.
In QM, it is different: $\{\Pi_j\}_{j=1}^n$ is informative evidence. 

We can also compute the conditional quantum credal set $\mdesirs_{\Pi_J}$ by applying the duality
transformation~\eqref{eq:credal1}. In this case too, 
the relationship between $\mdesirs_{\Pi_J}$
and the unconditional quantum credal set $\mdesirs$, which is the dual of $\domain$, can be obtained as 
\begin{equation}
\label{eq:rhobayesJ}
 \mdesirs_{\Pi_J}=\left\{\dfrac{\sum_{j\in J} \Pi_j \rho \Pi_j}{\sum_{j\in J} Tr(\Pi_j \rho \Pi_j)} \Big|  \rho \in 
\mdesirs\right\},
\end{equation} 
 provided that $\sum_{j=1}^n Tr(\Pi_j \rho \Pi_j)>0$  for every $\rho \in \mdesirs$.\footnote{
 Also in this case it should be observed that in general $\sum_{j=1}^n \Pi_j \rho \Pi_j\neq \rho$
 and, therefore, QM differers from classical probability for which we have that: $\sum_{j=1}^n p(y|x_i)p(x_i)=p(y)$, which is called the law of total probability. A violation of the law of total probability implies incoherence in classical probability. However, we have shown that the updating rule~\eqref{eq:rhobayesJ}
 is coherent, as it transforms SDG into SDG. This means that there is no violation of coherence in QM
 by using~\eqref{eq:rhobayesJ}.}
 The proof is similar to that for~\eqref{eq:rhobayes}.
 
 In QM, there are other type of ``measurements'' apart from the ones determined by projectors,
 such as \textit{Positive Operator Valued Measures} (POVM) and variants such as SIC-POVM,
 see for instance \citet{nielsen2010quantum}.  However POVM can be seen as projective measurements on a higher dimensional space
 and, therefore, they can be taken into account by the updating rule developed in this section.

\subsection{Conditioning on zero probability events}
As a side note, we can show with an example that the formalisation in terms of gambles is more general than that
in terms of density matrices.
Consider the case in which 
$$
\mdesirs=\{\rho\}=\{\Pi_j\},
$$
that is, the quantum state coincides with a projector, so it is a \emph{pure state}.
Let us assume that Alice considers the information ``the event $\Pi_i$ has been observed''
with $\Pi_i\Pi_j=0$.
Since  $Tr(\Pi_i\rho\Pi_i)=0$, we cannot apply~\eqref{eq:rhobayes}. Considered that $p_i=Tr(\Pi_i\rho\Pi_i)$ is the probability of observing $\Pi_i$, we have in particular that~\eqref{eq:rhobayes} is not defined for zero probability events -- the second axiom of QM is not defined in this case.

Conversely, the framework of desirable gambles allows us to compute the resulting set $\mdesirs_{\Pi}$
also in this case.
First, we start by defining the dual of $\mdesirs=\{\rho\}=\{\Pi_j\}$: by~\eqref{eq:credal2}, it is
\begin{equation}
\label{eq:condzero0}
\domain=\left\{G\in \CH \mid G\gneq0 \right\}\cup\left\{G\in \CH \mid G\cdot \rho >0, \forall \rho\in 
\mdesirs\right\}.
\end{equation} 
We can now compute $\domain_{\Pi_i}$, i.e., Alice's conditional SDG:
\begin{equation}
\label{eq:condzero}
\begin{array}{rcl}
\domain_{\Pi_i}&=&\left\{G\in \CH \mid G\gneq0 \textit{ or }\Pi_iG\Pi_i \in \domain\right\}\\
&=&\left\{G\in \CH \mid G\gneq0 \right\}\cup\left\{G\in \CH \mid \Pi_iG\Pi_i\cdot \rho >0, \forall \rho\in 
\mdesirs=\{\Pi_j\}\right\}.\end{array}
\end{equation} 
Now observe that
$$
\Pi_iG\Pi_i\cdot \rho =Tr(G^\dagger\Pi_i\rho \Pi_i)=0,
$$
and so the second set in~\eqref{eq:condzero} is empty, whence
$$
\domain_{\Pi_i}=\left\{G\in \CH \mid G  \gneq0\right\}.
$$
In this case, Alice only accepts gambles such that $G\gneq0$, i.e.,
she is in a state of full ignorance about the experiment. Hence, quantum desirability naturally addresses a ``singularity'' in second axiom of QM within the desirability framework itself rather than resorting to undefined quantities. 

The question of conditioning on zero probability events has a long history in probability, it is for instance at the heart of Borel-Kolmogorov's paradox -- see, e.g.,  \citet[Sec. 15.7]{jaynes2003probability}. It falls in the area of non-Archimedean problems. These problems exhibit a kind of discontinuity for which Archimedean models like probabilities or SDGs are not suited: these models are such that the unconditional beliefs they represent are not informative enough to constrain (by coherence) conditional beliefs in the case of zero probability events. It is precisely for this reason that $\domain_{\Pi_i}$ is vacuous. On the technical side, this happens because the information needed to constrain conditional beliefs under zero-probability events has to lie in the border of the cone of desirable gambles: and the border is not part of an SDG by definition. Using coherent sets of desirable gambles that drop the openness condition (S3), and hence are more general than SDGs, allows in fact non-vacuous conditional beliefs to be obtained also in the non-Archimedean case (we do not pursue this path any further as it is out of the scope of the present paper).

\subsection{Strong temporal coherence -- part I}\label{sec:stc1}
So far we have implicitly assumed that Alice has assessed her  beliefs about the quantum experiment at a certain defined time (say, present time, $t_0$). We have taken care not to make any statement about the evolution of Alice's beliefs in time. This may sound strange in that we have considered updated beliefs. However, strictly speaking, also updated beliefs are present-time beliefs: they are Alice's current beliefs under the assumption that the conditional event occurs. They bear no implication on Alice's actual future beliefs \emph{after} the event has occurred.

This is quite a subtle point in probability that goes back to its very foundations with Bayes and De Moivre trying to provide a justification of conditioning as a rule to compute future beliefs through a sort of ``temporal Dutch-book'' argument -- see the historical account portrayed by \citet{shafer1982,shafer1985}. In more recent times, this question appears to have been explicitly recognised first by \citet[Sec.~3, p.~316]{hacking1967}. It has later been deeply studied by a number of authors, in particular by \citet{goldstein1983} in statistics and \citet{fraassen1984} in philosophy. Their standpoint has recently been taken up also in QM by \citet{fuchs2012}. Our standpoint is related but not quite the same; we refer the reader to \citet{zaffalon2013a} for a detailed comparison of the different approaches.

In short, the question is that the evolution of beliefs in time is not part of the axioms of probability: in particular, a coherent cone of gambles expresses Alice's beliefs only at time $t_0$. So we need to impose some extra structure on the belief system if we want to formalise the dynamics of Alice's beliefs in time. Remember that the second axiom of QM explicitly refers to the state of the system \emph{after} the measurement. By doing so it forces us to consider what actually happens in the future. Therefore, if we want to be fully aligned with it in our desirability approach to QM, we need also to talk about future beliefs: and in particular, we have to require that future beliefs, after the measurement, i.e., the observation, equal updated beliefs, as established at present time $t_0$ through conditioning. This requirement is called \emph{strong temporal coherence} by \citet[Sec.~6.4]{zaffalon2013a}, since it indeed represents a strong form of Dutch-book coherence through time: it means that Alice cannot be made a partial loser even by combining assessments that she expresses at different points in time (at the same point in time, coherence was already guaranteed by the desirability axioms). 

Strong temporal coherence is assumed throughout the paper. Yet, we will treat it in a transparent way after completing the discussion of strong temporal coherence in Section~\ref{sec:tc2}.

\subsection{Interpretation of the second axiom as updating}
We have shown that the second Axiom of QM can be recovered as the updating of Alice's belief model. \textbf{It is formally and structurally equivalent to Bayes' rule -- the only actual difference
is the domain used, $\CH$ instead of $\reals^n$}. 

Under this interpretation, the ``collapse of the wave function'' is nothing else but the outcome of updating Alice's 
beliefs. Similarly to what happens in the experiment of the classical coin,
when she learns that the result is Head, her probability that the output of the experiment is Head becomes one and so there is no more uncertainty.
We have also shown that subjects (Alice A, B and C) with different initial beliefs about the possible outcomes of the experiment
can possibly agree on their updated beliefs.


\section{Third axiom of QM}
\label{sec:3ndaxiom}
In the previous sections, we have argued that a notion of coherence in time is needed to fully identify probabilistic conditioning with measurement in QM. Temporal coherence becomes even more fundamental when it comes to derive the temporal axiom of QM:

   \begin{SnugshadeF}
\begin{MethodF}[Axiom III:]
the evolution of a closed quantum system is described by a unitary
transformation. That is, the state $\rho$ of the system at time $t_0$ is related 
to the state
$\rho'$ of the system at time $t_1>t_0$ by a unitary operator $U$ which depends only 
on the
times $t_0$ and $t_1$, $\rho' = U \rho U^\dagger$.
\end{MethodF}
\end{SnugshadeF}

\subsection{Strong temporal coherence -- part II}\label{sec:tc2}
Let us consider the dynamics of Alice's beliefs between present time $t_0$ and future time $t_1$ under the assumption that she receives no information at all during such an interval of time (i.e., we have a closed quantum system). Similarly to the discussion in Section~\ref{sec:stc1}, the focus is on characterising the coherence of Alice's beliefs in this time. 

To this end, we add the following rule to the betting framework of Section~\ref{sec:foundations}:
\begin{enumerate}
\setcounter{enumi}{4}
 \item At time $t_0$, Alice declares which gambles $G_0$ she accepts. Bookie can perform the experiment at any time $t_1\geq t_0$. After having performed the experiment, Bookie determines Alice's payoff by  ``moving time-forward'' Alice's gambles ($G_0 \xhookrightarrow{\phi}G_1$),  using a map $\phi(\cdot,t_1,t_0)$ from $\CH$ to itself that satisfies the following properties: 
 \begin{itemize}
 \item[(i)] $\phi(\cdot,t_0,t_0)$ is the identity map; 
 \item[(ii)]  $\phi(\cdot,t_1,t_0)$ is onto; 
 \item[(iii)] $\phi(\gambles^+,t_1,t_0)=\gambles^+$;
\item[(iv)] $\phi(\cdot,t_1,t_0)$ is linear and constant preserving.
\end{itemize}
 \end{enumerate}
The rationale behind these conditions is the following.

Condition~(i) is obvious: when $t_1=t_0$, nothing changes with respect to what we have discussed in the previous sections. 

Condition~(ii) is a way of stating that Alice's beliefs are only established at present time $t_0$, since
any gamble $G_1$ at time $t_1$ corresponds to an element $G_0$  at time $t_0$. 

Condition~(iii)  means that if Alice is in a state of complete ignorance at time $t_0$, she remains in such a state at any later time: in other words, we are assuming that Alice does not receive any further information from time  $t_0$ to time  $t_1$.

Finally, condition~(iv) states once again that the utility scale is linear, that is, if Alice finds $G_0,H_0$ desirable at time $t_0$, then we know from Section \ref{sec:intro-clp} that she also finds $G_0+H_0$ desirable. Because of linearity of   $\phi(\cdot,t_1,t_0)$, this implies that $G_1+H_1=\phi(G_0+H_0,t_1,t_0)$, with $G_1=\phi(G_0,t_1,t_0)$ and $H_1=\phi(H_0,t_1,t_0)$, is desirable at time $t_1$.  Finally, $\phi(\cdot,t_1,t_0)$ is constant preserving.

Rule~5 leads to the following (the proof is in the Appendix, Theorem \ref{th:temp}):
   \begin{SnugshadeB}
\begin{MethodB}[Subjective formulation of the third axiom of QM]\\
(1) All the transformations $\phi(\cdot,t_1,t_0)$ defined above
are of the following form
$$
G  \xhookrightarrow{\phi} H=U^\dagger G U ,
$$
for some unitary or anti-unitary matrix $U \in \CM$, which only depends on the times $t_1,t_0$ and is equal to the identity
for $t_1=t_0$.\\
(2)  The transformation $\phi(\cdot,t_1,t_0)$  preserves coherence:\vspace{0.3cm}\\
\centerline{if $\domain_0 $ is SDG, then $\domain_1=\{G \in \CH \mid U^\dagger G U \in \domain_0\}$ is  SDG too.}
\end{MethodB}
\end{SnugshadeB}

That is, in spite of the weakness of the conditions required for $\phi$, we have obtained that:
\begin{itemize}
\item At any fixed time instant, Alice is rational. In fact, by point~(2) above, at time $t_1\geq t_0$ Alice's beliefs are modelled by an SDG.
\item The payoff of gambles $G_0$ and $G_1=\phi(G_0,t_1,t_0)$ are the same, for any gamble $G_0$.  That is, the gamble may change but the payoff is time-invariant (this is a consequence of the fact that $U$ is just a rotation matrix).

\end{itemize}
We can think of these two points as a slight weakening of the notion of strong temporal coherence given by \citet[Def.~25]{zaffalon2013a}. That notion eventually requires the cones to be equal at different time instants; here we just obtain the payoffs to be equal. The rationale does not really change in the two cases: intuitively, here we allow the two cones to be isomorphic to each other instead of just being equal. For this reason, with a small abuse of terminology, we call rule~5 above strong temporal coherence too. 

By exploiting duality, we can also reformulate the above results  in terms of credal quantum sets (see the Appendix, Proposition~\ref{prop:tempcoherencecreal}).
Indeed, let $\mdesirs_0$ be the credal quantum set associated to $\domain_0$.
Because of the relationship $\domain_1=\{G \in \CH \mid U^\dagger G U \in \domain_0\}$ 
and the duality between  $\mdesirs_0$ and $\domain_0$, it is possible to derive that:
\begin{equation}
 \domain_1=\left\{G \in \CH \mid G \gneq 0 \right\} \cup 
 \left\{G \in \CH \mid (U^\dagger G U)\cdot \rho >0, ~~\rho \in 
 \mdesirs_0\right\}.
\end{equation}

Notice that
\[
 (U^\dagger G U)\cdot \rho= Tr(U^\dagger G^\dagger U \rho)=Tr( G^\dagger U \rho U^\dagger),
\]
and moreover $Tr(U \rho U^\dagger)=Tr(\rho U^\dagger U)=Tr(\rho)=1$. Therefore
\begin{equation}
\label{eq:tempcredal}
 \mdesirs_1=\left\{U \rho U^\dagger \Big| \rho \in  \mdesirs_0 \right\}.
\end{equation}
As time forward is implemented by a unitary operator, while an anti-unitary operator implements time reversal, we have derived the third axiom of QM as a direct consequence of strong temporal coherence.

\begin{SnugshadeE}
\begin{example}[Classical coin.]
 \label{ex:coinE}
 Let us go back to our coin example.   First observe that any $2 \times 2$ unitary matrix can be written as
 \begin{equation}  
    U = \left[\begin{matrix} a & b \\ -e^{i\theta} b^\dagger & e^{i\theta} a^\dagger \\ \end{matrix}\right] , 
 \end{equation} 
 with $\theta \in \reals$, $a,b \in \complexs$ and  $ \left| a \right| ^2 + \left| b \right| ^2 = 1$.
 We have seen that in the classical case we are only interested in diagonal gambles
 $G_0=diag(g_1,g_2)$. Assume Alice finds $G_0$ desirable at time $t_0$, then
 $$
  G_1=U^\dagger G_0 U=\left[\begin{matrix} \left|a\right|^2g_1+\left|b\right|^2g_2 & a^\dagger b (g_1 -  g_2)  \\  ab^\dagger (g_1-g_2) & \left|b\right|^2g_1+\left|a\right|^2g_2  \\ \end{matrix}\right]. 
 $$
 Since we are only interested in diagonal gambles $G_1$ we must only consider transformations $U$ that transform
 diagonal matrices into diagonal matrices. This means that $a^\dagger b=ab^\dagger=0$.
 This together  with the constraint  $ \left| a \right| ^2 + \left| b \right| ^2 = 1$ implies that,  in the classical case, the only valid matrices are
 (i)~the \textbf{identity map} $U_1=I$ and (ii)~the \textbf{permutation matrix}
 $$
U_2=\begin{bmatrix} 0 &  1 \\  1 & 0\end{bmatrix}.
 $$
 There are many other valid matrices, like $U=[0,-1;-1,0]$ for instance, but their actions
 on $G$ are either the identity map  or a flipping $U_2^\dagger G_0 U_2=diag(g_2,g_1)$.
 Therefore, in the classical case, the only aspect of the experiment that is allowed to change with time is the name of the events: $\Omega=\{Head,Tail\}$ can become  $\Omega=\{Tail,Head\}$, but that is all. It is then clear that if Alice's set of strictly desirable gambles is coherent at time $t_0$ then it is also coherent at time $t_1$ (after the permutation).
 
 \centering
 \begin{minipage}{0.23\textwidth}
      \begin{tikzpicture}[scale=1.5]
    \draw[->] (-1,0) coordinate (xl) -- (1,0) coordinate (xu) node[right] {$g_1$};
    \draw[->] (0,-1) coordinate (yl) -- (0,1) coordinate (yu) node[above] {$g_2$};
    \draw (0,0) circle (1.pt);
    \begin{pgfonlayer}{background}
      \draw[dashed] (0,0) -- (-0.66,1) coordinate (a1away);
      \draw[dashed] (0,0) -- (1,-0.25) coordinate (a2away);
      \fill[blue!50,  opacity=0.5] (0,0) -- (a1away) -| (a2away) --(0,0);
    \end{pgfonlayer}
  \end{tikzpicture}
     \captionof{subfigure}{}
  \end{minipage}\quad
   \begin{minipage}{0.23\textwidth}
      \begin{tikzpicture}[scale=1.5]
    \draw[->] (-1,0) coordinate (xl) -- (1,0) coordinate (xu) node[right] {$g_2$};
    \draw[->] (0,-1) coordinate (yl) -- (0,1) coordinate (yu) node[above] {$g_1$};
    \draw (0,0) circle (1.pt);
    \begin{pgfonlayer}{background}
      \draw[dashed] (0,0) -- (1,-0.66) coordinate (a1away);
      \draw[dashed] (0,0) -- (-0.25,1) coordinate (a2away);
      \fill[blue!50,  opacity=0.5] (0,0) -- (a2away) -| (a1away) --(0,0);
    \end{pgfonlayer}
  \end{tikzpicture}
     \captionof{subfigure}{}
  \end{minipage}
  \captionof{figure}{Alices' sets of coherent strictly desirable gambles at time $t_0$  (left) and at time $t_1$ (right)
  after action $U_2$.\label{fig:coinflip}}
\end{example}
\end{SnugshadeE}
 A similar thing happens for the quantum coin. $U$ is a rotation matrix and, therefore, its only effect on Alice's SDG is to rotate the cone in a way that $\gambles^+$ is always mapped into $\gambles^+$ and that the payoffs stay the same. 
 
At this point it should be clear why there is a strong form of coherence in time in our formulation: because, by ``moving time forward'' through $\phi$, we can always think that Alice assesses her beliefs at time $t_1$ instead of $t_0$; and hence, that she is implicitly dealing always with the very same cone of desirable gambles. The same reasoning can be applied to the question of updating discussed in Section~\ref{sec:stc1}: if the measurement is done at time $t_1$, we can equivalently think of moving time forward up to $t_1$ and then use (fixed-time) conditioning to consequently update Alice's beliefs.\footnote{It is also possible to show that the same argument holds also in the case of multiple future time points, possibly with measurements in between.} 

Overall, this means that Alice is always using the same global probabilistic model, her SDG, irrespective of time. And that model cannot be made inconsistent by definition of SDG. Once again, this is the very essence of the idea of strong coherence in time. For these reasons, and to make things simpler, in the next sections we will assume that $t_1=t_0$ so as to avoid overcomplicating the treatment to explicitly account for time.

\subsection{Interpretation of the third axiom as temporal coherence}
 In this section we have derived the third axiom of QM by strong temporal coherence considerations. We have achieved this through
a simple rule that describes how a gamble at time $t_1$ corresponds to a gamble at time $t_0$. This rule implies that the future is completely determined by the past,
 that the utility scale is linear and that a state of complete ignorance is mapped into a state of
 complete ignorance. 
 Then we have shown as a consequence that the relationship between Alice's coherent set of desirable gambles  at time $t_0$  and Alice's coherent set of desirable gambles at time $t_1$ must follow a \textbf{unitary transformation} -- those that preserve coherence in time.
 
 
\section{Fourth axiom of QM}
\label{sec:4ndaxiom}
Let us turn now to consider coherence for composite systems. In particular, we focus on a composite 
system $AB$ made by two subsystems denoted as $A$ and $B$. Recall that QM gambles are Hermitian matrices whose dimension depends on the number of outcomes of the experiment; thus, if $A$ has $n$ outcomes and $B$ has $m$ outcomes, the composite systems has $nm$ outcomes and the corresponding gambles $G$ are matrices in $\mathbb{C}_h^{nm \times nm}$. We aim, as usual, at empolying  coherence-preserving operations to derive the fourth axiom of QM:

   \begin{SnugshadeF}
\begin{MethodF}[Axiom IV:]
the state space of a {separable} composite physical system is the 
tensor product of the state spaces of the component physical systems. Moreover, if we have
systems numbered $1$ through $\ell$, and system number $i$ is in the state $\rho_i$,
then the joint state of the total system is $\rho_1 \otimes \rho_2 \otimes \cdots \otimes \rho_{\ell}$.
\end{MethodF}
\end{SnugshadeF}

We will define two operations that are useful when dealing with composite systems: marginalisation and extension. Moreover, we will define independence concepts. To make things simpler, we will only consider the case $\ell=2$, but the next results hold for any $\ell$.

\subsection{Marginalisation}
\label{sec:marginalisation}
Let us denote Alice's SDG on the composite system $AB$ as $\domain^{AB}$.
Given $\domain^{AB}$, Alice may wish to focus on one component only of the composite system $AB$, meaning that she neglects one aspect of the experiment's outcomes. For instance, she can focus on $A$ -- which means ignoring $B$.
Then the interesting gambles are those whose payoffs do not depend on the $B$ component, which are those that are constant with respect to $B$.
It is easy to see that these gambles are of the form $G \otimes I_m$, with $G \in \mathbb{C}_h^{n \times n}$ and $I_m$ being the $m\times m$ identity matrix.

\begin{definition}\label{def:margin}
Let $\domain^{AB}$ be an SDG in $\mathbb{C}_h^{nm \times nm}$. The {\bf marginalisation} of $\domain^{AB}$ w.r.t.\  $A$ 
is defined as:
\begin{equation}
\label{eq:marginal}
\marg_A(\domain^{AB})=\left\{G \in \CH \mid G \otimes I_m \in \domain^{AB}\right\}.
\end{equation} 
Analogously, the marginalisation of $\domain^{AB}$ w.r.t.\  $B$ 
is defined as:
$$
\marg_B(\domain^{AB})=\left\{G \in \mathbb{C}_h^{m \times m} \mid I_n \otimes G \in \domain^{AB}\right\},
$$ 
where $I_n$ is the $n\times n$ identity matrix. 
\end{definition}

It can be proven that if $\domain^{AB}$ is an SDG in $\mathbb{C}_h^{nm \times nm}$, then both
$\marg_A(\domain^{AB})$ and $\marg_B(\domain^{AB})$ are SDGs (Appendix, first part of Proposition \ref{pro:marginal}): i.e., marginalisation \textbf{preserves coherence}. 
Moreover,
$\marg_A(\domain^{AB})$ and  $\marg_B(\domain^{AB})$ are both maximal SDGs whenever $\domain^{AB}$ is a maximal SDG (Appendix, second part of Proposition \ref{pro:marginal}).

In the case of a maximal SDG, Definition \ref{def:margin} provides a way to compute the marginal density matrices $\rho^A,\rho^B$ from  $\rho^{AB}$ (in particular from the dual $\domain^{AB}$ using~\eqref{eq:marginal}).
The standard way to obtain $\rho^{A}$ from $\rho^{AB}$ in QM is by using the \textbf{partial trace operator} \cite[Sec. 2.4.3]{nielsen2010quantum}:
$$
\rho^A=Tr_B(\rho^{AB}).
$$

In QM, $\rho^{A}$ is called  a \textit{reduced density operator}. By keeping  our analogy between
density matrices and probabilities, we simply call it a \textit{marginal density matrix}.
It is interesting that the motivations behind the use of the partial trace in QM is that it is the unique operation
that gives rise to the correct description of observable quantities for subsystems of a composite system \cite[Box 2.6]{nielsen2010quantum}.
Since our definition of marginalisation is consistent with that of partial tracing --~see in particular \citet[Box 2.6]{nielsen2010quantum}~--, the use of partial tracing has been justified from the perspective of our gambling system:  Alice wishes to focus on one of the components of the composite system $AB$, meaning that she neglects one aspect of the experiment's outcome.

From the fact that marginalisation preserves both coherence and maximality, we can also derive that
if the dual of a maximal SDG on $AB$, the density matrix $\rho^{AB}$, is
equal to the tensor product $\rho^{A}\otimes \rho^{B}$ (Appendix, Proposition \ref{prop:marginalstates}) then  the marginal SDG computed via~\eqref{eq:marginal},  $\marg_A(\domain^{AB})$, has a dual that only includes the density matrix $\rho^A$ --~and similarly
for $B$. In QM, a density matrix $\rho^{AB}$ such that $\rho^{AB}=\rho^{A}\otimes \rho^{B}$ is called \textit{separable}.
In this case too, by keeping  our analogy between density matrices and probabilities, we simply say
that the density matrix $\rho^{AB}$ factorises in the product of $\rho^{A}$ and $\rho^{B}$.
In probability, when a joint probability on $AB$ can be decomposed as the product of its marginals on $A$ and $B$, 
we talk about \textit{independence}. We will show in Section \ref{sec:independence} that we can do
the same in QM.

\begin{SnugshadeE}
 \begin{example}[Two classical coins.]
\label{ex:coinF}
Let us consider two classical coins $A$ and $B$.
A gamble on the composite system is in this case  a Hermitian matrix on $\mathbb{C}_h^{4 \times 4}$.
We know that in the classical case we are only interested on diagonal gambles: $G=diag(g)$ with $g \in \reals^4$, that is
$g=(g_1,g_2,g_3,g_4)$, where $g_1$ is the payoff of the outcome $(H,H)$ (Head on  coin $A$ and Head on coin $B$), and similarly $g_2$ for $(H,T)$, $g_3$ for $(T,H)$, $g_4$ for $(T,T)$.

Alice finds the gambles  $f=(1,-1,1,-1)$, $h=(-1,1,-1,1)$, $l=(1,1,-1,-1)$, $e=(-1,-1,1,1)$ desirable but not strictly so
(besides, as usual, the gambles $G=diag(g)\gneq 0$).
Hence her SDG $\domain^{AB}$ is:
$$
\domain^{AB}=\{G=diag(g) \in \mathbb{C}_h^{4 \times 4}  \mid g_{11}+g_{22}+g_{33}+g_{44}>0\}.
$$
Let us see what we can conclude about her beliefs on the probability of the four outcomes of the experiment.
By applying~\eqref{eq:credal1}, we derive that
$$
\mdesirs^{AB}=\{\rho \in \mathcal{D}_h^{4 \times 4} \text{ diagonal} \mid G\gneq0 \text{ or } G\cdot \rho>0~\forall G\in \domain^{AB}\}.
$$
Hence, by exploiting the fact that $Tr(\rho)=1$, we can easily derive that 
$G\cdot \rho=Tr(G^\dagger \rho)=\sum_i g_{ii} \rho_{ii}>0$ for any $g \in \domain$, which implies that
$\rho_{11}=\rho_{22}=\rho_{33}=\rho_{44}=1/4$.
Therefore, Alice believes that the four outcomes have all the same probability $1/4$.

Let us now compute the marginal $\marg_A(\domain^{AB})$.
By definition, we have that
$$
\marg_A(\domain^{AB})=\left\{G \in \mathbb{C}_h^{2 \times 2} \mid G \otimes I_2 \in \domain^{AB}\right\}.
$$ 
Since we are focusing on diagonal gambles, i.e., $G=diag(g_1,g_2)$, then $G \otimes I_m=(g_1,g_1,g_2,g_2)$.
Therefore, we have that 
$$
\marg_A(\domain^{AB})=\left\{G=diag(g) \in \mathbb{C}_h^{2 \times 2} \mid (g_1,g_1,g_2,g_2) \in \domain^{AB}\right\},
$$ 
and $(g_1,g_1,g_2,g_2) \in \domain^{AB}$ if and only if $g_1+g_1+g_2+g_2>0$ or, equivalently,
$g_1+g_2>0$. Hence:
$$
\marg_A(\domain^{AB})=\left\{G=diag(g) \in \mathbb{C}_h^{2 \times 2} \mid g_1+g_2>0\right\}.
$$ 
Its dual is clearly 
$$
\mdesirs^{A}=\{\rho \in \mathcal{D}_h^{2 \times 2} \text{ diagonal} \mid  \rho=diag(1/2,1/2)\},
$$
meaning that Alice believes the coin $A$ is fair too.
The same holds for $B$. We have derived that the marginal of the joint probability $(1/4,1/4,1/4,1/4)$ is $(1/2,1/2)$ both for $A$ and $B$.
\end{example}
\end{SnugshadeE}

\begin{SnugshadeE}
\begin{example}[Two quantum coins.]
\label{ex:spinF}
We assume that Alice's SDG on the composite system $AB$ is
$$
\domain^{AB}=\{G\in \mathbb{C}_h^{4 \times 4}  \mid G\gneq0 \}\cup\{G\in \mathbb{C}_h^{4 \times 4}  \mid G\cdot D >0\}, \text{ with }D=\frac{1}{2}\left[
 \begin{matrix}
1 & 0 &0 & 1\\  
0 & 0 &0 & 0\\ 
0 & 0 &0 & 0\\ 
1 & 0 &0 & 1\\ 
 \end{matrix}
 \right].
$$
By~\eqref{eq:credal2}, we have that $\mdesirs^{AB}=\{\rho=D\}$.
Note that $D$ is an entangled  state.  Alice's marginal SDG on the subsystem $A$ is given by
$$
\marg_A(\domain^{AB})=\left\{G \in \mathbb{C}_h^{2 \times 2} \mid G\otimes I_2 \in \domain^{AB}\right\}, \textit{ where } G \otimes I_2= \begin{bmatrix} g_{11} & 0 & g_{12} & 0 \\
 0 & g_{11} & 0 & g_{12} \\
 g_{12}^\dagger & 0 & g_{22} & 0 \\
 0 & g_{12}^\dagger & 0 & g_{22} \\
\end{bmatrix},
$$ 
where we have used the notation $G=[g_{11},g_{12};g_{12}^\dagger,g_{22}]$ for the elements of $G$.
$G \otimes I_2$ belongs to $\domain^{AB}$ provided that
$G \otimes I_2 \cdot D>0$, but
$G \otimes I_2 \cdot D=Tr(G^\dagger \otimes I_2  D)=g_{11}+g_{11}+g_{22}+g_{22}$.
Therefore, $\marg_A(\domain^{AB})=\left\{G \in \mathbb{C}_h^{2 \times 2} \mid g_{11}+g_{22}>0\right\}$ and by~\eqref{eq:credal2}:
$$
\mdesirs^A=\left\{\rho=\left[\begin{matrix}1/2 & 0\\0 & 1/2\end{matrix}\right]\right\}.
$$
So the marginal SDG on $A$ is maximal. $\rho=diag(1/2,1/2)$ is the marginal of the entangled state $D$.
The same holds for $B$.
\end{example}
\end{SnugshadeE}

\subsection{Coherent extension}
Let us assume that Alice assesses her beliefs separately about the systems $A$ and $B$ and study how she can coherently extend her beliefs to $AB$. This extension operation can be regarded as the opposite of marginalisation
discussed in Section \ref{sec:marginalisation}.

\begin{definition}\label{def:naturalextension}
 Let $\domain^A$ and $\domain^B$ be two SDGs in  $\CH$ and, respectively, $\mathbb{C}_h^{m \times m}$. Their
  \textbf{natural extension}  to $\mathbb{C}_h^{nm \times nm}$ is the smallest, i.e., the least-committal, SDG on $AB$, denoted with $\ext_{AB}(\domain^A,\domain^B)$, such that:
 \begin{equation}
  \marg_A( \ext_{AB}(\domain^A,\domain^B)) =\domain^A \text{ and } \marg_B(\ext_{AB}(\domain^A,\domain^B))=\domain^B.
 \end{equation} 
 \end{definition}
 The natural extension $\ext_{AB}(\domain^A, \domain^B)$ always exists (Appendix, Proposition \ref{pro:NE}).

In case  $\domain^A$ and $\domain^B$  are maximal SDGs with dual density matrices
$\rho^A$ and $\rho^B$, the natural extension reduces to find all density matrices $\rho^{AB}$ such that $\rho^A=Tr_B(\rho^{AB})$ and $\rho^B=Tr_A(\rho^{AB})$.
It is clear that $\rho^{AB}$ may be not unique in general, as shown in the next example. 
 
 \begin{SnugshadeE}
 \begin{example}[Two quantum coins.]
\label{ex:spinG}
Let us assume that $\domain^{A}=\domain^{B}=\left\{G \in \mathbb{C}_h^{2 \times 2} \mid g_{11}+g_{22}>0\right\}$,
which are the ones we have found in Example \ref{ex:spinG}. The set $\ext_{AB}(\domain^A, \domain^B)$ is the smallest SDG $\domain^{AB}$ such that $ \marg_A(\domain^{AB})=\domain^{A}$ and $ \marg_B(\domain^{AB})=\domain^{B}$.
In other words, we have to search for  the smallest SDG $\domain^{AB}$ such that:
 \begin{equation}
  G \otimes I_n\in \domain^{AB} \text{ and } I_n \otimes H \in \domain^{AB},
 \end{equation} 
 for all $G \in \domain^A$ and $H \in \domain^B$.
Note that
\begin{equation}
\label{eq:goog}
G \otimes I_2= \left[\begin{array}{cccc}
 g_{11} & 0 & g_{12} & 0 \\
 0 & g_{11} & 0 & g_{12} \\
 g_{12}^\dagger & 0 & g_{22} & 0 \\
 0 & g_{12}^\dagger & 0 & g_{22} \\
\end{array}\right], ~~I_2 \otimes H= \left[\begin{array}{cccc}
 h_{11} &  h_{12} & 0 & 0 \\
  h_{12}^\dagger  & h_{22} & 0 & 0 \\
 0 & 0& h_{11} &  h_{12}  \\
   0 & 0 & h_{12}^\dagger  & h_{22} \\
\end{array}\right],
\end{equation} 
with $g_{11}+g_{22}>0$, $h_{11}+h_{22}>0$. 
By duality (that $G \otimes I_2 \in \domain^{AB}$ provided that $G \otimes I_2 \cdot \rho^{AB}$ for any $\rho^{AB} \in \mdesirs^{AB}$)   it is clear that the valid density matrices $\rho^{AB}$ must satisfy:
$$
G \otimes I_2 \cdot \rho^{AB}>0 \text{ and }  I_2 \otimes H \cdot \rho^{AB}>0.
$$
For instance, all the following matrices satisfy the above constraints:
$$
\begin{array}{l}
\rho_1^{AB}=\frac{1}{2}\left[
 \begin{matrix}
1 & 0 &0 & 0\\  
0 & 0 &0 & 0\\ 
0 & 0 &0 & 0\\ 
0 & 0 &0 & 1\\ 
 \end{matrix}
 \right], ~~\rho_2^{AB}=\frac{1}{2}\left[
 \begin{matrix}
1 & 0 &0 & 1\\  
0 & 0 &0 & 0\\ 
0 & 0 &0 & 0\\ 
1 & 0 &0 & 1\\ 
 \end{matrix}
 \right],\\
 \rho_3^{AB}=\frac{1}{2}\left[
 \begin{matrix}
1 & 0 &0 & \iota\\  
0 & 0 &0 & 0\\ 
0 & 0 &0 & 0\\ 
-\iota & 0 &0 & 1\\ 
 \end{matrix}
 \right],~~\rho_4^{AB}=\frac{1}{2}\left[
 \begin{matrix}
1 & 0 &0 & \tfrac{1}{\sqrt{2}}(1-\iota)\\  
0 & 0 &0 & 0\\ 
0 & 0 &0 & 0\\ 
\tfrac{1}{\sqrt{2}}(1+\iota) & 0 &0 & 1\\ 
 \end{matrix}
 \right],\end{array}
 $$
and they are not the only ones. 
\end{example}
\end{SnugshadeE}
The natural extension has some resemblance with the operation of \textit{purification} in QM.
Given a density matrix $\rho^A$ of a quantum system, the goal is to find another subsystem $B$
and a pure state $\rho^{AB}$ such that $\rho^A=Tr_B(\rho^{AB})$ \cite[Sec.~2.5]{nielsen2010quantum}.
However, in the case of the natural extension $\rho^A$ and $\rho^B$ are given and the goal
is to find  all $\rho^{AB}$ compatible with these marginals. Therefore, the natural extension is just the inverse of marginalisation.

\subsection{Irrelevance and independence}
\label{sec:independence}
Let $\domain^{AB}$ be Alice's SDG for the composite system $AB$. Let us assume that Alice receives some information only related to the subsystem $B$. 
We address the case in which information concerning the subsystem $B$ does not change Alice's beliefs about system $A$ and/or vice versa.
To this end, we define the concepts of \textbf{irrelevance and independence}.

First of all, we consider the following result from linear algebra.
Let us consider a projection measurement on $B$, $\Pi^B=\{\Pi_i^B\}_{i=1}^m$,
then  $I_n \otimes \Pi_i^B$ is a projection measurement for the composite system $AB$. 
The proof of this result is immediate.  First observe that
 $ (I_n \otimes \Pi_i^B)^\dagger=I_n \otimes \Pi_i^B=(I_n \otimes \Pi_i^B)^2$ (idempotent), and
$ (I_n \otimes \Pi_i^B)(I_n \otimes \Pi_j^B)=I_n \otimes (\Pi_i^B \Pi_j^B)=0$ (orthogonal),
and finally  $ \sum_i  I_n \otimes \Pi_i^B=I_n \otimes I_m =I_{nm}$ (identity). 
However, it should be observed that $ I_n \otimes \Pi_i^B$ has not rank $1$.

We now define the concept of epistemic irrelevance.
\begin{definition}\label{def:epirrelev}
An SDG $\domain^{AB}$ is said to satisfy {\bf epistemic irrelevance}
of $A$ to $B$ when $\marg_A(\domain^{AB})=\marg_A(\domain^{AB}_B)$ where
$$
\domain_B^{AB}=\left\{H \in \mathbb{C}_h^{nm \times nm} \mid H \gneq0 \textit{ or } (I_n  \otimes \Pi^B_i) H  (I_n \otimes \Pi^B_i) \in 
\domain^{AB}\right\},
$$
 for each projector $\Pi^B_i$ and each projection measurement $\Pi^B$. 
 \end{definition}
 
Let us briefly explain this definition. $\domain_B^{AB}$ is the SDG conditional on the event indicated 
by $I_n  \otimes \Pi^B_i$, as it follows from its definition and~\eqref{eq:condition}. Thus,  $\marg_A(\domain^{AB})=\marg_A(\domain^{AB}_B)$  means that
Alice's marginal SDG $\marg_A(\domain^{AB})$ on the subsystem $A$ and the marginal on $A$ of Alice's
SDG updated with the information ``the event indicated by $\Pi^B_i$ has happened'', which is $\marg_A(\domain^{AB}_B)$, coincide.
If this holds for all possible $\Pi^B_i$ then it means that any information on $B$ does not change Alice's beliefs on $A$:
$\domain^{AB}$ is said to satisfy  epistemic irrelevance of $A$ to $B$.
By exchanging the roles of $A$ and $B$, we define epistemic irrelevance of $B$ to $A$. From this we can get independence as a symmetrisation of irrelevance:

\begin{definition}\label{def:epindep}
 An SDG $\domain^{AB}$ is said to satisfy {\bf epistemic independence} of $A$ and $B$, when it satisfies epistemic 
irrelevance both of $A$ to $B$ and of $B$ to $A$.
 \end{definition}

If the density matrix of the composite system  $\rho^{AB}$ is the tensor product of the 
density matrices $\rho^{A}$ and $\rho^{B}$, then $\domain^{AB}$ satisfies epistemic independence of $A$ and $B$ (Appendix, Proposition \ref{prop:independence1}).
This result is similar to what we have in probability: if a joint probability $p_{AB}$  factorises
as $p_{AB}=p_Ap_B$, we say that the subsystems (variables) $A$ and $B$ are independent.
Actually also the inverse implication holds (Appendix, proposition \ref{prop:independence2}).
From these facts, we can finally derive the fourth axiom of QM:
    \begin{SnugshadeB}
\begin{MethodB}[Subjective formulation of the fourth axiom of QM]\\
Given a maximal SDG $\domain^{AB}$ that satisfies epistemic independence of $A$ and $B$, then
its dual density matrix $\rho^{AB}$ is equal to 
$$
\rho^{AB}=\rho^{A}\otimes \rho^{B},
$$
where $\rho^{A}$ is the  dual density matrix of $\marg_A(\domain^{AB})$ and
$\rho^{B}$ is the  dual density matrix of $\marg_B(\domain^{AB})$ and vice versa.
In other words, the last axiom of QM follows by \textbf{epistemic independence}
of $A$ and $B$. 
\end{MethodB}
\end{SnugshadeB}

\subsection{Separable natural extension}
The aim of this section is to show that Bell's inequality violation can easily be framed in our theory of desirable gambles.

Before doing that, let us characterise a possible feature of the natural extension as defined in Section~\ref{sec:4ndaxiom}: 

\begin{definition}
 Let $\domain^A$ and $\domain^B$ be two SDGs in  $\CH$ and, respectively, $\mathbb{C}_h^{m \times m}$. We say that their natural extension 
 $\ext_{AB}(K^A, K^B)$ is   \textbf{separable}  if its dual only includes
 separable density matrices, i.e., 
 \[
 \mdesirs^{AB}=\left\{ \sum_i w_i \rho_i^A \otimes \rho_i^B \Big| \sum_i w_i = 1, w_i \geq 0, \rho_i^A  \in \mathcal{D}_h^{n \times n}, \rho_i^B\in   \mathcal{D}_h^{m \times m}    \right\},
 \]
 where the sum can   be over a countable set of indices.
 \end{definition}

Now, assume that  $\domain^A$ and $\domain^B$ are maximal SDGs, so that they uniquely define
 $\rho^A$ and  $\rho^B$. Assume  that their natural extension is separable. Then the following inequalities hold (Appendix, Theorem \ref{th:frechet})
 for each $ \rho^{AB} \in  \mdesirs^{AB}$:
 \begin{equation}
   \rho^{AB} \leq  \rho^{A} \otimes I_m \;\;\text{ and }\;\; \rho^{AB} \leq  I_n \otimes \rho^{B},
 \end{equation} 
 and
 \begin{equation}
   \rho^{AB} \geq  \rho^A \otimes I_m + I_n \otimes \rho^B-I_{nm} \;\;\text{ and }\;\; \rho^{AB} \gneq 0.
   \end{equation}

 The upper bound is know in QM as ``reduction criterion'' for density matrices; it was
first proven by \citet{horodecki1999reduction} and independently formulated by \citet{cerf1999reduction}. 
Structurally the above inequalities are analogs of the Fr\'echet bounds in probability theory.
   \begin{SnugshadeF}
\begin{MethodF}[Fr\'echet bounds:] they are bounds for the probabilities of the conjunction of two events given the probabilities of the individual events.
If $A,B$ are  events, the Fr\'echet inequalities are
$$
\max(0,P(A) + P(B)-1) \leq P(A \cap B) \leq  \min(P(A),P(B)).
$$
If $P(A) = 0.7$ and $P(B)  = 0.6$, then the probability of the conjunction, i.e., $A \cap B$, is in the interval
$$
\begin{array}{rcl}
P(A \cap B) &\in& [\max(0,P(A) + P(B)-1), \min(P(A),P(B))]~\\
&=& [\max(0, 0.7 + 0.6-1), \min(0.7, 0.6)]= [0.3, 0.6].\end{array}
$$
\end{MethodF}
\end{SnugshadeF}

   \begin{SnugshadeB}
\begin{MethodB}[Entangled states violate the Fr\'echet bounds for density matrices.]

 Consider for instance the entangled density matrix
 $$
\rho^{AB}=\frac{1}{2}\left[
 \begin{matrix}
1 & 0 &0 & 1\\  
0 & 0 &0 & 0\\ 
0 & 0 &0 & 0\\ 
1 & 0 &0 & 1\\ 
 \end{matrix}
 \right],
 $$
which has marginal  $\rho^A=\rho^B=diag(1/2,1/2)$ as shown in Example \ref{ex:spinF}.  
Entangled states are not separable and it can easily be verified that 
 $$
 \rho^A \otimes I_m-\rho^{AB} \ngeqslant0 ~~\text{ and } ~~ I_n \otimes \rho^B -\rho^{AB} \ngeqslant0,
 $$ 
 since the  resulting matrices have one negative eigenvalue.
\end{MethodB}
\end{SnugshadeB}
In other words, from the viewpoint of our desirability foundations of QM, the motivation behind Bell's inequality violation is that entangled states exhibit a form of ``stochastic'' dependence stronger than the strongest classical dependence: and in fact they violate Fr\'echet bounds.
Explaining  Bell's inequality violation through credal sets (imprecise probability) was originally suggested by \citet{Hartmann2010}. Here, we have shown that this is in fact the case by working directly in $\CH$.

\subsection{Interpretation of the fourth axiom in terms of independence}
 In this section we have derived the fourth axiom of QM and shown that it preserves coherence.
 Moreover, by exploiting the semantics provided by the theory of desirable gambles,
 we have been able to extend to QM probabilistic concepts such as:
 \begin{itemize}
  \item  marginalisation;
  \item  natural extension;
  \item epistemic irrelevance and independence.
 \end{itemize}
 By duality, we have shown that marginalisation  corresponds  to
 partial tracing for density matrices.
For maximal SDGs, the epistemic independence of $A$ and $B$  implies that  $\rho^{AB}=\rho^{A} \otimes \rho^{B}$ -- as it happens in probability theory, independence means that the joint  $\rho^{AB}$ factorises as the product
 of its marginals $\rho^{A}, \rho^{B}$.
 Moreover, by defining  the concept of separable natural extension, we have also shown that in QM it is possible
 to derive bounds for the joint $\rho^{AB}$ that are similar to Fr\'echet bounds in probability theory
  and that  entangled states violate these bounds.
 This is all further evidence that QM is just a generalised theory of probability.
 
 
\section{Similarity and difference with QBism}\label{sec:qbysm}
The foundation of quantum mechanics through desirability presented in this paper shares a goal similar to that of QBism \mbox{-- }a classical-Bayesian interpretation of quantum mechanics\mbox{ --,} which is to justify quantum mechanics by using coherence arguments \cite{Caves02,Fuchs01,Schack01,Fuchs02,Fuchs03,Schack04,Fuchs04,Caves07,Appleby05a,Appleby05b,Timpson08,longPaper, 
Fuchs&SchackII,mermin2014physics}.

To this end, QBism aims at rewriting quantum mechanics in a way that 
does not involve the density operator $\rho$  (i.e., that is not defined on $\mathbb{C}_h^{n \times n}$)
and use classical probabilistic Bayesian arguments to prove coherence of quantum mechanics.
For instance, \citet{Fuchs&SchackII} have shown that it is possible 
to rewrite Born's rule in a way that does not include the density operator $\rho$ anymore, so as to give it a classical-Bayesian interpretation.
For this purpose, they consider as measure a \textit{symmetric informationally complete} (SIC) POVM  \cite{renes2003symmetric,appleby2007symmetric}, 
which allows Born's rule to be rewritten as:
\begin{equation} \label{eq:bornruleBayes0}
\begin{array}{rcl}
   p_j = \sum\limits_{i=1}^{n^2} \left((n+1)q_i - \dfrac{1}{n}\right) r(j|i) .\\
\end{array} 
\end{equation}
Here, $q_i$ represent the probability  of the outcomes of an experiment (a SIC measure with $n^2$ outcomes), which
is only imagined and never performed. The terms $r(j|i)$ are instead the probabilities of seeing the outcome $j$ in the 
actually performed experiment, given  the  outcome  $i$ in the imagined one. Finally, $ p_j$ is the probability of the 
outcomes in the actual experiment.  Fuchs and Shack  call~\eqref{eq:bornruleBayes0} the ``quantum law of total probability'', because they interpret it  as a modification  of the standard law of total probability:
\begin{equation} \label{eq:lawprob}
   p_j=\sum_{i=1}^{n^2} r(j|i) q_i,
\end{equation}
which is well known to follow from coherence, i.e., Dutch-book arguments  \cite{savage1954,finetti1970}.
Thus, according to Fuchs and Shack, the rule~\eqref{eq:bornruleBayes0} can be regarded as an addition to the standard Dutch-book coherence, a sort of ``coherence plus'', which defines the valid quantum 
states: 
\begin{quote}``It expresses a kind of empirically extended coherence -- not implied by Dutch-book coherence alone, but formally similar to the kind of relation one gets from Dutch-book coherence'' \cite{Fuchs&SchackII}.
\end{quote} 
Under this view, Equation~\eqref{eq:bornruleBayes0} can be interpreted as a consistency requirement, i.e., a law that 
must hold to produce a consistent quantum world. 
 Based on this idea, they suggest that~\eqref{eq:bornruleBayes0} should be taken \textit{as one of the basic axioms of quantum 
theory,}
since it provides a particularly clear way of thinking of Born's rule as an addition to Dutch-book coherence.
 This axiom formulates as ``everything in~\eqref{eq:bornruleBayes0} that should be a probability -- i.e., $p_j$, $q_i$, and 
$r(j|i)\,$) -- actually is a probability''. Using this axiom, they
 show that it is possible to derive some properties of quantum states directly in terms of probabilities.
 
While we definitely find this interpretation interesting, we regard it as portraying a partial view: in fact, the law~\eqref{eq:bornruleBayes0} is not total probability, and hence, strictly speaking, is incoherent according to de Finetti's (classical) definition of coherence. 

In contradistinction to QBism,  we do not aim at demonstrating that the probabilities \emph{derived} from QM are Bayesian, but rather that \emph{the Bayesian theory itself is QM} once it is generalised to work in $\mathbb{C}_h^{n\times n}$. This is achieved by bringing to light a duality relation between sets of gambles defined on the complex  space of hermitian matrices $\mathbb{C}_h^{n\times n}$ and density matrices, in the same way as there is a duality relation between sets of real-valued gambles and probabilities in the classical case. This also establishes in a definite sense that QM is strongly self-consistent. 

{Finally notice that, while in QBism Alice's beliefs are modelled by a single density matrix,  in our approach Alice can represent her beliefs by a closed convex set of density matrices: the more Alice is ignorant about the experiment the bigger is her credal quantum set (dually, the smaller is her cone of desirable gambles) and vice versa.
For instance, the proposed framework allows us  to represent cases in which Alice does not have any knowledge about the state of the quantum system. 
In these cases Alice should only accept  gambles given by PSDNZ matrices. The dual of her set of desirable gambles is therefore the set of all the
possible density matrices -- she does not know $\rho$ and, thus, all the $\rho$'s are possible candidates to be the actual one for her.}


\section{Conclusions and Future Work}\label{sec:conc}
In this paper we have derived the four postulates of QM from a mathematical framework made of rational gambling on a quantum experiment. Technically speaking, such a framework is the dual of the Bayesian theory extended to the space of Hermitian matrices. To say it with a slogan, we have obtained that QM is the Bayesian theory in the complex numbers. Our results mean, in other words, that QM is a theory of probability -- not just that probabilities can be derived from QM.

This result has implications that we find worth considering. On the one hand, it means that QM is a self-consistent theory, in a very strong sense: there is no way to draw incoherent conclusions by using QM. On the other, it means that the common operations employed with quantum mechanics can be regarded as rules of probability; and hence, that we can take advantage in a novel way of well-known probabilistic viewpoints to handle QM. We have given a number of examples to show that this is in fact the case, where, for instance, entanglement turns out to be a form of stochastic dependence, and Bell's inequalities become Fr\'echet bounds. It also gives us a new tool, which is coherence, to orient ourselves in the quantum world: we can in fact directly use the gambling framework to understand whether or not some approaches or proposals are consistent with QM. For example,
it can be shown that dispersion free probability measures, used as counter-examples to Gleason's theorem in dimension two, are incoherent (see \citet{BenavoliTODO}).

We should also at this point better clarify one peculiarity of the Bayesian theory we are interpreting QM with. The coherent cones of gambles that represent Alice's beliefs enable us to deal with probability (that is, QM) in a very flexible way: in fact, they allow us to smoothly deal with one or more than one density matrix at the same time. We have used this feature, for example, every time we had to express the case of Alice being ignorant about a quantum experiment. And it is just thanks to using sets of density matrices that some aspects of the theory of quantum mechanics -- like the violation of Bell's inequalities -- have been easily explained. They have also allowed us to solve the ``singularity'' corresponding to conditioning on an event of null probability and to define the natural extension of all entangled states compatible with some marginals.

Sets of Bayesian probabilities are traditionally part of so-called \emph{Bayesian robustness}, and they are a very active field of research nowadays under the generic name of \emph{imprecise probability}. Therefore, to be fully clear, we should say that QM in our present formulation is a theory of Bayesian robustness (see, e.g., \citet{berger1984,ber94}), one that deals with a set of ``probabilities'' and that becomes truly Bayesian in the limit of a single matrix in the set. 

However, the opportunities of using the flexible, set-based, formulation seem to be potentially many and worth exploring. One of these could be checking whether there is room in QM for cones of desirable gambles that are not strictly desirable and that for this reason are far more expressive than those we have used in this paper.

\appendix

\section{Propositions and proofs}\label{sec:app}

 \subsection{Results in Section \ref{sec:intro-clp}}

  \begin{Proposition}
  \label{pro:CGC>0}
  Let $G \in \CH$. It holds that 
  \[G\geq 0 \text{ iff }
  CGC^\dagger \geq0
 \text{ for any matrix } C \in \CM.\] 
 Moreover, if $G\gneq 0$ and $\Pi=\{\Pi_1,\dots,\Pi_n\}$ is a projection measurement, then there 
is $j\in\{1,\dots,n\}$
 such that $ \Pi_jG\Pi_j \gneq0$.
 \end{Proposition}
 \begin{proof}
   For the direct implication, it is enough to consider $C=I$. For the converse we reason as follows. By hypothesis $G\geq 0$, and therefore we can write $G=LL^\dagger$ using Cholesky factorisation.
 Then observe that
 $$
  CGC^\dagger=CL(CL)^\dagger,
 $$
 which is clearly PSD. 
 
 To prove the second part, we use the previous result to get
 $$
 \Pi_iL(\Pi_iL)^\dagger\geq0
 $$
 and observe that it must be  $j$ such that $ \Pi_jG\Pi_j \gneq0$, since
 $\Pi$ is a basis.
 \end{proof}

 \subsection{Results in Section \ref{sec:1staxiom}}
 
\begin{Proposition}
\label{pro:rho}
The set $\domain^\bullet$ is a closed convex cone.
\end{Proposition}
\begin{proof}
 $\domain^\bullet$ is a cone since if $R_1,R_2 \in \domain^\bullet$, i.e., $G\cdot R_1\geq 0$ and $G\cdot R_2\geq 
0$ for all $G \in \domain$, then $R_1+R_2 \in \domain^\bullet$
 because 
 $$
 G \cdot (R_1+R_2)= G\cdot R_1+ G\cdot R_2 \geq 0 ~ \forall G \in \domain,
 $$
 and $\lambda R_1 \in \domain$ for all $\lambda > 0$, since $\lambda G \cdot R_1 \geq 0$. 
 Closedness of $\domain^\bullet$ follows by the continuity of trace. 
\end{proof}
 
 \begin{Proposition}\label{prop:credal1equiv}
 Let $K$ be an SDG. The sets 
 $$
\mdesirs_1=\{ \rho \in \CD \mid G\cdot \rho \geq  0~ \forall G \in \domain\}
 $$
 and
 $$
\mdesirs_2=\{ \rho \in \CD \mid   G\gneq0 \text{ or } G\cdot \rho > 0~ \forall G \in \domain\}
$$
coincide.
\end{Proposition}
\begin{proof}
First, notice that $G\cdot \rho \geq  0$ for every $G\gneq0$ and every $\rho \in \CD$. This means that $\mdesirs_2 \subseteq \mdesirs_1$.
For the other direction, we reason as follows. Let $\rho \in \mdesirs_1$. We have to verify that $G\cdot \rho >  0$  for every $G \in \domain \setminus \gambles^+$. Consider such a $G \in \domain \setminus \gambles^+$ and assume that $G\cdot \rho = 0$. Since $\domain$ satisfies (S3), i.e. openness, there is $\Delta > 0$  such that $G-\Delta \in \domain$. In particular $\Delta$ can be chosen such that $\Delta\cdot \rho > 0$, whence $(G-\Delta)\cdot \rho = - \Delta\cdot \rho < 0$. \end{proof}

\noindent {\bf Theorem \ref{thm:dualityopen}} [Representation theorem]
\textit{ It holds that
\begin{itemize}
\item $(\domain^\bullet)^\circ=\domain$ for every SDG $\domain$, and 
\item $(\mdesirs^\circ)^\bullet=\mdesirs$ for every credal quantum set $\mdesirs$. 
\end{itemize}}

\begin{proof}
For the first item, we reason as follows.
By construction, $(\domain^\bullet)^\circ \supseteq \domain$. 
It remains to show that $(\domain^\bullet)^\circ -K=\emptyset$. Assume that there is $H \in 
(\domain^\bullet)^\circ -\domain$;
since $\domain$ is a convex set we can use the Hahn-Banach theorem to separate $\{H\}$ from $\domain$.
If $H$ is not in the closure of $\domain$, then by the  Hahn-Banach theorem there is a  Hermitian operator $\sigma$ such that 
$Tr(H \cdot \sigma) < 0$ and $Tr(G \cdot \sigma) \geq 0$ for all $G\in \domain$.
So there is $\sigma\in \domain^\bullet$ such that  $Tr(H \cdot \sigma) < 0$ and so
$H\notin (\domain^\bullet)^\circ$, a contradiction.
Now assume that $H\notin \domain$ is in the border of the closure of $\domain$, then since $\domain$ is a non-pointed 
 cone satisfying the openness property (S3)
that includes $G\gneq0$, it means that $H\gneq0$ is impossible. Then we can consider the convex open set $O \supset \domain$ that excludes $H$ and apply the Hahn-Banach theorem so that there is a  Hermitian operator $\sigma$ such that $Tr(H 
\cdot \sigma) \leq 0$ and $Tr(G \cdot \sigma) > 0$ for all $G\in O$. Again this implies that $H\notin 
(\domain^\bullet)^\circ$, a contradiction.

To prove the second point ($(\mdesirs^\circ)^\bullet=\mdesirs$), we exploit the fact that $\mdesirs$ is a closed convex cone
and we use Hahn-Banach similarly to the previous point.
\end{proof}

\begin{Proposition}
\label{prop:linearprev} Let $\domain$ be an SDG.
If $\overline{P}(G)=\underline{P}(G)$ for any $G \in  \domain$, then  
its dual credal set $  \mdesirs$ includes a single density matrix.
\end{Proposition}
\begin{proof}
Let us assume that this is not true and $  \mdesirs$ includes $\rho^a,\rho^b$ 
and $\rho^a, \rho^b$ differ for one element.
If this element is  the $i$-th element of the diagonal, then we can take $G=diag(0,\dots,0,1,0,\dots,0)$, where 
$1$ is in the $i$-th row.
Since $G\gneq0$, it is included in $\domain$ and moreover
$$
Tr(G^\dagger \rho^a)=\rho_{ii}^a \neq \rho_{ii}^b=Tr(G^\dagger \rho^b).
$$
Assume now that the element is in the $i$-th row and in the $j$-th column, that is
$\rho_{ij}^a \neq \rho_{ij}^b$ and because they are Hermitian matrices this also means that
$\rho_{ji}^a \neq \rho_{ji}^b$.
Then take the following gamble $G$:
$$
G= I + S, ~~~H= I + T,
$$
where $S$ is a matrix of zeros with $S_{ij}=S_{ji}=1$ and $T$ is a matrix of zeros with $T_{ij}^\dagger=T_{ji}=\iota$. 
The gambles $G,H$  are PSDNZ. Indeed, for instance in the case of $G$, for any $0 \neq u\in \mathbb{C}^n$ we have that $u^\dagger G u=u^\dagger u 
+u_iu_j^\dagger+u_ju_i^\dagger=(u_i+u_j)(u_i+u_j)^\dagger+\sum_{k\neq i,j} u_{kk}u_{kk}^\dagger\gneq 0$.   By (S2), both gambles $G$ and $H$ are 
in $\domain$.
Moreover
$$
\begin{array}{l}
Tr(G^\dagger \rho^a)=\sum_i \rho_{ii}^a+\rho_{ij}^a+\rho_{ji}^a, ~~~Tr(G^\dagger \rho^b)=\sum_i 
\rho_{ii}^a+\rho_{ij}^b+\rho_{ji}^b,\\
Tr(H^\dagger \rho^a)=\sum_i \rho_{ii}^a+\iota\rho_{ij}^a-\iota\rho_{ji}^a, ~~~Tr(H^\dagger \rho^b)=\sum_i 
\rho_{ii}^a+\iota\rho_{ij}^b-\iota\rho_{ji}^b.\\
\end{array}
$$
Write $\rho_{ij}^a=a+\iota b$ and $\rho_{ij}^b=c+\iota d$, we have that 
$$
\begin{array}{l}
 \rho_{ij}^a+\rho_{ji}^a=a+\iota b+a-\iota b=a, ~~~\rho_{ij}^b+\rho_{ji}^b=c+\iota d+c-\iota d=c,\\
  \iota\rho_{ij}^a-\iota\rho_{ji}^a=\iota a + b-\iota a+ b=b, ~~~ \iota\rho_{ij}^b- \iota\rho_{ji}^b=d,
\end{array}
$$
and one of them must be different. 
\end{proof}

\begin{Proposition}
 \label{pro:eigenvalue}
 Let $\mdesirs=\CD$,
 then we have that
   $$
  \underline{P}(G)= \inf_{\rho \in \CD } Tr(G^\dagger \rho)=\lambda_{min}(G), ~~  \overline{P}(G)= \sup_{\rho \in \CD } Tr(G^\dagger \rho)=\lambda_{max}(G),
  $$
  where $\lambda_{min}(G)$ (resp. $\lambda_{max}(G)$) is the minimum (resp. maximum) eigenvalue of $G$.
\end{Proposition}
\begin{proof}
 Let us assume that
 $$ 
\rho=\sum_{i=1}^k \lambda_i u_iu^\dagger_i.
 $$
 Let  $G=\sum_{j=1}^n \delta_j v_jv_j^\dagger$,
now if we decompose $v_j=\sum_{i=1}^n c_i^{(j)} u_i$ (w.r.t.\ the basis $\{u_i\}$), we can rewrite $G$ as
\begin{equation}
\label{eq:changebasis0}\begin{array}{l}
G=\sum\limits_{j=1}^n \delta_j v_jv_j^\dagger=\sum\limits_{j=1}^n \delta_j (\sum\limits_{i=1}^n c_i^{(j)} u_i)(\sum\limits_{k=1}^n c_k^{(j)} u_k)^\dagger\\
=\sum\limits_{j=1}^n \delta_j \left(\sum\limits_{i=1}^n c_i^{(j)}(c_i^{(j)})^\dagger u_iu_i^\dagger+\sum\limits_{i\neq  k, i,k=1}^n c_i^{(j)}(c_k^{(j)})^\dagger u_iu_k^\dagger\right).\\
\end{array}
\end{equation}
We exploit the fact that
$$
Tr(G^\dagger\rho)=\sum_{i=1}^k \lambda_i Tr(G^\dagger u_iu^\dagger_i)=\sum_{i=1}^k \lambda_i Tr( u_iu^\dagger_i G^\dagger u_iu^\dagger_i).
$$
Hence, observe that
$$
\begin{array}{l}
Gu_lu_l^\dagger=\sum\limits_{j=1}^n \delta_j \left( c_l^{(j)}(c_l^{(j)})^\dagger u_lu_l^\dagger +\sum\limits_{l \neq i=1}^n c_i^{(j)}(c_l^{(j)})^\dagger u_iu_l^\dagger\right)\\
= \left(\sum\limits_{j=1}^n \delta_j c_l^{(j)}(c_l^{(j)})^\dagger\right) u_lu_l^\dagger +\sum\limits_{l \neq i=1}^n \left(\sum\limits_{j=1}^n \delta_j c_i^{(j)}(c_l^{(j)})^\dagger \right)u_iu_l^\dagger
\end{array}
$$
and so 
\begin{equation}
\label{eq:changebasis}
\begin{array}{l}
u_lu_l^\dagger Gu_lu_l^\dagger=\left(\sum\limits_{j=1}^n\delta_j  c_l^{(j)}(c_l^{(j)})^\dagger\right) u_lu_l^\dagger=\hat{\delta}_l u_lu_l^\dagger,
\end{array}
\end{equation}
 where $\hat{\delta}_l=\sum\limits_{j=1}^n \delta_jc_l^{(j)}(c_l^{(j)})^\dagger $.
 Therefore 
 $$
 \sum_{i=1}^k \lambda_i Tr( u_iu^\dagger_i G^\dagger u_iu^\dagger_i)=\sum_{i=1}^k \lambda_i \hat{\delta}_i.
 $$ 
 Now since $\sum_{i=1}^k \lambda_i =1$, the minimum is obtained
 by choosing $\lambda_k=1$ where $k=\arg\min_i \hat{\delta}_i$, whence
 $$
  \sum_{i=1}^k \lambda_i Tr( u_iu^\dagger_i G^\dagger u_iu^\dagger_i)=\sum\limits_{j=1}^n \delta_jc_k^{(j)}(c_k^{(j)})^\dagger.
 $$
 Since $c_k^{(j)}(c_k^{(j)})^\dagger\leq 1$ (this follows from the fact that $1=v_j^\dagger v_j=\sum_{i=1}^n c_i^{(j)}(c_i^{(j)})^\dagger$ for each $j$)
 and $\sum\limits_{j=1}^n c_k^{(j)}(c_k^{(j)})^\dagger\leq 1$, the minimum value is obtained
 by putting all the mass on the index $j$ corresponding to $\arg\min_i \delta_i$
 and zero otherwise.
\end{proof}

 \subsection{Results in Section \ref{sec:2ndaxiom}}
\begin{Proposition}
\label{prop:conditioning} Let $\domain$ be an SDG. Then the set of desirable gambles conditional on  $\Pi_i$ \[\domain_{\Pi_i}=\left\{G \in \CH \mid  G \gneq0 \textit{ or }\Pi_i G \Pi_i \in \domain \right\}\]
is also an SDG.
\end{Proposition}
\begin{proof}
 By definition $\gambles^+ \subseteq \domain_{\Pi_i}$, thus $\domain_{\Pi_i}$ satisfies (S2).
The fact that $\domain_{\Pi_i}$ is a non-pointed convex cone (S1) can be shown as follows. Firstly 
$$
  \Pi_i G_1 \Pi_i,  \Pi_i G_2 \Pi_i\in \domain  ~~\Rightarrow~~  \Pi_i G_1 \Pi_i+ \Pi_i G_2 \Pi_i=  \Pi_i (G_1 +G_2) 
\Pi_i \in \domain
$$
and so $G_1 +G_2 \in \domain_{\Pi_i}$. Similarly
$$
  \Pi_i G \Pi_i \in \domain ~~\Rightarrow~~ \lambda \Pi_i G \Pi_i= \Pi_i (\lambda G) \Pi_i \in \domain.
$$
Finally $  \Pi_i 0 \Pi_i=0 \notin \domain$, meaning that $\domain_{\Pi_i}$ is non-pointed.

We finally verify that $\domain_{\Pi_i}$ satisfies openness (S3). Let $G \in \domain_{\Pi_i} \setminus \gambles^+$. It holds that $\Pi_i G \Pi_i \in \domain $.
We have two cases to consider, depending on whether $\Pi_i G \Pi_i \gneq 0$ or not. Assume first $\Pi_i G \Pi_i \gneq 0$, and consider $\Pi_j$, $j \neq i$. By definition $\Pi_j \gneq 0$. Moreover $\Pi_i \Pi_j \Pi_i =0$. Then $\Pi_i G \Pi_i - \Pi_i \Pi_j \Pi_i = \Pi_i (G -  \Pi_j) \Pi_i \gneq 0$. This means that $\Pi_i (G -  \Pi_j) \Pi_i \in \domain$ and therefore $(G -  \Pi_j) \in \domain_{\Pi_i}$, with $\Pi_j \gneq 0$. 
We now check the case that $\Pi_i G \Pi_i \notin \gambles^+$. Since $\domain = \mdesirs^\circ$, it holds that, for $\alpha > 0$,  
\begin{center}
\begin{tabular}{ l l l l }
  $\Pi_i (G  - \alpha \Pi_i)\Pi_i  \in \domain$ & $\Leftrightarrow$ & $\Pi_i G \Pi - \alpha \Pi_i \in \domain$ &  \\
    & $\Leftrightarrow$ & $Tr( (\Pi_i G^\dagger \Pi_i) \rho) - \alpha Tr(\Pi_i \Pi_i \Pi_i \rho) > 0$ & for every $\rho \in \mdesirs$ \\
     & $\Leftrightarrow$ & $Tr( (\Pi_i G^\dagger \Pi_i) \rho) - \alpha Tr(\Pi_i \rho \Pi_i) > 0$ & for every $\rho \in \mdesirs$. \\
\end{tabular}
\end{center}

Consider $\Delta > 0$. This implies that $\Delta= \sum_{j} \alpha_j \Pi_j$, with $\alpha_j > 0$. Assume  $\Pi_i G \Pi_i - \Delta  \in \domain$. We have that $\Pi_i G \Pi_i -\sum_{j} \alpha_j \Pi_j  \in \domain$ and then $Tr( (\Pi_i G^\dagger \Pi_i) \rho) - \sum_{j}\alpha_j Tr(\Pi_i \Pi_j \Pi_i \rho) > 0$ for every $\rho \in \mdesirs$. But this implies that $Tr( (\Pi_i G^\dagger \Pi_i) \rho) - \alpha_i Tr(\Pi_i \rho \Pi_i) > 0$ for every $\rho \in \mdesirs$. By the representation theorem, we conclude that for $\alpha_i\Pi_i \gneq 0$,  $\Pi_i (G  - \alpha \Pi_i)\Pi_i  \in \domain$ and thus $G  - \alpha \Pi_i \in \domain_{\Pi_i}$.
\end{proof}

The next results is an immediate consequence of  Proposition \ref{prop:conditioning}. Remember that conditioning for non-elementary events is defined as
$$
\domain_{\Pi_J}=\left\{G \in \CH \Bigg| G \gneq 0 \textit{ or }\sum_{j\in J} \Pi_j G  \Pi_j\in \domain 
\right\},
$$
for some subset $J$ of the indexes  $\{1,\dots,n\}$. 

\begin{Corollary}\label{cor:nonelementary}
 $\domain_{\Pi_J}$ is SDG.
\end{Corollary}

 \subsection{Results in Section \ref{sec:3ndaxiom}}
 For the sake of notation, by $h$ we denote the map $\phi(\cdot,t_1,t_0)$ introduced in Section \ref{sec:3ndaxiom}.
 \begin{definition}\label{def:tct}
A temporal coherent transformation is any map from $\gambles$ onto itself such that
  \begin{description}
    \item[C1] $h(cI)=cI$ for every $c \in \mathbb{N}$,
  \item[C2] $h(\gambles^+)= \gambles^+$,
  \item[C3] $h(\alpha G_1 + \beta G_2)=\alpha h(G_1) + \beta h(G_2)$ for every $G_1, G_2 \in \CH$ and every $\alpha, \beta > 0$.
  \end{description}
\end{definition}
  \begin{Theorem}
  \label{th:temp}
  Let $h$ be a temporal coherent transformation.
  Then there is a unitary or anti-unitary $U \in \CM$ such that \[h(G)= U^\dagger G U \] for every $G \in \CH$, which is unique up to a phase. Moreover
   if $\domain $ is an SDG, then $\domain'=\{G \in \CH \mid U^\dagger G U \in \domain\}$ is an SDG too.
 \end{Theorem}
\begin{proof}
 First notice that from (C3) in Definition \ref{def:tct} we get that $h$ is linear, i.e., for every $G_1, G_2 \in \CH$, and every $\alpha, \beta \in \reals$,
\[h(\alpha G_1 + \beta G_2)=\alpha h(G_1) + \beta h(G_2).\]
From \citet[Theorem~2]{schneider1965positive},
 there exists a non-singular matrix $U\in \CM$ such that either $h(G) = U^\dagger G U$ for every $G\in \CH$, or $h(G)=U^\dagger G' U$ for every $G\in \CH$. By condition~(C1) in Definition~\ref{def:tct}, $U^\dagger U =I$, meaning that $U$ is unitary or anti-unitary. Moreover, notice that, 
 up to a phase, the two operations ($h(G)=U^\dagger G U$ or $h(G)=U^\dagger G' U$) are the same.

We finally verify that $\domain'$ is an SDG. It is immediate to verify that it satisfies properties (S1),~(S2) in Definition \ref{def:sdg}. For property (S3) -- openness --, we reason as follows. 
Assume $G \in \domain'$. This means that $U^\dagger G U \in \domain$. Since $\domain$ is an SDG, it satisfies openness: $U^\dagger G U \gneq0$ or $U^\dagger G U -\Delta \in  \domain$ for some $0<\Delta \in \CH$. In the first case, notice that in this case too there is  $0<\Delta \in \CH$ such that from $U^\dagger G U -\Delta \in  \domain$. From (C2) in Definition \ref{def:tct}, we can pick $\Delta'>0$ such that $h(\Delta')=\Delta$. Thus $(G - \Delta') \in \domain'$, since $h(G-\Delta')=h(G) - h(\Delta')= h(G) - \Delta \in \domain$. 
\end{proof}

\begin{Proposition}\label{prop:tempcoherencecreal}
Consider a temporal coherent transformation $h$ represented by the unitary matrix $U$.
Let $\domain_0$ be an SDG,  $\domain_1=\{G \in \CH \mid U^\dagger G U \in \domain_0\}$ and $(\domain_0)^\bullet=\mdesirs_0$. Then  
\[
(\domain_1)^\bullet=\mdesirs_1=\left\{U \rho U^\dagger \Big| \rho \in  \mdesirs_0 \right\}.
\]
\end{Proposition}
\begin{proof}
First of all, notice that since the transformation is represented by the unitary matrix $U$, it preserves the eigenvalues, meaning that 
\begin{equation}\label{eq:unitary}
G \gneq 0 \Leftrightarrow U^\dagger G U \gneq 0.
\end{equation}
Let $\mdesirs_0$ be the credal quantum set associated to $\domain_0$.
If we verify that:
\begin{equation}\label{eq:tempodual}
 \domain_1=\left\{G \in \CH \mid G \gneq 0 \right\} \cup 
 \left\{G \in \CH \mid (U^\dagger G U)\cdot \rho >0, ~~\forall\rho \in 
 \mdesirs_0\right\},
\end{equation}
we are done. Indeed, note that 
\[
 (U^\dagger G U)\cdot \rho= Tr(U^\dagger G^\dagger U \rho)=Tr( G^\dagger U \rho U^\dagger),
\]
and observe that $Tr(U \rho U^\dagger)=Tr(\rho U^\dagger U)=Tr(\rho)=1$ and, thus, we can derive that
\begin{displaymath}
 \mdesirs_1=\left\{U \rho U^\dagger \Big| \rho \in  \mdesirs_0 \right\}.
\end{displaymath}
We exploit the relationship $\domain_1=\{G \in \CH \mid U^\dagger G U \in \domain_0\}$ 
and the duality between  $\mdesirs_0$ and $\domain_0$ to derive Equation~\eqref{eq:tempodual}:
\begin{center}
\begin{tabular}{ c c l c r }
  $G \in \domain_1$ & $\Leftrightarrow$ & $U^\dagger G U \in \domain_0$ & \; & (definition of $\domain_1$) \\
     & $\Leftrightarrow$ & $U^\dagger G U \gneq 0  \lor (U^\dagger G U)\cdot \rho >0, ~~\forall\rho \in 
 \mdesirs_0$ & \; & (definition of $(\cdot)^\circ$) \\
     & $\Leftrightarrow$ & $G \gneq 0  \lor (U^\dagger G U)\cdot \rho >0, ~~\forall\rho \in 
 \mdesirs_0$ & \; & (Equation~\eqref{eq:unitary}). \\
\end{tabular}
\end{center}

\end{proof}

 \subsection{Results in Section \ref{sec:4ndaxiom}}

\begin{Proposition}
\label{pro:marginal}
Let $\domain^{AB}$ be an SDG in $\mathbb{C}_h^{nm \times nm}$. Then: 
\begin{enumerate}
\item Both
$\marg_A(\domain^{AB})$ and $\marg_B(\domain^{AB})$ are SDGs.
\item Moreover, they are both maximal SDGs whenever $\domain^{AB}$ is a maximal SDG too.
\end{enumerate}
\end{Proposition}
\begin{proof}
We just check the case of $\marg_A(\domain^{AB})$, the remaining one is analogous.

Given two matrices $N,M$, the eigenvalues of $N \otimes M$ are equal  to the product of the eigenvalues
of $N$ and $M$.
Hence,   if $G\gneq 0$ then $G \otimes I_m \gneq 0$ and, thus, $ G \otimes I_m \in \domain^{AB}$,
since $\domain^{AB}$ is an SDG (it accepts partial gain).
Since   $ (G_1+G_2) \otimes I_m= G_1 \otimes I_m+ G_2 \otimes I_m$ 
and $(\lambda G) \otimes I_m=\lambda (G \otimes I_m)$, it follows that 
$\marg_A(\domain^{AB})$ is a cone. It is also non-pointed since $ 0 \otimes I_m=0$
and $0\notin \domain^{AB}$. Moreover, for any $G \in \marg_A(\domain^{AB})$ either
$G\gneq0 $ or $G-\Delta\in \marg_A(\domain^{AB})$ for some $\Delta>0$ (openness).
To prove the last property, consider  a $G$ that is not PSDNZ but $G\otimes I_m \in \domain^{AB}$.
If such a $G$ does not exist then the openness property holds.
If it exists, then consider $(G-\Delta)\otimes I_m=G\otimes I_m -\Delta\otimes I_m$ with $\Delta>0$.
Take the dual of $\domain^{AB}$,  $\mdesirs^{AB}$; then, since  $G\otimes I_m$ is not PSDNZ, the constraint $G\otimes I_m \in \domain^{AB}$ means 
that 
$$
G\otimes I_m \cdot \rho>0 ~~\forall \rho\in \mdesirs^{AB}.
$$
Let $\epsilon=\min_{\rho\in \mdesirs^{AB}} G\otimes I_m \cdot \rho$, then take
$\Delta=\alpha \epsilon I_{n}>0$ with $0<\alpha<1$, it follows that
$$
G\otimes I_m \cdot \rho - \Delta\otimes I_m \cdot \rho>0 \;\;\;~~~\forall~ 0<\alpha<1,
$$
which proves that $\marg_A(\domain^{AB})$ is open.

We verify that $\marg_A(\domain^{AB})$ is maximal whenever $\domain^{AB}$ is maximal.
Since by hypothesis  $\domain^{AB}$  is maximal, its dual $\mdesirs^{AB}$ includes  only  a single density matrix $\hat{\rho}^{AB}$.
Let $H=F-\epsilon I_n \in \CH$ a gamble such that   $Tr(F^\dagger\otimes I_m  \hat{\rho}^{AB}) \leq 0$ 
and so $Tr(F^\dagger \otimes I_m \hat{\rho}^{AB})-\epsilon Tr(I_n\otimes I_m \hat{\rho}^{AB})=Tr(F^\dagger \otimes I_m \hat{\rho}^{AB})-\epsilon  < 0$
for any $\epsilon>0 $. Then consider the gamble $G=-H$. $G$  is in $\marg_A(\domain^{AB})$ because  $Tr(G^\dagger\otimes I_m \hat{\rho}^{AB}) =-Tr(H^\dagger\otimes I_m \hat{\rho}^{AB})> 0$ and so $(F-\epsilon I)\otimes I_m \in \domain^{AB}$  for any  $\epsilon>0$. Since $H-H=0$, $\marg_A(\domain^{AB})$ cannot include any further gamble while being an SDG (for any  $\epsilon>0$, openness).
\end{proof}

\begin{Proposition}\label{prop:marginalstates}
 Let $\domain^{AB}$ be a maximal SDG. Assume that 
 $$
 \mdesirs^{AB}=\{\rho^{AB} \in \mathcal{D}_h^{nm\times nm}: ~\rho^{AB}=\rho^{A}\otimes \rho^{B}\}, 
 $$
 with $\rho^{A}\in \mathcal{D}_h^{n\times n}$ and $\rho^{B}\in \mathcal{D}_h^{m\times m}$.
 Then we have that  
 $$
 \marg_A(\domain^{AB})=\{G \in \CH \mid G \gneq 0\} \cup \{G \in \CH \mid G \cdot \rho^{A}>0\},
 $$ 
 so that $ \mdesirs^{A}=\{\rho^{A}\}$. A similar result holds for $B$. 
\end{Proposition}
\begin{proof}
Let us start from $\mdesirs^{AB}$, its dual is defined via~\eqref{eq:credal2} as
$$
\domain^{AB}=\{H \in \complexs_h^{nm\times nm} \mid H\gneq0\} \cup \{H \in \complexs_h^{nm\times nm} \mid H \cdot \rho^A \otimes \rho^B>0\}.
$$
By applying the definition of marginal SDG, we have that
$$
\marg_A(\domain^{AB})=\left\{G \in \CH \mid G \otimes I_m \in \domain^{AB}\right\}.
$$ 
Since $G\gneq0$ implies $ G \otimes I_m \gneq0$, it follows that $G \otimes I_m \in \domain^{AB}$ (it includes all PSDNZ matrices).
If it is not the case, then  $H=G \otimes I_m$ is in $\domain^{AB}$ provided that
$$
0<G \otimes I_m \cdot \rho^A \otimes \rho^B=Tr((G^\dagger \otimes I_m) (\rho^A \otimes \rho^B))=Tr((G^\dagger \rho^A) \otimes (I_m \otimes \rho^B))=Tr(G^\dagger \rho^A),
$$
which proves the theorem.

\end{proof}

\begin{Proposition}
\label{pro:NE}
 The natural extension $\ext_{AB}(\domain^A, \domain^B)$ always exists.
\end{Proposition}
\begin{proof}
The natural extension is defined as
$$
\ext_{AB}(\domain^A, \domain^B)= \operatorname{posi}\left(\mathcal{H}^+ \cup \{G\otimes I_m\mid G\in \domain^A\}\cup \{I_n\otimes H\mid H \in \domain^B\}\right),
$$
where $\mathcal{H}^+ $ is the the set of PSDNZ matrices in $\mathbb{C}^{nm \times nm}$ and $ \operatorname{posi}$ denotes the conic closure:
$$
 \operatorname{posi} (S)=\left\{\sum_{i=1}^k \alpha_i H_i \;\Big|\; H_i\in S, \, \alpha_i\in \mathbb{R}, \, \alpha_i> 0, i, k=1, 2, \dots\right\}. 
$$
Therefore $\ext_{AB}(\domain^A, \domain^B)$ is a convex cone that includes all PSDNZ matrices; it remains to show
that it is non-pointed or, equivalently, that it does not include any negative semi-definite matrix (NSD).

First observe that  $G,H$ cannot be NSD matrices because $\domain^A$ and $\domain^B$ are SDG.
It follows that $G \otimes I_m$ and $I_n \otimes H $ are not NSD either.
Moreover,
 $$
\begin{array}{l}
 \alpha G \otimes I_m+ \beta I_n \otimes H \\
 =\alpha (V_G \otimes V_H)(\Lambda_G\otimes I_m)(V_G \otimes 
V_H)^\dagger+\beta (V_G \otimes V_H)(I_n\otimes \Lambda_H)(V_G \otimes V_H)^\dagger\\
=\alpha (V_G \otimes V_H)(\alpha \Lambda_G\otimes I_m+\beta I_n\otimes \Lambda_H)(V_G \otimes 
V_H)^\dagger,
\end{array}
 $$
 with $\alpha,\beta>0$ and where $G=V_G  \Lambda_G V^\dagger_G$,  $H=V_H  \Lambda_H V^\dagger_H$
 $I_m=V_HV^\dagger_H$ and $I_n=V_GV^\dagger_G$.
 Note that  $\max diagonal(\alpha \Lambda_G\otimes I_m+\beta I_n\otimes \Lambda_H) >0$
 because the vector  includes all the sums of the eigenvalues of $G$ and $H$, i.e.,  $\lambda_G+\lambda_H$.
 Thus, there is at least one positive sum.
 \end{proof}

\begin{Proposition}\label{prop:independence1}
 Let $\domain^{AB}$ be a maximal SDG. Assume that 
 $$
 \mdesirs^{AB}=\{\rho^{AB} \in \mathcal{D}_h^{nm\times nm}: ~\rho^{AB}=\rho^{A}\otimes \rho^{B}\}, 
 $$
with $\rho^{A}\in \mathcal{D}_h^{n\times n}$, and $\rho^{B}\in \mathcal{D}_h^{m\times m}$
and that $ Tr(\Pi_i^B\rho_B \Pi_i^B)>0$ for any $\Pi_i^B$ and $\Pi^B$.
 Then $\domain^{AB}$ satisfies epistemic independence of $A$ and $B$. 
 \end{Proposition}
 \begin{proof}
  From duality, we derive that
 $$
 \domain^{AB}=\{H \in  \mathbb{C}_h^{nm \times nm} \mid  H\gneq0\} \cup \{H \in  \mathbb{C}_h^{nm \times nm} \mid  H \cdot 
(\rho^{A}\otimes \rho^{B})>0\}.
 $$
 Consider the first set of gambles  and observe that for any $0 \lneq G\in \CH $, 
 $H=G \otimes I_m\gneq 0$ and so $G$ is in $ \marg_A(\domain^{AB})$.
Now consider the second set of gambles. For $H=G \otimes I_m $, we have:
 $$
 (G \otimes I_m)\cdot (\rho^{A}\otimes \rho^{B})=Tr\Big((G^\dagger \rho^{A}) \otimes  \rho^{B}\Big)=Tr(G^\dagger 
\rho^{A}) Tr(\rho^{B}).
 $$
 Since $Tr(\rho^{B})>0$, we finally  have that  
 $$
 \marg_A(\domain^{AB})=\{G \in  \mathbb{C}_h^{n \times n} \mid  G\gneq0\} \cup \{G \in  \CH \mid  G \cdot \rho^{A}> 0\}.
 $$
 Now consider
   $$
 \domain_B^{AB}=\{H \in  \mathbb{C}_h^{nm \times nm} \mid  H \gneq0 \textit{ or } (I_n  \otimes \Pi^B_i) H  (I_n \otimes \Pi^B_i)\in   
\domain^{AB}\},
 $$
 which by the definition of $\domain^{AB}$  is 
 $$
\{H \in  \mathbb{C}_h^{nm \times nm} \mid   H \gneq0 \textit{ or } \big((I_n  \otimes \Pi^B_i) H  (I_n \otimes \Pi^B_i)\big)\cdot 
(\rho^{A}\otimes \rho^{B})> 0\}.
 $$
For the second inequality observe that
 $$
 \begin{array}{l}
\big((I_n  \otimes \Pi^B_i) H  (I_n \otimes \Pi^B_i)\big)\cdot (\rho^{A}\otimes \rho^{B})\\
  =Tr\Big( H^\dagger  (I_n \otimes \Pi^B_i) (\rho^{A}\otimes \rho^{B}) (I_n  \otimes \Pi^B_i) \Big)
 \end{array}
 $$
 and
  $$
\begin{array}{l}
  (I_n \otimes \Pi_i^B)(\rho_A \otimes \rho_B)(I_n \otimes \Pi_i^B)  =(\rho_A  \otimes \Pi_i^B\rho_B)(I_n \otimes 
\Pi_i^B)  =\rho_A  \otimes \Pi_i^B\rho_B \Pi_i^B
 \end{array}
 $$
 and, thus,
  $$
\begin{array}{rcl}
 \domain_B^{AB}&=&
 \{H \in  \mathbb{C}_h^{nm \times nm} \mid H\gneq0 \textit{ or }  H  \cdot (\rho_A  \otimes \Pi_i^B\rho_B \Pi_i^B)> 0\}.
 \end{array}
 $$
 For $H=G \otimes I_m $, we know that if $G\gneq0$ then  $H=G \otimes I_m \gneq 0$. 
 For the inequality, we have
 $$
G \otimes I_m  \cdot (\rho_A  \otimes \Pi_i^B\rho_B \Pi_i^B)=Tr(G^\dagger \rho^{A}) Tr(\Pi_i^B\rho_B \Pi_i^B).
 $$
Since $Tr(\Pi_i^B\rho_B \Pi_i^B)>0$ by hypothesis, we obtain that 
$$
\marg_A(\domain^{AB}_B)=\left\{G \in \CH \mid G \otimes  I_m \in \domain_B^{AB}\right\}.
$$
 \end{proof}

 \begin{Proposition}\label{prop:independence2} 
Let $ \domain^{AB}$ be an SDG satisfying independence of $A$ and $B$ and let  $\marg_A(\domain^{AB})$ and $\marg_B(\domain^{AB})$ be maximal SDGs
with associated credal quantum sets $\mdesirs_{\marg_A(\domain^{AB})}=\{\rho^{A}\}$ and $ \mdesirs_{\marg_B(\domain^{AB})}=\{\rho^{B}\}$
such that $ Tr(\Pi_i^A\rho_A \Pi_i^A)>0$ for any $\Pi_i^A$ and $\Pi^A$ and $ Tr(\Pi_i^B\rho_B \Pi_i^B)>0$ for any $\Pi_i^B$ and $\Pi^B$. Then 
\[
 \mdesirs^{AB}=\{\rho^{AB}=\rho^{A}\otimes \rho^{B}\}, 
 \]
 i.e.,  $ \domain^{AB}$ is maximal with dual density matrix $\rho^{A}\otimes \rho^{B}$.
 \end{Proposition}
  \begin{proof}
  The proof is quite tedious and, thus, we will mostly give just the  intuition behind it. Let assume there is $\rho^{AB} \neq \rho^{A}\otimes \rho^{B}$.
Since the marginal $\marg_A(\domain^{AB})$  of the composite system has density matrix $ \rho^A$, we know that 
$$
Tr(G^\dagger \rho^A)=Tr(G^\dagger \otimes I_m \rho^{AB}) ~~\forall G \in \marg_A(\domain^{AB}),
$$
 which implies that
$$
\rho^A=Tr_B(\rho^{AB}),
$$
where $Tr_B$ denotes the partial trace.
Similarly, because of epistemic independence we know that
$$
\rho^A=Tr_B((I_n \otimes \Pi_i^B)\rho^{AB}(I_n \otimes \Pi_i^B))/Tr((I_n \otimes \Pi_i^B)\rho^{AB}(I_n \otimes \Pi_i^B))
$$
for any $\Pi_i^B$. Similarly for the subsystem $B$.
This defines the system of equations
$$
\left\{\begin{array}{l}
 \rho^A=Tr_B(\rho^{AB})\\
 \rho^A=Tr_B((I_n \otimes \Pi_i^B)\rho^{AB}(I_n \otimes \Pi_i^B))/Tr((I_n \otimes \Pi_i^B)\rho^{AB}(I_n \otimes \Pi_i^B)) ~~~\forall \Pi_i^B\\
 \rho^B=Tr_A(\rho^{AB})\\
 \rho^B=Tr_A((\Pi_j^A \otimes I_m)\rho^{AB}(\Pi_j^A \otimes I_m))/Tr((\Pi_j^A \otimes I_m)\rho^{AB}(\Pi_j^A \otimes I_m))~~~\forall \Pi_j^A\\
 Tr(\rho^{AB})=Tr(\rho^{A})=Tr(\rho^{B})=1
\end{array}\right.
$$
and given the arbitrariness of $\Pi_i^B$ and $\Pi_j^B$, this allows us to always find
a worst case and so to show  that $\rho^{AB}=\rho^{A}\otimes \rho^{B}$.
  \end{proof}

\begin{Theorem}
\label{th:frechet}
Assume that  $\domain^A$ and $\domain^B$ are maximal SDGs and so they uniquely define
 $\rho^A$ and  $\rho^B$ and assume that their natural extension is separable. Then the following inequalities hold:
 \begin{equation}
   \rho^{AB} \leq  \rho^{A} \otimes I_m \;\;\text{ and }\;\; \rho^{AB} \leq  I_n \otimes \rho^{B}
 \end{equation} 
 and
 \begin{equation}
   \rho^{AB} \geq  \rho^A \otimes I_m + I_n \otimes \rho^B-I_{nm} \;\;\text{ and }\;\; \rho^{AB} \gneq 0.
   \end{equation}
 \end{Theorem}
\begin{proof}
  Note that, since 
  $$
  Tr_B(\sum_i w_i \rho_i^A \otimes \rho_i^B)=  \sum_i w_i  Tr_B(\rho_i^A \otimes \rho_i^B)=\sum_i w_i \rho_i^A Tr( 
\rho_i^B)=\sum_i w_i \rho_i^A,
  $$
  it follows that $\sum_i w_i \rho_i^A=\rho^A$ and, similarly, $\sum_i w_i \rho_i^B=\rho^B$. 
  Since $\rho_i^B \leq I_m$ (it can easily be verified using the eigenvalue decomposition) we have that
  $$
 \rho^{AB}=\sum_i w_i \rho_i^A \otimes \rho_i^B \leq \sum_i w_i \rho_i^A \otimes I_m=\rho^A \otimes I_m.
 $$
 The inequality can be proven by using the eigenvalue decomposition of $\rho_i^A$ and $\rho_i^B$ and
 $$
 \rho_i^A \otimes \rho_i^B= (V_i^A \otimes V_i^B)(\Lambda_i^A \otimes \Lambda_i^B) (V_i^A \otimes V_i^B)^\dagger;
 $$
 note that $(V_i^A \otimes V_i^B)(V_i^A \otimes V_i^B)^\dagger=I_{nm}$ and so $\Lambda_i^A \otimes \Lambda_i^B$
 are the eigenvalues of $ \rho_i^A \otimes \rho_i^B$. Now, since
 $$
 \Lambda_i^A \otimes \Lambda_i^B \leq \Lambda_i^A \otimes I_m,
 $$
 the inequality follows. The proof for $B$ is analogous. 
 For the second part,  we use the first part to derive that
 $$
 I_{nm} \geq \rho^A \otimes I_m +I_n \otimes \rho^B - \sum_i w_i \rho_i^A \otimes \rho_i^B \geq0.
 $$
In fact, since $I_n \otimes \rho^B - \sum_i w_i \rho_i^A \otimes \rho_i^B\geq 0$, we can assume that
 $$
 I_n \otimes \rho^B - \sum_i w_i \rho_i^A \otimes \rho_i^B=0,
 $$
whence it follows that $I_{nm} \geq \rho^A \otimes I_m$, which is obviously true.
 Similarly for $\rho^A \otimes I_m  - \sum_i w_i \rho_i^A \otimes \rho_i^B=0$.
Then we have that
 $$
 \sum_i w_i \rho_i^A \otimes \rho_i^B \geq \rho^A \otimes I_m +I_n \otimes \rho^B -  I_{nm}.
 $$
 Finally, $\rho^{AB}$ is by definition greater than zero.
\end{proof}


\bibliographystyle{abbrvnat}
\bibliography{qmb}

\begin{thebibliography}{53}
\providecommand{\natexlab}[1]{#1}
\providecommand{\url}[1]{\texttt{#1}}
\expandafter\ifx\csname urlstyle\endcsname\relax
  \providecommand{\doi}[1]{doi: #1}\else
  \providecommand{\doi}{doi: \begingroup \urlstyle{rm}\Url}\fi

\bibitem[Aliprantis and Tourky(2007)]{aliprantis2007cones}
C.~D. Aliprantis and R.~Tourky.
\newblock \emph{Cones and duality}, volume~84.
\newblock American Mathematical Soc., 2007.

\bibitem[Appleby(2005{\natexlab{a}})]{Appleby05a}
D.~Appleby.
\newblock Facts, values and quanta.
\newblock \emph{Foundations of Physics}, 35\penalty0 (4):\penalty0 627--668,
  2005{\natexlab{a}}.

\bibitem[Appleby(2005{\natexlab{b}})]{Appleby05b}
D.~Appleby.
\newblock Probabilities are single-case or nothing.
\newblock \emph{Optics and spectroscopy}, 99\penalty0 (3):\penalty0 447--456,
  2005{\natexlab{b}}.

\bibitem[Appleby(2007)]{appleby2007symmetric}
D.~Appleby.
\newblock Symmetric informationally complete measurements of arbitrary rank.
\newblock \emph{Optics and Spectroscopy}, 103\penalty0 (3):\penalty0 416--428,
  2007.

\bibitem[Augustin et~al.(2014)Augustin, Coolen, de~Cooman, and
  Troffaes]{augustin2014}
T.~Augustin, F.~Coolen, G.~de~Cooman, and M.~Troffaes, editors.
\newblock \emph{Introduction to Imprecise Probabilities}.
\newblock Wiley, 2014.

\bibitem[{Benavoli} et~al.(2016){Benavoli}, {Facchini}, and
  {Zaffalon}]{BenavoliTODO}
A.~{Benavoli}, A.~{Facchini}, and M.~{Zaffalon}.
\newblock {Dispersion-free probability measures are incoherent}.
\newblock \emph{ArXiv e-prints 1606.03615}, June 2016.

\bibitem[Berger(1984)]{berger1984}
J.~O. Berger.
\newblock The robust {B}ayesian viewpoint.
\newblock In J.~B. Kadane, editor, \emph{Robustness of {B}ayesian Analyses}.
  Elsevier Science, Amsterdam, 1984.

\bibitem[Berger(1994)]{ber94}
J.~O. Berger.
\newblock An overview of robust {B}ayesian analysis (with discussion).
\newblock \emph{Test}, 3\penalty0 (1):\penalty0 5--124, 1994.

\bibitem[Caves et~al.(2002)Caves, Fuchs, and Schack]{Caves02}
C.~M. Caves, C.~A. Fuchs, and R.~Schack.
\newblock Unknown quantum states: the quantum de {F}inetti representation.
\newblock \emph{Journal of Mathematical Physics}, 43\penalty0 (9):\penalty0
  4537--4559, 2002.

\bibitem[Caves et~al.(2007)Caves, Fuchs, and Schack]{Caves07}
C.~M. Caves, C.~A. Fuchs, and R.~Schack.
\newblock Subjective probability and quantum certainty.
\newblock \emph{Studies in History and Philosophy of Science Part B: Studies in
  History and Philosophy of Modern Physics}, 38\penalty0 (2):\penalty0
  255--274, 2007.

\bibitem[Cerf et~al.(1999)Cerf, Adami, and Gingrich]{cerf1999reduction}
N.~Cerf, C.~Adami, and R.~Gingrich.
\newblock Reduction criterion for separability.
\newblock \emph{Physical Review A}, 60\penalty0 (2):\penalty0 898, 1999.

\bibitem[Couso and Moral(2011)]{Couso20111034}
I.~Couso and S.~Moral.
\newblock Sets of desirable gambles: Conditioning, representation, and precise
  probabilities.
\newblock \emph{International Journal of Approximate Reasoning}, 52\penalty0
  (7):\penalty0 1034 -- 1055, 2011.

\bibitem[de~Cooman et~al.(2011)de~Cooman, Miranda, and Zaffalon]{zaffalon2011a}
G.~de~Cooman, E.~Miranda, and M.~Zaffalon.
\newblock Independent natural extension.
\newblock \emph{Artificial Intelligence}, 175:\penalty0 1911--1950, 2011.
\newblock \doi{10.1016/j.artint.2011.06.001}.

\bibitem[{d}e Finetti(1937)]{finetti1937}
B.~{d}e Finetti.
\newblock La pr\'evision: ses lois logiques, ses sources subjectives.
\newblock \emph{Annales de l'Institut Henri Poincar\'e}, 7:\penalty0 1--68,
  1937.
\newblock {E}nglish translation in \cite{kyburg1964}.

\bibitem[{d}e Finetti(1970)]{finetti1970}
B.~{d}e Finetti.
\newblock \emph{Teoria delle Probabilit\`a}.
\newblock Einaudi, Turin, 1970.

\bibitem[{d}e Finetti(1974)]{finetti1974}
B.~{d}e Finetti.
\newblock \emph{Theory of Probability: A Critical Introductory Treatment},
  volume~1.
\newblock John Wiley \& Sons, Chichester, 1974.
\newblock {E}nglish translation of \cite{finetti1970}.

\bibitem[Fuchs(2002)]{Fuchs02}
C.~A. Fuchs.
\newblock Quantum mechanics as quantum information (and only a little more).
\newblock \emph{arXiv preprint quant-ph/0205039}, 2002.

\bibitem[Fuchs(2003{\natexlab{a}})]{Fuchs01}
C.~A. Fuchs.
\newblock \emph{Notes on a Paulian Idea: Foundational, Historical, Anecdotal,
  and Forward-looking Thoughts on the Quantum: Selected Correspondence,
  1995-2001}.
\newblock V{\"a}xj{\"o} University Press, 2003{\natexlab{a}}.

\bibitem[Fuchs(2003{\natexlab{b}})]{Fuchs03}
C.~A. Fuchs.
\newblock Quantum mechanics as quantum information, mostly.
\newblock \emph{Journal of Modern Optics}, 50\penalty0 (6-7):\penalty0
  987--1023, 2003{\natexlab{b}}.

\bibitem[Fuchs and Schack(2004)]{Fuchs04}
C.~A. Fuchs and R.~Schack.
\newblock 5 unknown quantum states and operations, a {B}ayesian view.
\newblock In \emph{Quantum State Estimation}, pages 147--187. Springer, 2004.

\bibitem[Fuchs and Schack(2011)]{Fuchs&SchackII}
C.~A. Fuchs and R.~Schack.
\newblock A quantum-{B}ayesian route to quantum-state space.
\newblock \emph{Foundations of Physics}, 41\penalty0 (3):\penalty0 345--356,
  2011.

\bibitem[Fuchs and Schack(2012)]{fuchs2012}
C.~A. Fuchs and R.~Schack.
\newblock \emph{Probability in Physics}, chapter {B}ayesian conditioning, the
  reflection principle, and quantum decoherence, pages 233--247.
\newblock The Frontiers Collection. Springer, 2012.

\bibitem[Fuchs and Schack(2013)]{longPaper}
C.~A. Fuchs and R.~Schack.
\newblock Quantum-{B}ayesian coherence.
\newblock \emph{Reviews of Modern Physics}, 85\penalty0 (4):\penalty0 1693,
  2013.

\bibitem[Gleason(1957)]{gleason1957measures}
A.~M. Gleason.
\newblock Measures on the closed subspaces of a hilbert space.
\newblock \emph{Journal of mathematics and mechanics}, 6\penalty0 (6):\penalty0
  885--893, 1957.

\bibitem[Goldstein(1983)]{goldstein1983}
M.~Goldstein.
\newblock The prevision of a prevision.
\newblock \emph{Journal of the American Statistical Association}, 87:\penalty0
  817--819, 1983.

\bibitem[Hacking(1967)]{hacking1967}
I.~Hacking.
\newblock {Slightly more realistic personal probability}.
\newblock \emph{Philosophy of Science}, 34\penalty0 (4):\penalty0 311--325,
  1967.

\bibitem[Hartmann and Suppes(2010)]{Hartmann2010}
S.~Hartmann and P.~Suppes.
\newblock \emph{EPSA Philosophical Issues in the Sciences: Launch of the
  European Philosophy of Science Association}, chapter Entanglement, Upper
  Probabilities and Decoherence in Quantum Mechanics, pages 93--103.
\newblock Springer Netherlands, Dordrecht, 2010.

\bibitem[Horodecki and Horodecki(1999)]{horodecki1999reduction}
M.~Horodecki and P.~Horodecki.
\newblock Reduction criterion of separability and limits for a class of
  distillation protocols.
\newblock \emph{Physical Review A}, 59\penalty0 (6):\penalty0 4206, 1999.

\bibitem[Jaynes(2003)]{jaynes2003probability}
E.~T. Jaynes.
\newblock \emph{Probability theory: the logic of science}.
\newblock Cambridge University Press, 2003.

\bibitem[Kolmogorov(1950)]{kolmogorov1950foundations}
A.~N. Kolmogorov.
\newblock \emph{Foundations of the Theory of Probability}.
\newblock Chelsea Publishing Co., 1950.

\bibitem[Kyburg~Jr. and Smokler(1964)]{kyburg1964}
H.~E. Kyburg~Jr. and H.~E. Smokler, editors.
\newblock \emph{Studies in Subjective Probability}.
\newblock Wiley, New York, 1964.
\newblock Second edition (with new material) 1980.

\bibitem[Levi(1980)]{levi1980enterprise}
I.~Levi.
\newblock \emph{The Enterprise of Knowledge, an Essay on Knowledge, Credal
  Probability, and Chance}.
\newblock MIT Press, Cambridge, MA, 1980.

\bibitem[Mermin(2014)]{mermin2014physics}
N.~D. Mermin.
\newblock Physics: Qbism puts the scientist back into science.
\newblock \emph{Nature}, 507\penalty0 (7493):\penalty0 421--423, 2014.

\bibitem[Miranda(2008)]{miranda2008a}
E.~Miranda.
\newblock A survey of the theory of coherent lower previsions.
\newblock \emph{International Journal of Approximate Reasoning}, 48\penalty0
  (2):\penalty0 628--658, 2008.
\newblock \doi{10.1016/j.ijar.2007.12.001}.

\bibitem[Miranda and Zaffalon(2010)]{zaffalon2010e}
E.~Miranda and M.~Zaffalon.
\newblock Notes on desirability and conditional lower previsions.
\newblock \emph{Annals of Mathematics and Artificial Intelligence}, 60\penalty0
  (3--4):\penalty0 251--309, 2010.
\newblock \doi{10.1007/s10472-011-9231-4}.

\bibitem[Nielsen and Chuang(2010)]{nielsen2010quantum}
M.~A. Nielsen and I.~L. Chuang.
\newblock \emph{Quantum computation and quantum information}.
\newblock Cambridge university press, 2010.

\bibitem[Pitowsky(2003)]{pitowsky2003betting}
I.~Pitowsky.
\newblock Betting on the outcomes of measurements: a {B}ayesian theory of
  quantum probability.
\newblock \emph{Studies in History and Philosophy of Science Part B: Studies in
  History and Philosophy of Modern Physics}, 34\penalty0 (3):\penalty0
  395--414, 2003.

\bibitem[Pitowsky(2006)]{Pitowsky2006}
I.~Pitowsky.
\newblock \emph{Physical Theory and its Interpretation: Essays in Honor of
  Jeffrey Bub}, chapter Quantum Mechanics as a Theory of Probability, pages
  213--240.
\newblock Springer Netherlands, Dordrecht, 2006.

\bibitem[Renes et~al.(2003)Renes, Blume-Kohout, Scott, and
  Caves]{renes2003symmetric}
J.~M. Renes, R.~Blume-Kohout, A.~J. Scott, and C.~M. Caves.
\newblock Symmetric informationally complete quantum measurements.
\newblock \emph{arXiv preprint quant-ph/0310075}, 2003.

\bibitem[Savage(1954)]{savage1954}
L.~J. Savage.
\newblock \emph{The Foundations of Statistics}.
\newblock Wiley, New York, 1954.
\newblock First edition, second revised edition published in 1972.

\bibitem[Schack(2003)]{Schack04}
R.~Schack.
\newblock Quantum theory from four of hardy's axioms.
\newblock \emph{Foundations of Physics}, 33\penalty0 (10):\penalty0 1461--1468,
  2003.

\bibitem[Schack et~al.(2001)Schack, Brun, and Caves]{Schack01}
R.~Schack, T.~A. Brun, and C.~M. Caves.
\newblock Quantum {B}ayes rule.
\newblock \emph{Physical Review A}, 64\penalty0 (1):\penalty0 014305, 2001.

\bibitem[Schneider(1965)]{schneider1965positive}
H.~Schneider.
\newblock Positive operators and an inertia theorem.
\newblock \emph{Numerische Mathematik}, 7\penalty0 (1):\penalty0 11--17, 1965.

\bibitem[Shafer(1982)]{shafer1982}
G.~Shafer.
\newblock {B}ayes's two arguments for the rule of conditioning.
\newblock \emph{The Annals of Statistics}, 10:\penalty0 1075--1089, 1982.

\bibitem[Shafer(1985)]{shafer1985}
G.~Shafer.
\newblock Conditional probability.
\newblock \emph{International Statistical Review}, 53:\penalty0 261--277, 1985.

\bibitem[Smith(1961)]{smith1961}
C.~A.~B. Smith.
\newblock Consistency in statistical inference and decision.
\newblock \emph{Journal of the Royal Statistical Society, Series A},
  23:\penalty0 1--37, 1961.

\bibitem[Timpson(2008)]{Timpson08}
C.~G. Timpson.
\newblock Quantum {B}ayesianism: a study.
\newblock \emph{Studies in History and Philosophy of Science Part B: Studies in
  History and Philosophy of Modern Physics}, 39\penalty0 (3):\penalty0
  579--609, 2008.

\bibitem[Van~Fraassen(1984)]{fraassen1984}
B.~C. Van~Fraassen.
\newblock Belief and the will.
\newblock \emph{Journal of Philosophy}, 81:\penalty0 235--256, 1984.

\bibitem[Vicig(2000)]{Vicig2000235}
P.~Vicig.
\newblock Epistemic independence for imprecise probabilities.
\newblock \emph{International Journal of Approximate Reasoning}, 24\penalty0
  (2-3):\penalty0 235 -- 250, 2000.

\bibitem[Walley(1991)]{walley1991}
P.~Walley.
\newblock \emph{Statistical Reasoning with Imprecise Probabilities}.
\newblock Chapman and Hall, London, 1991.

\bibitem[Williams(1975)]{williams1975}
P.~M. Williams.
\newblock Notes on conditional previsions.
\newblock Technical report, School of Mathematical and Physical Science,
  University of Sussex, UK, 1975.
\newblock Reprinted in \cite{williams2007}.

\bibitem[Williams(2007)]{williams2007}
P.~M. Williams.
\newblock Notes on conditional previsions.
\newblock \emph{International Journal of Approximate Reasoning}, 44:\penalty0
  366--383, 2007.
\newblock Revised journal version of \cite{williams1975}.

\bibitem[Zaffalon and Miranda(2013)]{zaffalon2013a}
M.~Zaffalon and E.~Miranda.
\newblock Probability and time.
\newblock \emph{Artificial Intelligence}, 198\penalty0 (1):\penalty0 1--51,
  2013.

\end{thebibliography}

\end{document}